\documentclass[11pt]{article}

\usepackage{fullpage}
\usepackage{amsmath, amsthm, amssymb}
\usepackage{bbm}
\usepackage{mathtools}
\usepackage{graphicx}
\usepackage{xcolor}
\usepackage{enumerate}
\usepackage{natbib}
\usepackage{subfig}
\usepackage{multirow}
\usepackage{float}
\usepackage{tikz}
\usepackage{longtable}
\usepackage{indentfirst}
\usepackage{setspace}
\usepackage{xurl}
\RequirePackage[colorlinks = true, linkcolor = blue, citecolor=blue,urlcolor=blue]{hyperref}

\usepackage{multirow,arydshln} % hdashline in tables
\usepackage{enumitem}

\newtheorem{lemma}{Lemma}
\newtheorem{cor}{Corollary}
\newtheorem{theorem}{Theorem}

\newtheorem{prop}{Proposition}

\DeclareMathOperator*{\argmin}{argmin}

\DeclareMathOperator{\RCV}{RCV}

\DeclareMathOperator{\diam}{diam}

\DeclareMathOperator{\diag}{diag}
\DeclareMathOperator{\SPE}{SPE}
\DeclareMathOperator{\SEE}{SEE}
\DeclareMathOperator{\NSR}{NSR}

\newcommand{\SN}[1]{\textcolor{orange}{\textbf{#1}}}
\newcommand{\dd}{\,\mathrm{d}\,}    % differential in integrals
\newcommand{\mypar}[1]{\smallskip\noindent\textbf{\emph{#1}}}

\DeclareRobustCommand \triangled{\tikz{\filldraw[color =  red,fill= red](0,0.1) -- (0.1,0.3) -- (0.2,0.1) -- cycle;}}
\DeclareRobustCommand \circled{\tikz{\filldraw[color =  black,fill= black] circle (2.5pt);}}
\DeclareRobustCommand \circledd{\tikz{\filldraw[color =  yellow, fill= yellow] circle (3pt);}}

\DeclareRobustCommand\dashedd {\tikz[baseline=-0.6ex]\draw[blue,thick, dash pattern=on 6pt off 4pt] (0,0)--(0.62,0);} 
\DeclareRobustCommand\solidd {\tikz[baseline=-0.6ex]\draw[black,thick] (0,0)--(0.62,0);} 

\begin{document}

\title{Robust Functional Regression with Discretely Sampled Predictors}
\author{Ioannis Kalogridis$^1$ and Stanislav Nagy$^2$}
\date{%
    $^1$ Department of Mathematics, KU Leuven  \\%
    $^2$ Department of Probability and Mathematical Statistics, Charles University \\[2ex]%
    \today
}
\maketitle

\begin{abstract}

The functional linear model is an important extension of the classical regression model allowing for scalar responses to be modeled as functions of stochastic processes. Yet, despite the usefulness and popularity of the functional linear model in recent years, most treatments, theoretical and practical alike, suffer either from (i) lack of resistance towards the many types of anomalies one may encounter with functional data or (ii) biases resulting from the use of discretely sampled functional data instead of completely observed data. To address these deficiencies, this paper introduces and studies the first class of robust functional regression estimators for partially observed functional data. The proposed broad class of estimators is based on thin-plate splines with a novel computationally efficient quadratic penalty, is  easily implementable and enjoys good theoretical properties under weak assumptions. We show that, in the incomplete data setting, both the sample size and discretization error of the processes determine the asymptotic rate of convergence of functional regression estimators and the latter cannot be ignored. These theoretical properties remain valid even with multi-dimensional random fields acting as predictors and random smoothing parameters. The effectiveness of the proposed class of estimators in practice is demonstrated by means of a simulation study and a real-data example.

\end{abstract}
\section{Introduction}

In recent years, technological innovations and improved storage capabilities have led practitioners to observe and record increasingly complex high-dimensional data that are characterized by an underlying functional structure. Such data are nowadays commonly referred to as functional data and relevant research has been enjoying considerable popularity, following works such as \citet{Ramsay:1982}, \citet{Ramsay:1991} and \citet{Ramsay:2005}. While the field of functional data analysis (FDA) has become very broad with many specialized subpaths, see, e.g., \citet{Ferraty:2006, Horv:2012, Kokoszka:2017}, the functional linear model continues to occupy a prominent place within FDA. This model stipulates that an $\mathcal{L}^2(\mathcal{I})$-process $\{X(t), t \in \mathcal{I}\}$, for some bounded interval $\mathcal{I} \subset \mathbbm{R}$, influences the response variable $Y$ through an unknown coefficient function $\beta_0 \in \mathcal{L}^2(\mathcal{I})$, viz,
\begin{align}
\label{eq:FLM}
Y = \alpha_0 + \int_{\mathcal{I}} X(t) \beta_0(t) \dd t + \epsilon,
\end{align}
where $\alpha_0 \in \mathbbm{R}$ is an unknown offset (intercept) term and $\epsilon$ a random error, which is assumed to be independent of $X$. Typically, $\epsilon$ is also assumed to possess finite second moments, but, as we shall see, moment assumptions on the errors will not be needed for the theoretical results of this paper.

Assuming that the researcher readily possesses a random sample $(X_1, Y_1), \ldots, (X_n, Y_n)$ following model \eqref{eq:FLM}, there is a wealth of possible estimation methods that she could employ in order to estimate the unknown parameters $(\alpha_0, \beta_0) \in \mathbbm{R} \times \mathcal{L}^2(\mathcal{I})$. Since the covariance operator of the process $\{X(t), t \in \mathcal{I} \}$ does not possess a bounded inverse, such estimation methods are invariably based on lower rank representations of either $X$ or $\beta_0$, regularization through appropriate penalties or a combination thereof. We mention, in particular, the popular approaches of functional principal component regression and penalized basis expansions regression expounded by \citet{Ramsay:2005} and \citet{Kokoszka:2017} and theoretically investigated by \citet{Hall:2007} and \citet{Cardot:2003, Li:2007}, respectively. More advanced approaches include the hybrid method of \citet{Reiss:2007} which combines functional principal components and penalized spline expansions, as well as the reproducing kernel Hilbert space estimator of \citet{Yuan:2010}. To these methods we may add the sparsity-inducing proposal of \citet{James:2009}, which ensures the smoothness and interpretability of the estimates by taking advantage of the sparsity of higher order derivatives of $\beta_0$ when those exist.

From the point of view of robustness, that is, resistance against atypical observations and model misspecification, the aforementioned proposals are not satisfactory, as they are all generalizations of the ordinary least-squares estimator. The latter is known to be very susceptible even to mild deviations from the ideal model assumptions \citep[see, e.g.,][]{Huber:2009}. To overcome this weakness of least-squares based estimators, several proposals have been put forth in the literature over the years. These proposals essentially amount to replacing the square loss with a more slowly increasing loss function thereby ensuring that the influence of atypical observations and model misspeficication on the estimates is better controlled. Examples of such an approach include \citet{Qingguo:2017} and \citet{Kal:2019}, who generalized the functional principal component estimator of \citet{Hall:2007}. Later, \citet{Boente:2020, Kalogridis:2023} proposed robust spline estimators based on the principle of MM-estimation, and \citet{Shin:2016} generalized the work of \citet{Yuan:2010} by allowing for a large number of loss functions.

Nevertheless, a severe drawback of all aforementioned theoretical contributions, robust and non-robust alike, is their reliance on completely observed functional predictors $X_1, \ldots, X_n$ throughout $\mathcal{I}$. In practice, this is an unrealistic assumption and in the vast majority of cases either computational or resource constraints dictate that the curves are only recorded at a finite grid of points within $\mathcal{I}$: $t_1, \ldots, t_p$, say. Functional regression with discretely sampled data is a setting that has received much less attention despite the prevalence of partially observed functional data in practice. To the best of our knowledge, the only theoretical contributions for this setting are the works of \citet{Crambes:2009} and \citet{Kato:2012}, who proposed least-squares smoothing-spline estimation with a slightly modified penalty and quantile regression on the leading functional principal components, respectively. The estimator of \citet{Crambes:2009} is non-robust on account of the least-squares criterion employed therein whereas the $L_1$-estimator of \citet{Kato:2012}, while robust towards heavy-tailed error distributions, is inefficient in the case of clean data, e.g., under light tailed sub-Gaussian errors. Moreover, both works rely on the assumption that $p$, the number of discretization points, is large relative to $n$ in order for the discretization error to be asymptotically negligible  and it is not clear how these estimators behave for sparsely observed functional data. 

To overcome the drawbacks of existing methods either with respect to lack of robustness or reliance on densely observed functional data, this paper introduces a flexible family of penalized thin-plate spline M-estimators that can operate under both densely and sparsely observed functional data. This flexibility is accomplished through the introduction of a novel penalty functional that is inspired by norms on  multi-dimensional Sobolev spaces, which are widely used in the study of partial differential equations. Robustness towards atypical observations and model misspecification can be achieved through appropriate selection of the loss function. Possible loss functions in that respect include not only the quantile loss, but also the Huber loss, which can lead to estimators that are resistant to atypical observations and model misspecification while remaining highly efficient under light-tailed errors. 

We study the proposed class of estimators not only under the classical one-dimensional functional linear model with $\mathcal I \subset\mathbbm{R}$, but also with general $\mathcal I \subset \mathbbm{R}^d$. This extension of the functional linear model is particularly important, as it allows more general random objects, such as images, to be used as explanatory variables for the response variable $Y$. To the best of our knowledge, functional regression with general, possibly multi-dimensional, random fields has not been considered previously even for relatively simple least-squares estimators so that the present contributions is the first of its kind. Moreover, unlike the theoretical results of \citet{Crambes:2009, Kato:2012}, our theoretical results hold even with random smoothing parameters. Since smoothing parameters are normally selected in a data-dependent way, treating them as random variables leads to considerably more realistic and useful results.

The rest of the paper is structured as follows. Section~\ref{sec:2} introduces the proposed class of penalized M-estimators and establishes its existence under general conditions. Section~\ref{sec:3} is dedicated to the asymptotic study of our estimators. We obtain a useful decomposition of the frequently used mean squared error highlighting the delicate interplay between the sample size and level of discretization, which jointly determine the asymptotic rate of convergence of the proposed family of estimators. Section~\ref{sec:4} outlines an effective computational algorithm and a model selection procedure, which works well in a variety of settings as we demonstrate by means of our numerical experiments and real-data example in Section~\ref{sec:5} and Section~\ref{sec:6}, respectively. Finally, Section~\ref{sec:7} briefly discusses two interesting directions for future research.

\section{Thin-Plate Splines for Functional Regression}
\label{sec:2}
\subsection{The Proposed Family of Estimators}
Throughout this section we suppose that $X_1, \ldots, X_n$ are independent and identically distributed (i.i.d.) copies of the random field $\{X(\mathbf{t}), \mathbf{t} \in \mathcal{I}\}$ for some $\mathcal{I} \subset \mathbbm{R}^d$. As commonly done in functional regression, we also assume that $X$ is a second-order process over $\mathcal{I}$, or, equivalently, $\mathbb{E} \{\|X\|^2\}<\infty$ with $\|\cdot\|$ denoting the classical $\mathcal{L}^2(\mathcal{I})$-norm. The response variables $Y_1, \ldots, Y_n$, are assumed to be generated according to the model
\begin{align}
\label{eq:FLMI}
Y_i = \alpha_0 +  \int_{\mathcal{I}} X_i(\mathbf{t}) \beta_0(\mathbf{t}) \dd \mathbf{t} + \epsilon_i, \quad (i = 1, \ldots, n),
\end{align}
for unknown quantities $(\alpha_0, \beta_0) \in \mathbbm{R} \times \mathcal{L}^2(\mathcal{I})$ that are the objects of interest and  i.i.d. errors $\{\epsilon_i \}_{i=1}^n$.

For completely observed random fields $\{X_i\}_{i=1}^n$ most current estimation proposals would boil down to the minimization of 
\begin{align}
\label{eq:Loss}
\frac{1}{n} \sum_{i=1}^n \rho\left(Y_i-\alpha -  \int_{\mathcal{I}} X_i(\mathbf{t}) \beta(\mathbf{t}) \dd \mathbf{t} \right) + \lambda J(\beta),
\end{align}
over $(\alpha,\beta) \in \mathbbm{R} \times \Theta$, where $\Theta$ is a suitable subspace of functions, $\rho:\mathbbm{R} \to \mathbbm{R}_{+}$ represents an a-priori chosen loss function, e.g., the square loss $\rho(x) = x^2$ or the $L_1$ loss $\rho(x) =|x|$, $J: \Theta \to \mathbbm{R}_{+}$ is a penalty functional, usually a semi-norm, and $\lambda \geq 0$ a penalty parameter that regulates the trade-off between smoothness and fidelity to the data. In the one-dimensional case, that is, $d=1$, popular choices of $\Theta$ include the subspace spanned by a small number of eigenfunctions of the covariance operator of $X$ or a spline subspace, whereas $J(\beta)$ would be a function of the derivatives of $\beta$ or it could even be omitted entirely whenever $\Theta$ is a finite dimensional subspace and thus the risk of overfitting is mitigated. The interested reader is referred to \citet[Chapter~15]{Ramsay:2005} for a discussion of the most popular estimation methods for scalar-on-function regression in the one-dimensional setting.

In cases where the data is only discretely observed, at $\mathbf{t}_1, \ldots, \mathbf{t}_p \in \mathbbm{R}^d$, say, the above recipe is not applicable, as the $\mathcal{L}^2(\mathcal{I})$-inner product between the random fields and each candidate coefficient function $\beta$ cannot be computed. Instead, a popular method of estimation with discretely sampled functional data consists  of approximating the integral in \eqref{eq:Loss} with a Riemann sum and applying any of the aforementioned methods. In essence, this strategy consists of ignoring the error associated with the discrete nature of the data. However, while appealing due to its simplicity, the theoretical effects of such a practice are not entirely understood. We take a different approach in this paper by dealing with the Riemann sum directly and deriving our estimators from this Riemann sum in combination with a suitable penalty functional $J(\cdot)$ on the space of smooth functions on $\mathbbm{R}^d$. 

To explain our method in detail, let  $\{A_j\}_{j=1}^p$ denote a disjoint partition of $\mathcal{I}$ such that $\mathbf{t}_j \in A_j$ and $\mu(A_j)>0$ for each $j = 1, \dots, p$, where $\mu(A)$ denotes the Lebesgue-measure (volume) of a measurable set $A$. With this notation we may rewrite the functional linear model \eqref{eq:FLMI} as
\begin{align*}
%\label{eq:FLMA}
Y_i = \alpha_0 +  \sum_{j=1}^p X_i(\mathbf{t}_j) \beta_0(\mathbf{t}_j) \mu(A_j) + d_i + \epsilon_i, \quad (i = 1, \ldots, n),
\end{align*}
where  $\{d_i\}_{i=1}^n$ represent the unobserved discretization errors, i.e.,
\begin{align}   \label{eq:di}
d_i = \int_{\mathcal{I}} X_i(\mathbf{t}) \beta_0(\mathbf{t}) \dd \mathbf{t} - \sum_{j=1}^p X_i(\mathbf{t}_j) \beta_0(\mathbf{t}_j) \mu(A_j), \quad (i = 1,\ldots, n).
\end{align}
In order to both deal with the discretization error theoretically and allow for a great variety of coefficient functions $\beta_0$, we assume in this paper that $\beta_0$ is a smooth function on $\mathbbm{R}^d$ in the sense of possessing partial derivatives of a given order. A rich space of functions fulfilling this property is the Sobolev-Hilbert space of functions of order $m \geq 1$, $\mathcal{H}^{m}(\mathbbm{R}^d)$, defined by
\begin{align*}
\mathcal{H}^{m}(\mathbbm{R}^d) = \left\{ f:\mathbbm{R}^d \to \mathbbm{R}, \  \frac{\partial^m f(t_1, \ldots, t_d)}{\partial t_{1}^{m_1} \ldots \partial t_{d}^{m_d}} \ \text{exists for all} \ m_1 + \ldots + m_d = m \  \text{and} \  I_m^2(f) < \infty \right\},
\end{align*}
where the semi-norm $I_{m}:\mathcal{H}^{m}(\mathbbm{R}^d) \to \mathbbm{R}_{+}$ is given by
\begin{align}
\label{eq:Sob}
I_{m}^2(f) = \sum_{m_1 + \ldots + m_d = m} \binom{m}{m_1, \ldots, m_d} \int_{\mathbbm{R}} \ldots \int_{\mathbbm{R}} \left( \frac{\partial^m f(t_1, \ldots, t_d)}{\partial t_{1}^{m_1} \ldots \partial t_{d}^{m_d}}  \right)^2 \dd t_1 \ldots \dd t_d.
\end{align}
In other words, $\mathcal{H}^m(\mathbbm{R}^d)$ consists of all functions with square integrable partial derivatives of total order $m$. It can be shown \citep[see, e.g.,][Chapter 3]{Adams:2003} that $\mathcal{H}^m(\mathbbm{R}^d)$ is the completion of the space of functions with continuous partial derivatives of total order $m$ under a suitable norm and as such it affords us slightly greater generality. In what follows we place no restrictions on the coefficient function $\beta_0$, except that it is an element of the Sobolev-Hilbert space of functions, i.e., $\beta_0 \in \mathcal{H}^m(\mathbbm{R}^d)$. 

If we accept that $\beta_0 \in \mathcal{H}^m(\mathbbm{R}^d)$, it is intuitively appealing to minimize  a suitably chosen objective function over the whole space $\mathbbm{R} \times \mathcal{H}^m(\mathbbm{R}^d)$ in order to recover the unknown $(\alpha_0,\beta_0)$. To achieve both resistance against atypical observations $(X_i, Y_i)$ and high efficiency in regular data we propose to estimate $(\alpha_0, \beta_0)$ by $(\widehat\alpha_n, \widehat\beta_n)$ solving
\begin{align}
\label{eq:Est}
    \left(\widehat\alpha_n, \widehat\beta_n\right) = \argmin_{(\alpha, \beta) \in \mathbbm{R} \times \mathcal{H}^{m}\left(\mathbbm{R}^d\right) } \left[\frac{1}{n} \sum_{i=1}^n \rho \left( \frac{Y_i - \alpha - \sum_{j=1}^p X_i(\mathbf{t}_j) \beta(\mathbf{t}_j) \mu(A_j) }{\widehat\sigma_n} \right) + \lambda J_m^2(\beta)\right],
\end{align}
where $\rho:\mathbbm{R} \to \mathbbm{R}_{+}$ is a convex loss function, $\widehat\sigma_n$ is an auxiliary scale estimator for the scale of the error $\epsilon$, $\lambda$ is a penalty parameter, and the penalty functional $J_m:\mathcal{H}^m(\mathbbm{R}^d) \to \mathbbm{R}_{+}$ is given by
\begin{align}
\label{eq:Pen}
J_m^2(\beta) =  \sum_{j=1}^p \left|\beta(\mathbf{t}_j)\right|^2 \mu(A_j) + I_m^2(\beta), \quad \beta \in \mathcal{H}^m(\mathbbm{R}^d).
\end{align}
That is, our penalty functional $J_m^2$ is the sum of two semi-norms: the empirical squared semi-norm at the discretization points $\{\mathbf{t}_j\}_{j=1}^p$ and the Sobolev semi-norm $I_m^2$ from \eqref{eq:Sob}.

It is worth commenting on two particular aspects of our estimation framework, namely on the flexibility allowed with respect to the choice of the loss function $\rho$ and the slightly unconventional form of the penalty functional $J_m^2(\cdot)$. Indeed, our framework allows not only for the square loss $\rho(x)=x^2$, which is known to lead to efficient but highly sensitive estimates, and the $L_1$ loss $\rho(x) = |x|$, which leads to resistant but inefficient estimates under light-tailed errors. It also allows for a compromise between these two extremes through the use of many other loss functions, such as the logistic loss $\rho(x) = 2x + 4 \log(1+e^{-x})$ and the celebrated Huber loss given by
\begin{align}
\label{eq:huber}
    \rho_{k}(x) = \begin{cases}
        x^2 & |x| < k \\
        2k|x|-k^2 & |x| \geq k,
    \end{cases}
\end{align}
for some $k > 0$ regulating the mixing of the square and $L_1$ losses. For loss functions that are not power functions, the scale estimate, $\widehat\sigma_n$ in \eqref{eq:Est}, is necessary to ensure approximate scale equivariance of the estimates. That is, to ensure that our estimates for $(\alpha_0, \beta_0)$ do not heavily depend  on the measurement units of the response variables $\{Y_i\}_{i=1}^n$. Standardization with auxiliary scale estimates in order to achieve equivariance has been standard practice in classical (unpenalized) robust regression \citep[see, e.g.][]{Maronna:2019}. Such scale estimates can be defined without assuming the existence of any moments of the error $\epsilon$, see Section~\ref{sec:3} for an example.

The penalty functional~\eqref{eq:Pen} employed herein is a modification of the thin-plate penalty function appearing in the classical monographs of \citet{Wahba:1990} and \citet{Green:1994} where it only involves $I_m^2(\beta)$. It will be shown, however, that in our framework the addition of $\sum_{j=1}^p |\beta(\mathbf{t}_j)|^2 \mu(A_j)$ is essential in order to ensure the existence of the proposed estimators without heavy assumptions on the process $\{X(\mathbf{t}), \mathbf{t} \in \mathcal{I}\}$. We will also show that despite the inclusion of this additional term the solution to the minimization problem may still be found in the space of natural thin-plate splines, as is the case when only employing the classical thin-plate penalty functional. The practical interpretation of $J_m^2$ is similar to that of $I_m^2$: a premium is placed on functions with ``large" partial derivatives of total order $m$ but, unlike $I_m^2$, $J_m^2$ also penalizes functions assuming large values at the $\{\mathbf{t}_j\}_{j=1}^p$. Thus, $J_m^2$ enforces the smoothness of $\widehat\beta_n$ not only by penalizing roughness but also by shrinking it towards zero.

\subsection{Existence of the Estimators}

We now establish the existence of the penalized M-estimators defined in \eqref{eq:Est}. We shall require the following two general conditions involving the set $\mathcal{I}$, its boundary $\partial \mathcal{I}$ and the discretization points $\{\mathbf{t}_j\}_{j=1}^p$. We use $\left\Vert \cdot \right\Vert_{\mathbbm{R}^d}$ to denote the Euclidean norm on $\mathbbm{R}^d$.

\begin{enumerate}[label=(A\arabic*), ref=(A\arabic*)]
\item \label{A1} The discretization points $\{\mathbf{t}_j\}_{j=1}^p$ are contained in a bounded open set $\mathcal{I} \subset \mathbbm{R}^d$ whose boundary, $\partial \mathcal{I}$, satisfies the uniform cone condition of \citet[p. 83]{Adams:2003}. 
\iffalse
That is, there exists a locally finite open cover $\{U_j\}_{j=1}^\infty$ of $\partial \mathcal{I}$ and a corresponding sequence of finite cones $\{C_j\}_{j=1}^\infty$, each congruent to some fixed cone $C$, \SN{This terminology has to be introduced, or just refer to Adams without spelling these conditions out} such that
\begin{itemize}
\item[(i)] $\{ \mathbf{x} \in \mathcal{I}: \inf_{\mathbf{y} \in \partial\mathcal{I}} \| \mathbf{x}-\mathbf{y} \|_{\mathbbm{R}^d} < \delta \} \subset \bigcup_{j=1}^{\infty} U_j$ for some $\delta>0$.
\item[(ii)] $Q_j := \bigcup_{\mathbf{x} \in \mathcal{I} \cap U_j}(\mathbf{x} + C_j) \subset \mathcal{I} $ for every $j$.
\item[(iii)] For some finite $R$, every collection of $R+1$ of the sets $Q_j$ has an empty intersection.
\end{itemize}
\fi
\item \label{A2} Define the quantities
\begin{align*}
h_{\max} & = \sup_{\mathbf{t} \in \mathcal{I}} \min_{1 \leq j \leq p} \|\mathbf{t}-\mathbf{t}_j\|_{\mathbbm{R}^d} 
\\ 
h_{\min} & = \min_{j \neq k} \| \mathbf{t}_j - \mathbf{t}_k\|_{\mathbbm{R}^d}.
\end{align*}
Then, for all large $n$ and $p$, there exists a finite constant $B$ such that $h_{\max}/h_{\min} \leq B$. 
\end{enumerate}
Condition \ref{A1} is technical in nature and precludes very irregular boundaries. It is satisfied quite generally. For example, it is valid for balls and rectangles in $\mathbbm{R}^d$. Condition \ref{A2} is essentially a density condition for the discretization points; it is required that discretization points are unique and cover $\mathcal{I}$ sufficiently well. Both of these implications follow from the bound $h_{\max}/h_{\min} \leq B$. This condition was introduced by \citet{Utr:1988} for the study of thin-plate splines in the context of non-parametric regression.

Proposition~\ref{Prop:1} below establishes the existence of the estimators and provides a useful characterization that will form the basis of our computational algorithm in Section~\ref{sec:4}. Prior to the statement of Proposition~\ref{Prop:1}, we remind the reader that a thin-plate spline of order $m$ with knots at $\{\mathbf{t}_j\}_{j=1}^p$ is any function $g: \mathbbm{R}^d \to \mathbbm{R}$ of the form
\begin{align}
\label{eq:TP}
g\left(\mathbf{t}\right) = \sum_{j=1}^p \gamma_j \eta_{m,d} \left(\left\|\mathbf{t}-\mathbf{t}_j \right\|_{\mathbbm{R}^d} \right) + \sum_{k=1}^M \delta_k \phi_k\left(\mathbf{t}\right),
\end{align} 
with $\eta_{m,d}: \mathbbm{R}_{+} \to \mathbbm{R}$ given by
\begin{align*}
\eta_{m, d}(x) = \begin{cases} \frac{(-1)^{m+1+d/2}}{2^{2m-1} \pi^{d/2} (m-1)!(m-d/2)!} x^{2m-d} \log(x) & d \ \text{even} \\ 
\frac{\Gamma(d/2-m)}{2^{2m} \pi^{d/2} (m-1)!}x^{2m-d} & d \ \text{odd}, \end{cases}
\end{align*}
where $\Gamma(\cdot)$ denotes Euler's gamma function. The functions $\{\phi_{k}\}_{k=1}^M$ are (any) basis for the space of polynomials on $\mathbbm{R}^d$ of total order less than $m$, which has dimension $M = \binom{m+d-1}{d}$. Moreover, $g$ is called a natural thin-plate spline, if, in addition to \eqref{eq:TP}, $I_m(g)<\infty$ where $I_m(\cdot)$ is defined in \eqref{eq:Sob}. As we explain in Section~\ref{sec:4} below, this condition places an orthogonality restriction on the coefficients $\boldsymbol{\gamma} = (\gamma_1, \ldots, \gamma_p)$.

\begin{prop}
\label{Prop:1}
Suppose that $\rho$ is a convex loss function, $2m>d$ and \ref{A1} and \ref{A2} hold. Then, there exists a solution to \eqref{eq:Est} in $\mathbbm{R} \times \mathcal{H}^{m}(\mathbbm{R}^d)$ denoted by $(\widehat{\alpha}_n, \widehat{\beta}_n)$. Moreover, if the set $\{\mathbf{t}_j\}_{j=1}^p$ contains a $\mathcal{P}_{m}$-unisolvent set, then $\widehat{\beta}_n$ is necessarily a natural thin-plate spline of order $m$ with knots at $\{\mathbf{t}_j\}_{j=1}^p$.
\end{prop}

The condition $2m>d$ in the statement of Proposition~\ref{Prop:1} is both necessary and sufficient for $\mathcal{H}^{m}(\mathbbm{R}^d)$ to be a reproducing kernel Hilbert space of functions wherein point evaluation is well-defined and therefore this condition cannot be weakened. This condition is satisfied for all integers $m \geq 1$ when $d=1$, but for larger $d$ it may require us to increase $m$ accordingly, e.g., $m>2$ whenever $d = 2$. The additional condition that $\{\mathbf{t}_j\}_{j=1}^p$ contains a $\mathcal{P}_{m}$-unisolvent set is a general condition requiring that there exists a subset of $\{\mathbf{t}_j\}_{j=1}^p$ that allows for unique polynomial interpolation. Equivalently, the $p \times M$ matrix $\boldsymbol{\Phi}$ with elements $\phi_j(\mathbf{t}_i)$, $i = 1, \dots, p$, $j=1,\dots,M$, needs to have full column rank. In the unidimensional case $d=1$ it is easy to show that this condition is satisfied if the points $\{t_j\}_{j=1}^p$ are distinct and $p \geq m$, which is a classical condition for the existence of smoothing spline estimators \citep[see, e.g.][Theorem 2.3]{Green:1994}.  Lastly, we remark that although Proposition~\ref{Prop:1} does not establish the uniqueness of the minimizers, uniqueness can be established rather easily whenever $\rho$ is a strictly convex loss function, such as the square or logistic losses. 

\begin{cor}
\label{Cor:1}
Suppose that the condititions of Proposition~\ref{Prop:1} are satisfied and that $\rho$ is a strictly convex function. Then, the solution $(\widehat{\alpha}_n, \widehat{\beta}_n)$ to \eqref{eq:Est} in $\mathbbm{R} \times \mathcal{H}^m(\mathbbm{R}^d)$ is unique.
\end{cor}

It is interesting to relate the penalty functional $J_m$ employed herein with the penalty functional earlier used by \citet{Crambes:2009} in the specific least-squares case, i.e., $\rho(x) = x^2$. In particular, for the one-dimensional case, \citet{Crambes:2009} proposed using
\begin{align*}
\widetilde{J}_m^{2}(\beta) = \frac{1}{p} \sum_{j=1}^p (\mathcal{P}\beta)^2(t_j) + \int_{0}^1|\beta^{(m)}(t)|^2 \dd t,
\end{align*}
where $\mathcal{P}$ is the projection operator onto the space of discretized order $m$ polynomials, that is, the polynomials of order $m$ evaluated only at $\{t_j\}_{j=1}^p$. Note that for $d = 1$, $I_m^2(\beta) = \int_{\mathbbm{R}} |
\beta^{(m)}(t)|^2 \dd t$ but for $\mathcal{I} = (0,1)$, as is assumed by \citet{Crambes:2009}, it can be shown \citep[see, e.g.,][Theorem 2.3]{Green:1994} that we can reduce the range of integration from $\mathbbm{R}$ to $(0,1)$ without loss of generality. This simplification is not valid for $d>1$. For $\mathcal{I}= (0,1)$ it is easy to see that condition \ref{A1} is satisfied while assumption \ref{A2} is satisfied, e.g., if $t_j = (j-1/2)/p$ in which case $h_{\max} \lesssim p^{-1}$ and $h_{\min} = p^{-1}$. Although seemingly different, $J_m(\beta)$ and $\widetilde{J}_m(\beta)$ turn out to be equivalent semi-norms, as Proposition~\ref{Prop:2} shows. 

\begin{prop}
\label{Prop:2}
Suppose that $\mathcal{I} = (0,1)$. Then, there exist $0<c_1 \leq c_2 < \infty $ with the property that $c_1 \widetilde{J}_m(\beta) \leq J_m(\beta) \leq c_2 \widetilde{J}_m(\beta)$ for all $\beta \in \mathcal{H}^m(\mathbbm{R})$.
\end{prop}
\noindent
Hence, for $d = 1$ it does not make a difference whether one employs $J_m(\beta)$ or $\widetilde{J}_m(\beta)$, as these penalty functionals are qualitatively similar and share the same null space, which, by the proof of Proposition~\ref{Prop:1} in the Appendix, consists solely of the zero function. The advantage of our penalty, however, is that it avoids the computation of the projection onto the discretized polynomials so that the associated computational burden is reduced. For higher dimensions the computational simplicity of our penalty functional is even more appealing as in that case the projection $\mathcal P$ would have to be on a space of polynomials whose dimension grows essentially like $m^{d}/(d!)$.

\section{Asymptotic Properties}
\label{sec:3}
\subsection{Thin-Plate Splines without Scale Estimation}

To simplify our notation we shall assume in this section that the response variables are centered so that $\alpha_0 = 0$ and the object of interest is the coefficient function $\beta_0$, as is typically the case. Furthermore, we also assume for simplicity that $\mu(A_j) = \mu(\mathcal{I})/p$, i.e., $\mathcal{I}$ is partitioned into $p$ sets of equal volume.  We begin our study of the asymptotic properties of our estimators by first examining the asymptotic behavior of thin-plate spline estimators that do not require standardization with a scale $\widehat\sigma_n$. In other words we can set $\widehat\sigma_n = 1$ in \eqref{eq:Est}. As noted in Section~\ref{sec:2}, the $L_1$ estimator with $\rho(x) = |x| $ belongs to this class of estimators and so does its quantile generalization with $\rho_{\tau}(x) = 2x(\tau-I(x<0))$ for $\tau \in (0,1)$. Symmetry of the loss function is not required in our treatment, hence this important class of estimators is also covered by our theory.

Our aim is to establish the rate of convergence of our estimators as a function of the sample size and the discretization error with respect 
to the distance given by
\begin{align}
\label{eq:np norm}
    \left\|\widehat\beta_n - \beta_0 \right\|_n^2 = \frac{1}{n} \sum_{i=1}^n \left|\langle X_i, \widehat\beta_n - \beta_0 \rangle\right|^2,
\end{align}
where $\langle \cdot, \cdot \rangle$ denotes the standard $\mathcal{L}^2(\mathcal{I})$ inner product. This distance is related to the centered empirical covariance operator $\Gamma_n(\beta) = \sum_{i=1}^n X_i \langle X_i, \beta \rangle/n$ of a process $X$ through $\|\widehat\beta_n - \beta_0 \|_n^2 = \langle \Gamma_n(\widehat\beta_n - \beta_0), \widehat\beta_n - \beta_0 \rangle$. It bears the intuitive interpretation of the mean squared error resulting from using $\langle X_i, \widehat\beta_n \rangle$  to predict $\mathbbm{E}\{Y_i|X_i\} = \langle X_i, \beta_0\rangle$ for $i = 1,\ldots, n$. As a first step in our analysis, we establish a rate of convergence with respect to the discretized variant of \eqref{eq:np norm} given by
\begin{align}
\label{eq:npd norm}
    \left\|\widehat\beta_n - \beta_0 \right\|_{n,p}^2 = \frac{\mu(\mathcal I)}{np^2} \sum_{i=1}^n \left| \sum_{j=1}^p X_i(\mathbf{t}_j)\left(\widehat\beta_n(\mathbf{t}_j) - \beta_0(\mathbf{t}_j) \right)\right|^2.
\end{align}
Clearly, \eqref{eq:npd norm} may be viewed as a Riemann approximation to \eqref{eq:np norm} based on the discretization points $\{\mathbf{t}_j\}_{j=1}^p$ and partition sets $\{A_j\}_{j=1}^p$ with $\mu(A_j) = \mu(\mathcal I)/p$. The assumptions that we need for our theoretical development are as follows.

\begin{enumerate}[label=(A\arabic*), ref=(A\arabic*), start=3]
%\item \label{A3} $\{X(\mathbf{t}), \mathbf{t} \in \mathcal{I}\}$ is a mean-zero Gaussian process and there exist $C_1>0$ and $\kappa \in (0,1]$ such that $\sup_{\mathbf{t} \in \mathcal{I}}\mathbb{E} \{ \left|X(\mathbf{t}) \right|^2 \} \leq C_1$ and
%\begin{align*}
%\mathbb{E}\left\{ \left|X\left(\mathbf{t}\right) - X\left(\mathbf{s}\right) \right|^2 \right\} \leq C_1\left\|\mathbf{t}-\mathbf{s} \right\|^{2\kappa}_{\mathbbm{R}^d},
%\end{align*}
%for all $\mathbf{t},\mathbf{s} \in \mathcal{I}$.
\item \label{A3} There exist a $t \in (0,1)$ such that $\mathbb{E}\{e^{t \|X\|_{\infty}^2} \}<\infty$ and a  $C_1>0$ and $\kappa \in (0,1]$ such that
\begin{align*}
\mathbb{E}\left\{ \left|X\left(\mathbf{t}\right) - X\left(\mathbf{s}\right) \right|^2 \right\} \leq C_1\left\|\mathbf{t}-\mathbf{s} \right\|^{2\kappa}_{\mathbbm{R}^d},
\end{align*}
for all $\mathbf{t},\mathbf{s} \in \mathcal{I}$.
\item \label{A4} The loss function $\rho: \mathbbm{R} \to \mathbbm{R}_{+}$ is convex and satisfies a Lipschitz condition, i.e., there exists a $C_2>0$ such that
\begin{align*}
|\rho(x)-\rho(y)| \leq C_2|x-y|, \quad \forall (x,y) \in \mathbbm{R}^2.
\end{align*}
\item \label{A5} The function $g(t) := \mathbb{E}\{\rho(\epsilon_1+t)\}$ is uniquely minimized at $t = 0$ and is twice differentiable with a uniformly bounded second derivative. 
\item \label{A6} There exists an $\varepsilon>0$ such that for all $|t| \leq \varepsilon$, $g(t) - g(0) \geq \varepsilon \, t^2$. 

%\item \label{A6} There exist $C_2>0$ and $\kappa \in (0,1]$ such that $\mathbb{P}\left(\left\|X\right\|_{\infty} \leq C_2 \right) = 1$ and 
%\begin{align*}
%\mathbb{E}\left\{ \left|X\left(\mathbf{t}\right) - X\left(\mathbf{s}\right) \right|^2 \right\} \leq C_2\left\|\mathbf{t}-\mathbf{s} \right\|^{2\kappa}_{\mathbbm{R}^d},
%\end{align*}
%for all $\mathbf{s},\mathbf{t} \in \mathcal{I}$.
\end{enumerate}
\iffalse
Assumption \ref{A3} concerns the distribution and the sample paths of the functional covariate $X$. We require that $X$ is a light-tailed but not necessarily smooth process. Our requirement that $X$ is Gaussian can be relaxed at the cost of more complicated technical arguments. It is worth noting that the required H\"{o}lder-continuity in mean-square only implies that $X$ has a modification with almost surely continuous sample paths.  Beyond continuity, however, the criterion does not imply any higher order smoothness of the sample paths. For example, for $X$ the Wiener process on $\mathcal{I} = (0,1)$, it is well-known that $\mathbb{E}\{|X(t)|^2\} = t$ and $\mathbb{E}\{X(t) X(s) \} = \min(t,s)$ for all $t, s \in (0,1)$. Hence, \ref{A3} is satisfied with $\kappa = 1/2$, although the Wiener process is well-known to possess sample paths that are nowhere differentiable.
\fi

Assumption \ref{A3} concerns the distribution and the sample paths of the functional covariate $X$.  We require that $X$ is a light-tailed but not necessarily smooth process. The requirement of the existence of a squared exponential moment of $\|X\|_{\infty}$ is clearly satisfied by bounded processes but it is also satisfied by the vast majority of Gaussian processes \citep[see, e.g.][Theorem 2.1.2]{adler2007random}. Thus, \ref{A3} is a considerable generalization of the boundedness requirement of \citep{Boente:2020, Kalogridis:2023} for their respective robust functional regression estimators. The required H\"{o}lder-continuity in mean-square implies that $X$ has a modification with almost surely continuous sample paths.  Beyond continuity, however, the criterion does not imply any higher order smoothness of the sample paths. For example, if $X$ is the Wiener process on $\mathcal{I} = (0,1)$, it is well-known that $\mathbb{E}\{|X(t)|^2\} = t$ and $\mathbb{E}\{X(t) X(s) \} = \min(t,s)$ for all $t, s \in (0,1)$. Hence, \ref{A3} is satisfied with $\kappa = 1/2$, although the Wiener process is well-known to possess sample paths that are nowhere differentiable.

Assumption \ref{A4} is standard in robust regression \citep[see, e.g.,][]{vandeGeer:2000} and is satisfied by many loss functions including the quantile loss. It should be noted that \ref{A4} is the only assumption that is placed directly on the loss function, as \ref{A5} and \ref{A6} are in effect placed on the distribution of the errors. Specifically, in \ref{A5} we require that the expectation $\mathbb{E}\{\rho(\epsilon_1+t)\}$ is uniquely minimized at zero and is smooth as a function of $t$. The former is required for identifiability while the latter for convenience in the proofs. It is also worth noting that by requiring $\mathbb{E}\{\rho(\epsilon_1+t)\}$ to be smooth instead of $\rho$ itself to be smooth, we are effectively allowing smoothness to be traded between $\rho$ and the distribution of the error, thereby permitting even non-smooth $\rho$-functions within our theoretical framework. It is also possible, but more notationally cumbersome, to relax the i.i.d. assumption and instead consider merely independent errors in \eqref{eq:FLMI}. In this case, \ref{A5} and \ref{A6} would need to hold uniformly for all $\epsilon_i$. 

Both parts of \ref{A5} are satisfied quite generally. For instance, for $\rho(x) = |x|$, \ref{A5} is satisfied if the errors have unique median at zero. Furthermore, as for this $\rho$,
\begin{align}
\label{eq:10}
\mathbb{E} \left\{\rho\left(\epsilon_1+t\right) -\rho\left(\epsilon_1\right) \right\} = 2 \int_{0}^{-t} F(x) \dd x + t,
\end{align}
it is easy to see that the second part of \ref{A5} is satisfied if the distribution function of the error $F(x) = \mathbbm{P}(\epsilon_1 \leq x)$ is differentiable with a bounded derivative (density) $f = F^{\prime}$. This assumption is less restrictive than the corresponding assumption in \citet{Kato:2012} (see assumption (A5) there), as we do not require the density to be differentiable. Similarly to \ref{A5}, assumption \ref{A6} is satisfied quite generally. Continuing with the $L_1$ loss, it is easy to see that for all small $|t|$ expression \eqref{eq:10} becomes 
\begin{align*}
g(t) - g(0) = \mathbb{E} \left\{\rho\left(\epsilon_1+t\right) -\rho\left(\epsilon_1\right) \right\} =  f(-t)t^2 + o(t^2),
\end{align*}
so that \ref{A6} is ensured if $f$ is strictly positive in a neighborhood about zero. These observations generalize straightforwardly to quantile estimation with $\tau \in (0,1)$. 

With assumptions \ref{A1}--\ref{A6} in place we may establish a rate of convergence for the minimizer $\widehat{\beta}_n$ in \eqref{eq:Est}, whose existence is ensured by Proposition~\ref{Prop:1}. In the statement of Theorem~\ref{thm:1} below we use $\diam(A)$ to denote the diameter of a set $A \subset \mathbbm{R}^d$, that is, $\diam(A) = \sup_{\mathbf{x},\mathbf{y} \in A} \|\mathbf{x}-\mathbf{y}\|_{\mathbbm{R}^d}$. 

\begin{theorem}
\label{thm:1}
Assume that \ref{A1}--\ref{A6} hold, $2m>d+1$ and $\lambda$ satisfies
\begin{enumerate}[label=$(\roman*)$, ref=$(\roman*)$]
    \item \label{i} $\lambda  = \log^{\frac{3}{2}}(n) O_\mathbb{P}(n^{-\frac{2m}{2m+d}} +\max_{1 \leq j \leq p} \diam^{2\kappa}(A_j))$.
    \item \label{ii} $ \lambda^{-1} \log^{\frac{3}{2}}(n) \left( n^{-\frac{2m}{2m+d}} +  \max_{1 \leq j \leq p} \diam^{2\kappa}(A_j) \right) = O_{\mathbb{P}}(1)$.
    %\item \label{iii} For $D>1$, $ D \limsup_{n\to\infty}\lambda^{-1} \log(n) \max_{1 \leq j \leq p} \diam^{2 \kappa}(A_j)  <1$ with high probability.
\end{enumerate}
Then, for $\widehat{\beta}_n$ the solution to \eqref{eq:Est} with $\widehat\sigma_n = 1$, we have
\begin{align*}
\left\|\widehat{\beta}_n - \beta_0 \right\|_{n,p} = O_{\mathbb{P}} \left( \log^2(n) \left\{ n^{-\frac{m}{2m+d}} +  \max_{1 \leq j \leq p} \diam^{\kappa}(A_j) \right\}\right) \quad \text{and} \quad J_m(\widehat{\beta}_n) = O_{\mathbb{P}}(1).
\end{align*}
\end{theorem}
\noindent
The side condition of the theorem $2m>d+1$ is a strengthening of our previous condition $2m>d$, but likewise, it is not too restrictive. For instance,  for both $d=1$ and $d=2$ it is fulfilled by the popular choice $m=2$ corresponding to cubic thin-plate splines. The limit conditions on $\lambda$ are more complex than those for the smoothing parameters in \citet{Crambes:2009} and \citet{Kato:2012}, but these conditions ensure that our results remain valid even if $\lambda$ is random. For deterministic $\lambda$ conditions \ref{i} and \ref{ii} simplify to
\begin{align*}
    \lambda \asymp \log^{\frac{3}{2}}(n) \left(n^{-\frac{2m}{2m+d}} +\max_{1 \leq j \leq p} \diam^{2\kappa}(A_j)\right).
\end{align*}
As noted in the introduction, the results of \citet{Crambes:2009} and \citet{Kato:2012} do not hold if the smoothing parameters of their respective methods are random. Hence, our theoretical results are novel in this respect. Theorem~\ref{Prop:1} shows that the rate of convergence for functional regression problems with discretely sampled predictors depends on both the sample size and the grid resolution, i.e., the discretization error. In particular, except for a $\log^2(n)$-term, the rate of convergence is the optimal non-parametric rate  of convergence for $d$-dimensional data \citep{stone1982optimal} plus the Riemann approximation error, which is governed by the regularity of the sample paths of $X$. We see, in particular, that processes with $\kappa$ closer to $1$, i.e., processes with smoother sample paths, will produce better rates of convergence. It can be shown that the $\log^2(n)$-term in the result of Theorem~\ref{thm:2} does not appear if $X$ is assumed bounded, as other authors have done \citep[e.g., by][]{Cardot:2007, Boente:2020}. As noted earlier, however, our assumptions are considerably more general, because they permit both bounded and unbounded processes.

Another interesting observation that emerges from the rate of convergence presented in Theorem~\ref{thm:1} is the existence of a threshold value between the sample size, $n$, and the discretization error, $\max_{j \leq p} \diam^{\kappa}(A_j)$, that determines the leading term in the asymptotic error of the estimator, namely,  $T_{n,p} = n^{m/(2m+d)} \max_{ j \leq p} \diam^{\kappa}(A_j)$. For  $T_{n,p} < 1$, that is, for relatively fast decay of the discretization error, the asymptotic error, $\|\widehat{\beta}_n - \beta_0\|_{n,p}$, is solely determined by the sample size through $n^{-m/(2m+d)}$ whereas in the case $T_{n,p} \geq 1$ it is the discretization error, $\max_{j \leq p} \diam^{\kappa}(A_j)$, that determines up to a logarithmic factor the rate of convergence. This finding provides support for the intuitive idea that as long as the discretization grid is sufficiently dense, then the discretization error becomes negligible in the limit. However, Theorem~\ref{thm:1} also demonstrates that the discretization error cannot be ignored otherwise and indeed dominates the asymptotic error whenever the discretization grid is relatively sparse.

The dimension $d$ plays an important role in the foregoing discussion, as, for larger $d$, one needs a larger sample size relative to the discretization error in order to obtain the optimal non-parametric rate of convergence $n^{-m/(2m+d)}$. This may be viewed as a manifestation of the curse of dimensionality that is often encountered in non-parametric estimation. For $d = 1$ and equispaced points $\{t_j\}_{j=1}^p$ within $\mathcal{I} = (0,1)$, the result of Theorem~\ref{thm:1} simplifies to 
\begin{align*}
\left\|\widehat{\beta}_n - \beta_0 \right\|_{n,p} = O_{\mathbb{P}} \left( \log^2(n) \left\{ n^{-\frac{m}{2m+1}} +  p^{-\kappa}\right\}\right) \quad \text{and} \quad J_m(\widehat{\beta}_n) = O_{\mathbb{P}}(1),
\end{align*}
which, except for the $\log^2(n)$-term, is the same rate of convergence obtained by \citet{Crambes:2009} for their least-squares estimator without any tail conditions on the eigenvalues of the covariance operator, that is, with $q = 0$ in the notation of those authors. It should be noted, nevertheless, that our rate of convergence holds much more broadly, that is, for a much greater collection of loss functions as well as for higher dimensions.

Theorem~\ref{thm:1} establishes not only a rate of convergence with respect to the semi-norm $\| \cdot\|_{n,p}$ but also the boundedness in probability of the penalty functional $J_m(\widehat \beta_n)$ evaluated at $\widehat{\beta}_n$. As this penalty functional consists of two terms involving both $ \sum_{j=1}^p |\widehat\beta_n(\mathbf{t}_j)|^2 \mu(A_j)$ and the integrals of the squared partial derivatives, $I_m^2(\widehat \beta _n)$, the boundedness of $J_m(\widehat \beta_n)$ is essential in extending this rate of convergence to a rate of convergence in the popular $\| \cdot \|_n$-norm given in \eqref{eq:np norm}. Under the same assumptions as in Theorem~\ref{thm:1}, we obtain the following useful corollary.

\begin{cor}
\label{cor:2}
Suppose that the conditions of Theorem~\ref{thm:1} hold. Then, $\sup_{\mathbf{t} \in \mathcal{I}}|\widehat\beta_n(\mathbf{t})| = O_{\mathbb{P}}(1)$ and
\begin{align*}
\left\|\widehat{\beta}_n - \beta_0 \right\|_{n} = O_{\mathbb{P}} \left( \log^2(n) \left\{ n^{-\frac{m}{2m+d}} +  \max_{1 \leq j \leq p} \diam^{\kappa}(A_j) \right\}\right).
\end{align*}
\end{cor}
\noindent
Much like Theorem~\ref{thm:1}, Corollary~\ref{cor:2} reveals that unless $n^{-m/(2m+d)} > \max_{1 \leq j \leq p} \diam^{\kappa}(A_j)$ or, equivalently, $T_{n,p}<1$, the average squared error of prediction depends on both the sample size and the discretization error and the latter cannot be ignored in one's analysis. 

\subsection{Thin-Plate Splines with Scale Estimation}

We now turn to the theoretical investigation of thin-plate estimators requiring an auxiliary scale estimate, such as the estimators based on the logistic or Huber losses. We require the following assumptions in order to establish the equivalent result of Theorem~\ref{thm:1} in this setting.

\begin{enumerate} [label=(B\arabic*), ref=(B\arabic*), start=4]
\item \label{B4} The loss function $\rho: \mathbbm{R} \to \mathbbm{R}_{+}$ is convex with bounded derivative $\psi$ that satisfies the following tail condition: for every $\varepsilon_1>0$ there exists $D_{\varepsilon_1}>0$ such that 
\begin{align*}
\sup_{x\in \mathbbm{R}}\left|\psi(tx)-\psi(sx) \right| \leq D_{\varepsilon_1}|t-s|,
\end{align*}
for any $t\geq\varepsilon_1$, $s \geq \varepsilon_1$.
\item \label{B5} There exists a $\sigma_0\in(0,\infty)$ such that $\widehat{\sigma}_n \xrightarrow{\mathbb{P}} \sigma_0$.
\item \label{B6} There exists an $\varepsilon_2>0$ such that for all $\sigma \in (\sigma_0- \varepsilon_2, \sigma_0 + \varepsilon_2)$, $\mathbb{E}\{\psi(\epsilon_1/\sigma)\} = 0$ and the function $g_{\sigma}(t) := \mathbb{E}\{\psi(\epsilon_1/\sigma + t)\}$ is continuously differentiable with a  uniformly bounded derivative $g_{\sigma}^{\prime}$ and $g_{\sigma}^{\prime}(0)>0$.

\end{enumerate}

Assumption \ref{B4} is slightly more restrictive than \ref{A4}, as a bounded derivative of $\rho$ ensures that the Lipschitz condition in \ref{A4} is satisfied. The tail condition in \ref{B4} is satisfied whenever $\psi(x) = \rho^{\prime}(x)$ changes slowly in the tail. For differentiable $\psi$-functions it is easy to show that this tail condition is satisfied provided that $\sup_{x \in \mathbbm{R}}|x\psi^{\prime}(x)|<\infty$, but differentiability is not required. For instance, \ref{B4} is also satisfied for the Huber $\psi$-function derived from \eqref{eq:huber} although $\psi$ in that case is only piecewise differentiable. It is important to note that the limiting value $\sigma_0$ in \ref{B5} needs not be the standard deviation of the error, which may not exist. Lastly, \ref{B6} parallels \ref{A5} and \ref{A6} and ensures the Fisher consistency of the estimators while also allowing for smoothness to be traded between $\psi$ and $F$, the distribution function of the error. As a specific example, consider again the Huber $\rho$-function of \eqref{eq:huber}. Then, if $F$ is symmetric about $0$ with density $f$, direct calculation shows that
\begin{align*}
    g_{\sigma}^{\prime}(0) = 2F\left(\frac{k}{\sigma}\right)-1,
\end{align*}
for every $\sigma>0$. Thus, $g_{\sigma}^{\prime}(0)$ is strictly positive under symmetric $F$ and $k>0$ so that \ref{B6} is satisfied. We are now ready to state the main result of this section.

\begin{theorem}
    \label{thm:2}
    Assume that \ref{A1}--\ref{A3} and \ref{B4}--\ref{B6} hold, $2m>d+1$ and $\lambda = o_{\mathbb{P}}(1)$ in such a way that
\begin{enumerate}[label=$(\roman*)$, ref=$(\roman*)$]
    \item \label{ip} $\lambda  = \log^{\frac{5}{2}}(n) O_\mathbb{P}(n^{-\frac{2m}{2m+d}} +\max_{1 \leq j \leq p} \diam^{2\kappa}(A_j)) $.
    \item \label{iip} $ \lambda^{-1} \log^{\frac{5}{2}}(n) \left( n^{-\frac{2m}{2m+d}} +  \max_{1 \leq j \leq p} \diam^{2\kappa}(A_j) \right) = O_{\mathbb{P}}(1)$.
\end{enumerate}
Then, for $\widehat{\beta}_n$ the solution to \eqref{eq:Est}, we have
\begin{align*}
\left\|\widehat{\beta}_n - \beta_0 \right\|_{n,p} = O_{\mathbb{P}} \left( \log^3(n) \left\{ n^{-\frac{m}{2m+d}} +  \max_{1 \leq j \leq p} \diam^{\kappa}(A_j) \right\}\right) \quad \text{and} \quad J_m(\widehat{\beta}_n) = O_{\mathbb{P}}(1).
\end{align*}
\end{theorem}
\noindent
Comparing Theorem~\ref{thm:2} to Theorem~\ref{thm:1} reveals that after the addition of the random $\widehat{\sigma}_n$ to \eqref{eq:Est}, the the power on the $\log(n)$-term has been raised from two to three. This small change is a consequence of the possibly slow rate of convergence of $\widehat\sigma_n$ to $\sigma_0$ and it can be shown that it does not occur provided that $n^{1/2}(\widehat\sigma_n-\sigma_0)= O_{\mathbb{P}}(1)$. It is worth noting, nevertheless, that such high a rate is extremely difficult to achieve in inverse statistical problems, such as functional regression. By contrast, the generality of \ref{B5} allows for explicit checking for many interesting scale estimates, as we now demonstrate by means of an example.

Consider the popular M-scale estimator based on the residuals of an initial estimator not requiring standardization with a scale estimate, such as the $L_1$ estimator with $\rho(x) = |x|$. In particular, let $\{r_i\}_{i=1}^n$ denote the residuals of the initial estimator, i.e., $r_i = Y_i - \widehat\alpha_n- \sum_{j=1}^p X_i(\mathbf{t}_j)\widehat\beta_n(\mathbf{t}_j) \mu(A_j)$, and let $\chi: \mathbbm{R} \to (0,1)$ denote an even loss function. Then, the M-scale estimate $\widehat\sigma_n$ is the solution of
\begin{align*}
%\label{eq:mscale}
    \frac{1}{n}\sum_{i=1}^n \chi \left(\frac{r_i}{\widehat \sigma_n} \right) = \frac{1}{2}.
\end{align*}
Observe that for $\chi(x) = x^2$, $\widehat\sigma_n$ is $\sqrt{2}$ times the standard deviation of the $\{r_i\}_{i=1}^n$. We refer to \citet[Chapters 2 and 10]{Maronna:2019} for an extensive discussion of such estimates. Define the population value $\sigma_0$ as the solution to
\begin{align*}   %\label{eq:sigma0}
        \mathbb{E} \left\{\chi \left(\frac{\epsilon_1}{\sigma_0} \right) \right\} = \frac{1}{2}.
\end{align*}
It can be shown that $\sigma_0$ is unique and positive provided that $\chi$ is continuous and strictly increasing in all $x>0$ for which $\chi(x) <1$ (recall that $\chi$ is even). Thus, for bounded $\chi$ no moments of the error $\epsilon$ are required for $\sigma_0$ to be well-defined. Now, the residuals $\{r_i\}_{i=1}^n$ approximate $\{\epsilon_i\}_{i=1}^n$ and it may be expected that $\widehat\sigma_n \xrightarrow{\mathbb{P}} \sigma_0$ under the right conditions.  This is indeed true, as Theorem~\ref{thm:3} below shows.

\begin{theorem}
    \label{thm:3}
    Suppose that the conditions of Theorem~\ref{thm:1} hold for the initial thin-plate estimates $(\widehat\alpha_n, \widehat\beta_n)$. Suppose also that $\chi: \mathbbm{R} \to [0,1]$ is even, strictly increasing for all $x>0$ for whuch $\chi(x) <1$ and has a bounded derivative $\chi^{\prime}$. Then, $\widehat\sigma_n \xrightarrow{\mathbb{P}} \sigma_0$, i.e. condition~\ref{B5} is satisfied.
\end{theorem}
\noindent
A popular loss function fulfilling the requirements of Theorem~\ref{thm:3} is the Tukey bisquare given by $\chi(x) = \min(1-(1-x^2)^3,1)$, which we will also use in our numerical and real data examples. We close this section by remarking that the rate of convergence  of thin-plate spline estimators with auxiliary scale in the $\|\cdot\|_n$-semi-norm given in \eqref{eq:np norm} is exactly the same as the rate obtained in Theorem~\ref{thm:2} for the $\|\cdot\|_{n,p}$-norm. The proof of this follows from completely analogous arguments as those employed in the proof of Corollary~\ref{cor:2} and is thus omitted.

\section{Practical Implementation}
\label{sec:4}

By Proposition~\ref{Prop:1} the solution to \eqref{eq:Est} may be found in the space of natural thin-plate splines with knots at $\{\mathbf{t}_j\}_{j=1}^p$ and each $g$ in that space admits the representation \eqref{eq:TP}. As a basis $\left\{ \phi_j \right\}_{j=1}^M$ for the space of polynomials of total order less than $m$, we use the monomials $t_1^{m_1} \ldots t_d^{m_{d}}$ for all nonnegative integers $m_1, \ldots, m_d$ satisfying $0 \leq m_1+\ldots+m_d \leq m-1$. That space is of dimension $M = \binom{m+d-1}{d}$. The requirement of $g$ to satisfy $I_m(g) < \infty$ is equivalent to $\mathbf{\Phi}^{\top} \boldsymbol{\gamma} = \mathbf{0}$ where $\mathbf{\Phi}_{i,j}= \phi_j(\mathbf{t}_i)$, $i=1,\dots,p$, $j=1,\dots,M$ \citep[see, e.g.][Chapter 7]{Green:1994}. Moreover, for $\boldsymbol{\gamma} = (\gamma_1, \ldots, \gamma_M)$ satisfying this constraint, \citet{Green:1994} show that $I_m^2(g) = \boldsymbol{\gamma}^{\top} \mathbf{\Omega} \boldsymbol{\gamma}$ with $\mathbf{\Omega}_{i,j} = \eta_{m,d}(\|\mathbf{t}_i - \mathbf{t}_j \|)$, $i, j = 1,\dots,p$. In order to automatically incorporate the constraint into the problem, we write $\boldsymbol{\gamma} = \mathbf{Q} \boldsymbol{\xi}$ for a $p \times (p-M)$ matrix  $\mathbf{Q}$ whose columns span the null space of $\boldsymbol{\Phi}^{\top}$ and $\boldsymbol{\xi} \in \mathbbm{R}^{p-M}$.

Denoting the matrix of the discretized processes by $\mathbf{X}$, i.e., $\mathbf{X}_{i,j} = \mathbf{X}_i(\mathbf{t}_j)$, $i=1,\dots,n$, $j=1,\dots,p$, and combining the above facts, we may deduce that to obtain $(\widehat\alpha_n, \widehat\beta_n)$ in~\eqref{eq:Est}, it suffices to minimize
\begin{align}
\label{eq:Mest}
    L_n\left(\boldsymbol{\theta}\right)  = \frac{1}{n} \sum_{i=1}^n \rho \left( \frac{Y_i - \mathbf{Z}_i^{\top} \boldsymbol{\theta} }{\widehat{\sigma}_n} \right) + \lambda \boldsymbol{\theta}^{\top} \mathbf{H} \boldsymbol{\theta},
\end{align}
with $\mathbf{Z} = [\mathbf{1}, \mathbf{X} \mathbf{W} \mathbf{\Omega} \mathbf{Q}, \mathbf{X} \mathbf{W} \boldsymbol{\Phi}]$, $\mathbf{Z}_i$ its $i$th row, $\boldsymbol{\theta} = (\alpha, \boldsymbol{\xi}^{\top}, \boldsymbol{\delta}^{\top})^{\top}$, and the $(p+1) \times (p+1)$ penalty matrix $\mathbf{H}$ given by
\begin{align*}
    \mathbf{H} = \begin{bmatrix}
        0 & 0 & 0 \\
        0 & \boldsymbol{Q}^{\top} \boldsymbol{\Omega} \mathbf{W} \boldsymbol{\Omega} \boldsymbol{Q} + \mathbf{Q}^{\top} \boldsymbol{\Omega} \mathbf{Q} & \mathbf{Q}^{\top} \boldsymbol{\Omega} \mathbf{W} \boldsymbol{\Phi} 
        \\ 0 & \boldsymbol{\Phi}^{\top} \mathbf{W} \boldsymbol{\Omega} \mathbf{Q} & \boldsymbol{\Phi}^{\top} \mathbf{W} \boldsymbol{\Phi}
    \end{bmatrix}.
\end{align*}
Here, $\mathbf{W}$ is a $p \times p$ diagonal matrix with $\mu(A_1), \ldots, \mu(A_p)$ in its diagonal. Thus, after suitable simplification, the objective function in \eqref{eq:Est} may be reduced to the objective function of a finite-dimensional penalized M-estimator, as given in \eqref{eq:Mest}.

Except for the least-squares case $\rho(x) = x^2$, the minimizer of \eqref{eq:Mest} cannot be obtained in closed form. Nevertheless, the minimizer can be identified through the penalized variant of the well-known and efficient iteratively reweighted least-squares (IRLS) algorithm, see \citet[Chapter 5]{Maronna:2019}. The algorithm amounts to a penalized weighted least-squares regression with weights at the $(k+1)$th step given by $w_i = \rho^{\prime}(r_i^k/\widehat\sigma_n)/(2 \widehat\sigma_n r_i^k)$  where $\rho'$ is the derivative of $\rho$ and $\left(r_1^k, \dots, r_n^k\right)^\top = \mathbf{r}^k$ denote the residuals from the $k$th step of the algorithm, i.e., $\mathbf{r}^k = \mathbf{Y} - \mathbf{Z}\boldsymbol{\theta}^k$, for $\boldsymbol{\theta}^k$ the value of $\boldsymbol{\theta}$ in the $k$th step of IRLS. The IRLS algorithm has the remarkable property of monotonically decreasing the objective function until convergence to the minimizer is reached, which for convex loss functions is guaranteed irrespective of the starting point \citep[Chapter 9]{Maronna:2019}. The IRLS algorithm is not directly applicable for the quantile loss function $\rho_{\tau}(x) = 2x(\tau - I(x<0))$, as $\rho_{\tau}$ is not differentiable at zero. In this case, we replace $\rho_{\tau}$ with its smooth approximation
\begin{align*}
    \widetilde{\rho}_{\tau}(x) = \begin{cases}
        \rho_{\tau}(x) & |x| \geq \varepsilon
        \\ 2\tau x^2/\varepsilon & 0 \leq x < \varepsilon
        \\  2(1-\tau)x^2/\varepsilon & -\varepsilon< x  \leq 0,
    \end{cases}
\end{align*}
for some small $\varepsilon<0$. The IRLS algorithm is applicable for $\widetilde{\rho}_{\tau}$ and, in our experience, converges fast to the minimizer of \eqref{eq:Mest}.

In order to select the penalty parameter $\lambda$ from the data we rely on the minimization of a robust scale of an approximation of the leave one out residuals. Specifically, let $\mathbf{r}_{-}(\lambda) = (r_{-1}(\lambda), \ldots, r_{-n}(\lambda))$ denote the leave-one-out residuals as obtained with $\lambda$ as the smoothing parameter. For the least-squares thin-plate estimator, $\mathbf{r}_{-}(\lambda)$ may be obtained explicitly using the ``leaving-out-one" lemma of \citep[Chapter 4]{Wahba:1990}:
\begin{align*}
    r_{-i}(\lambda) =  \frac{r_i(\lambda)}{1-\diag\left(\mathbf{H}(\lambda) \right)_i} , \quad( i =1  , \ldots, n),
\end{align*}
where $\{r_{i}(\lambda)\}_{i=1}^n$ denote the residuals and $\mathbf{H}(\lambda)$ the ``hat"-matrix, i.e., $\mathbf{H}(\lambda) = \mathbf{Z} \left( \mathbf{Z}^{\top}\mathbf{Z} + \lambda \mathbf{H} \right)^{-1} \mathbf{Z}^{\top}$ in our notation. For general thin-plate spline estimators a good approximation of $\mathbf{r}_{-}(\lambda)$  may be obtained from the last step of the IRLS algorithm, viz,
\begin{align*}
    r_{-i}(\lambda)  \approx  \frac{r_i(\lambda)}{1-\diag(\mathbf{H}_{\mathbf{W}}(\lambda))_i}, \quad( i =1  , \ldots, n),
\end{align*}
where now $\mathbf{H}_{\mathbf{W}}(\lambda)$ is the weighted hat-matrix with weighting matrix $\mathbf{W} = \diag(w_1, \ldots, w_n)$, i.e., $
    \mathbf{H}_\mathbf{W}(\lambda) = \mathbf{Z} \left( \mathbf{Z}^{\top} \mathbf{W} \mathbf{Z} + \lambda \mathbf{H} \right)^{-1} \mathbf{Z}^{\top} \mathbf{W}$.

With either the exact or approximate leave-one-out residuals, we propose to select the value of $\lambda$ minimizing
\begin{align*}
    \RCV(\lambda) = \tau(\mathbf{r}_{-}(\lambda))^2,
\end{align*}
where $\tau$ is the robust and efficient $\tau$-scale introduced by \citet{yohai1988high}, which may be loosely interpreted as a weighted standard deviation. This criterion has also been used by \citet{maronna2013robust} and is a robustification of the classical ordinary cross validation criterion, which is based on the mean of the squared leave-one-out residuals. Implementations and illustrative examples of the least-squares, $L_1$, Huber and logistic thin-plate spline estimators are provided in the \texttt{R}-package \texttt{RobustSpline} \citep{RobustSpline} accompanying this paper.

\section{Finite-Sample Performance}
\label{sec:5}

\subsection{Numerical Experiments}

In our numerical experiments we are interested in the sensitivity of the proposed family of estimators to the grid size $p$, the noise-to-signal ratio as well as the distributions of the functional covariate $X$ and the error $\epsilon$. The thin-plate spline estimators that we consider for this simulation study are:
\begin{itemize}
    \item The least-squares estimator with $\rho(x)=x^2$ abbreviated as \textsf{S}.
    \item The $L_1$ estimator with $\rho(x) = |x|$ abbreviated as \textsf{A}.
    \item The Huber estimator with $\rho(x)$ given in \eqref{eq:huber} and abbreviated as \textsf{H}.
    \item  The logistic estimator with $\rho(x) = 2x + 4\log(1+e^{-x})$ abbreviated as \textsf{L}.
\end{itemize}
For the Huber estimator in~\eqref{eq:huber} we use $ k= 1.345$, which ensures $95\%$ efficiency in the Gaussian location model \citep[see, e.g.,][Chapter 2]{Maronna:2019}. In order to cut down on computing times, in all the simulation settings we calculate the scale $\widehat\sigma_n$ for the Huber and logistic estimators from the residuals of initial undersmoothed $L_1$ estimator instead of a properly smoothed $L_1$ estimator that would be computationally more intensive. As we shall see, this simplification drastically reduces the computation times while not affecting the performance of the Huber and logistic estimators. We consider one-dimensional functional data in this study and present an example of two-dimensional functional data in the following section. For all estimators considered we use $m = 2$ corresponding to the popular cubic thin-plate splines. 

\begin{figure}[H]
    \centering
    \subfloat{\includegraphics[width = 0.496\textwidth]{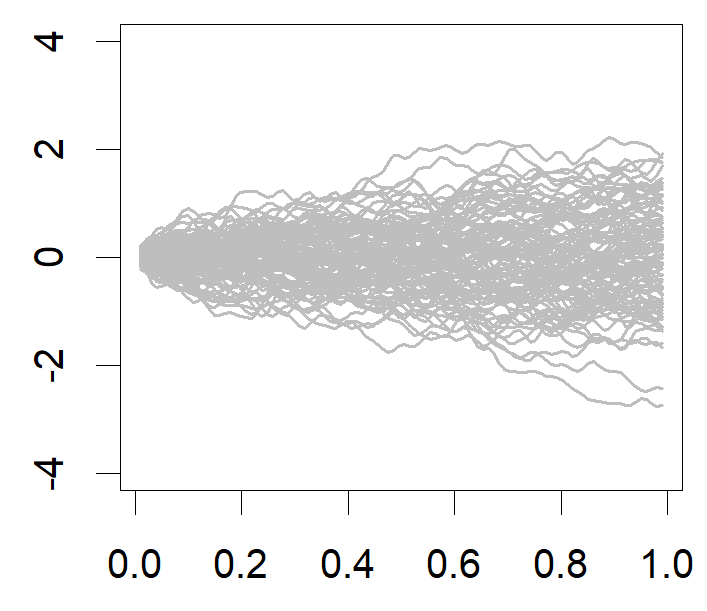}} \ \subfloat{\includegraphics[width = 0.496\textwidth]{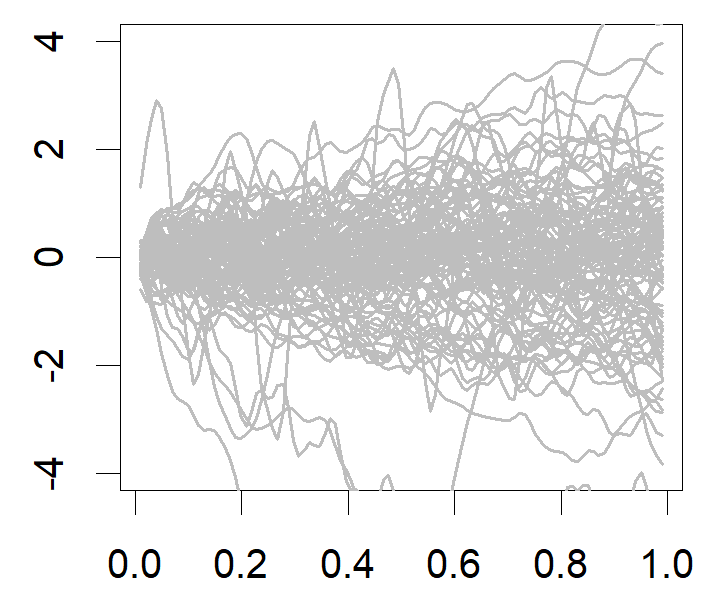}}
    \caption{$100$ curves generated with standard Gaussian and $t_2$-distributed $\{W_{ij}\}_{i, j=1}^{100,50}$ on the left and right panels, respectively. }
    \label{fig:curves}
\end{figure}

We generate the functional covariates $\{X_i\}_{i=1}^n$ defined on $\mathcal{I} = (0,1)$ according to the truncated Karhunen-Lo\`eve decomposition
\begin{align*}
%\label{eq:Karh}
    X_i(t) = \sqrt{2} \sum_{j=1}^{50} W_{ij} \frac{\sin \left( \left(j-1/2 \right) \pi t \right)}{\left(j-1/2\right)\pi t}, \quad (i = 1, \ldots, n),
\end{align*}
where $\{W_{ij}\}_{i,j=1}^{n,50}$ are independent mean-zero random variables.  The response variables $\{Y_i\}_{i=1}^n$ are subsequently generated from
\begin{align*}
    Y_i = \frac{1}{100}\sum_{k = 1}^{100} X_i(t_k) \beta_0(t_k) + \sigma \epsilon_i, \quad (i=1, \ldots, n),
\end{align*}
where $\{t_k\}_{k=1}^{100}$ are equidistant within $(0,1)$,  $\beta_0(t) = -\sin(5\,t/1.2)/0.5-1$ is the coefficient function, $\{\epsilon_i\}_{i=1}^n$ are random errors and $\sigma>0$ a constant that is chosen so as to achieve a given noise-to-signal-ratio (NSR). The NSR that we consider herein are $\{0.1, 0.2\}$ corresponding to moderate and high levels of noise in the data. We have also considered other coefficient functions but found similar results so that we only report the results for this $\beta_0$ here.

We test the estimators in the following three scenarios involving different distributions of $\{W_{ij}\}_{i,j=1}^{n,50}$ and $\{\epsilon_i\}_{i=1}^n$. The quantities $W_{ij}$ and $\epsilon_i$ are always independent.
\begin{enumerate}
[label = \textbf{Model \arabic*}, ref= Model \arabic*, start=1]
    \item $\{W_{ij}\}_{i,j=1}^{n,50}$ and $\{\epsilon_i\}_{i=1}^n$ follow standard Gaussian distributions.
    \item $\{W_{ij}\}_{i, j=1}^{n,50}$ follow a standard Gaussian distribution and $\{\epsilon_i\}_{i=1}^n$ follow a t-distribution with two degrees of freedom.
    \item $\{W_{ij}\}_{i,j=1}^{n,50}$ follow a $t$-distribution with two degrees of freedom and $\{\epsilon_i\}_{i=1}^n$ follow the standard Gaussian distribution.
\end{enumerate}
The above scenarios reflect the settings of regular data (\textbf{Model~1}) and data contaminated with vertical (\textbf{Model~2}) and leverage (\textbf{Model~3}) outliers, respectively. For a better appreciation of the effect of the distribution of the scores $\{W_{ij}\}_{i,j=1}^{n,50}$ on the shapes of the curves, Figure~\ref{fig:curves} plots 100 representative curves under \textbf{Model~1} and \textbf{Model~3} on the left and right panels, respectively. It may be seen that switching from the Gaussian to a t-distribution leads to several types of functional outliers among the curves, e.g., shape or amplitude outliers, which in turn complicate the estimation of $\beta_0$.

% latex table generated in R 4.3.1 by xtable 1.8-4 package
% Thu Aug 17 18:20:04 2023
\begin{table}[htpb]
\centering
\caption{Means and standard errors (in brackets)
                          of the prediction/estimation errors in \textbf{Models~1}--\textbf{3} based on 1000 replications with $n= 300$.}                 
\label{tab:res}
\resizebox{\textwidth}{!}{\begin{tabular}{ccc|ccc|ccc}
%\hline
& & & \multicolumn{3}{c}{$\NSR=0.1$} & \multicolumn{3}{c}{$\NSR=0.2$} \\
& & & $p=25$ & $p=50$ & $p=100$ & $p=25$ & $p=50$ & $p=100$ \\
\hline
\parbox[t]{2mm}{\multirow{8}{*}{\rotatebox[origin=c]{90}{\textbf{Model~1}}}} &\parbox[t]{2mm}{\multirow{4}{*}{\rotatebox[origin=c]{90}{$\SPE$}}} & \textsf{S} & 0.202  \scriptsize{(0.004)}  & 0.045  \scriptsize{(0.001)}  & 0.018  \scriptsize{(0.000)}  & 0.258  \scriptsize{(0.005)}  & 0.091  \scriptsize{(0.001)}  & 0.068  \scriptsize{(0.002)}  \\ 
& & \textsf{A} & 0.211  \scriptsize{(0.004)}  & 0.052  \scriptsize{(0.001)}  & 0.026  \scriptsize{(0.001)}  & 0.281  \scriptsize{(0.005)}  & 0.113  \scriptsize{(0.002)}  & 0.090  \scriptsize{(0.002)}  \\ 
& & \textsf{H} & 0.203  \scriptsize{(0.004)}  & 0.046  \scriptsize{(0.001)}  & 0.019  \scriptsize{(0.000)}  & 0.260  \scriptsize{(0.005)}  & 0.093  \scriptsize{(0.002)}  & 0.070  \scriptsize{(0.002)}  \\ 
& & \textsf{L} & 0.202  \scriptsize{(0.004)}  & 0.045  \scriptsize{(0.001)}  & 0.019  \scriptsize{(0.000)}  & 0.260  \scriptsize{(0.005)}  & 0.092  \scriptsize{(0.002)}  & 0.070  \scriptsize{(0.002)}  \\   \cdashline{3-9}
&  \parbox[t]{2mm}{\multirow{4}{*}{\rotatebox[origin=c]{90}{$\SEE$}}} & \textsf{S} & 4.652  \scriptsize{(0.247)}  & 3.414  \scriptsize{(0.217)}  & 3.107  \scriptsize{(0.200)}  & 11.451  \scriptsize{(0.948)} & 10.194  \scriptsize{(0.886)} & 11.312  \scriptsize{(0.834)} \\ 
& & \textsf{A} & 4.910  \scriptsize{(0.275)}  & 3.361  \scriptsize{(0.238)}  & 3.403  \scriptsize{(0.238)}  & 11.724  \scriptsize{(0.863)} & 10.016  \scriptsize{(0.972)} & 10.332  \scriptsize{(0.798)} \\ 
& & \textsf{H} & 4.525  \scriptsize{(0.226)}  & 3.535  \scriptsize{(0.235)}  & 3.280  \scriptsize{(0.217)}  & 11.659  \scriptsize{(0.962)} & 10.265  \scriptsize{(0.818)} & 11.430  \scriptsize{(0.843)} \\ 
& & \textsf{L} & 4.533  \scriptsize{(0.237)}  & 3.553  \scriptsize{(0.233)}  & 3.313  \scriptsize{(0.220)}  & 11.992  \scriptsize{(0.960)} & 10.358  \scriptsize{(0.886)} & 12.146  \scriptsize{(0.923)} \\ 
\hline
\parbox[t]{2mm}{\multirow{8}{*}{\rotatebox[origin=c]{90}{\textbf{Model~2}}}} & \parbox[t]{2mm}{\multirow{4}{*}{\rotatebox[origin=c]{90}{$\SPE$}}} & \textsf{S} & 0.542  \scriptsize{(0.118)}   & 0.591  \scriptsize{(0.404)}   & 0.210  \scriptsize{(0.035)}   & 0.824  \scriptsize{(0.036)}   & 0.709  \scriptsize{(0.045)}   & 0.610  \scriptsize{(0.036)}   \\ 
& & \textsf{A} & 0.232  \scriptsize{(0.004)}   & 0.069  \scriptsize{(0.001)}   & 0.043  \scriptsize{(0.001)}   & 0.362  \scriptsize{(0.006)}   & 0.193  \scriptsize{(0.004)}   & 0.162  \scriptsize{(0.004)}   \\ 
& & \textsf{H} & 0.226  \scriptsize{(0.004)}   & 0.067  \scriptsize{(0.001)}   & 0.041  \scriptsize{(0.001)}   & 0.347  \scriptsize{(0.006)}   & 0.183  \scriptsize{(0.004)}   & 0.152  \scriptsize{(0.003)}   \\ 
& & \textsf{L} & 0.229  \scriptsize{(0.004)}   & 0.070  \scriptsize{(0.001)}   & 0.044  \scriptsize{(0.001)}   & 0.359  \scriptsize{(0.006)}   & 0.195  \scriptsize{(0.004)}   & 0.162  \scriptsize{(0.003)}   \\   \cdashline{3-9}
&  \parbox[t]{2mm}{\multirow{4}{*}{\rotatebox[origin=c]{90}{$\SEE$}}} & \textsf{S} & 26.266  \scriptsize{(2.830)}  & 20.125  \scriptsize{(2.071)}  & 23.092  \scriptsize{(2.132)}  & 88.298  \scriptsize{(13.230)} & 88.913  \scriptsize{(10.978)} & 78.883  \scriptsize{(12.039)} \\ 
& & \textsf{A} & 13.615  \scriptsize{(0.829)}  & 10.327  \scriptsize{(0.540)}  & 9.963  \scriptsize{(0.496)}   & 29.432  \scriptsize{(1.601)}  & 35.711  \scriptsize{(2.221)}  & 35.942  \scriptsize{(2.173)}  \\ 
& & \textsf{H} & 10.587  \scriptsize{(0.753)}  & 8.601  \scriptsize{(0.448)}   & 8.327  \scriptsize{(0.477)}   & 25.757  \scriptsize{(1.742)}  & 29.599  \scriptsize{(1.924)}  & 28.171  \scriptsize{(1.771)}  \\ 
& & \textsf{L} & 10.097  \scriptsize{(0.538)}  & 8.603  \scriptsize{(0.457)}   & 8.128  \scriptsize{(0.474)}   & 25.784  \scriptsize{(1.603)}  & 30.830  \scriptsize{(2.021)}  & 27.636  \scriptsize{(1.705)}  \\ 
\hline
\parbox[t]{2mm}{\multirow{8}{*}{\rotatebox[origin=c]{90}{\textbf{Model~3}}}} & \parbox[t]{2mm}{\multirow{4}{*}{\rotatebox[origin=c]{90}{$\SPE$}}} & \textsf{S} & 1.727  \scriptsize{(0.047)}  & 0.520  \scriptsize{(0.068)}  & 0.191  \scriptsize{(0.014)}  & 2.638  \scriptsize{(0.104)}  & 1.911  \scriptsize{(0.726)}  & 0.746  \scriptsize{(0.046)}  \\ 
& & \textsf{A} & 2.699  \scriptsize{(0.239)}  & 0.679  \scriptsize{(0.082)}  & 0.238  \scriptsize{(0.013)}  & 3.352  \scriptsize{(0.190)}  & 2.154  \scriptsize{(0.670)}  & 0.933  \scriptsize{(0.056)}  \\ 
& & \textsf{H} & 2.418  \scriptsize{(0.201)}  & 0.564  \scriptsize{(0.070)}  & 0.195  \scriptsize{(0.014)}  & 2.816  \scriptsize{(0.111)}  & 1.931  \scriptsize{(0.733)}  & 0.769  \scriptsize{(0.049)}  \\ 
& & \textsf{L} & 2.395  \scriptsize{(0.200)}  & 0.556  \scriptsize{(0.069)}  & 0.193  \scriptsize{(0.014)}  & 2.777  \scriptsize{(0.107)}  & 1.935  \scriptsize{(0.733)}  & 0.754  \scriptsize{(0.047)}  \\   \cdashline{3-9}
&  \parbox[t]{2mm}{\multirow{4}{*}{\rotatebox[origin=c]{90}{$\SEE$}}} & \textsf{S} & 32.467  \scriptsize{(1.898)} & 8.766  \scriptsize{(0.685)}  & 4.168  \scriptsize{(0.473)}  & 44.604  \scriptsize{(3.095)} & 16.257  \scriptsize{(1.326)} & 16.348  \scriptsize{(1.967)} \\ 
& & \textsf{A} & 10.076  \scriptsize{(0.780)} & 5.884  \scriptsize{(0.627)}  & 3.617  \scriptsize{(0.274)}  & 16.881  \scriptsize{(1.386)} & 10.832  \scriptsize{(0.933)} & 13.711  \scriptsize{(1.629)} \\ 
& & \textsf{H} & 12.434  \scriptsize{(0.806)} & 5.774  \scriptsize{(0.488)}  & 4.349  \scriptsize{(0.569)}  & 24.784  \scriptsize{(1.759)} & 12.262  \scriptsize{(0.962)} & 16.448  \scriptsize{(2.006)} \\ 
& & \textsf{L} & 12.631  \scriptsize{(0.933)} & 5.942  \scriptsize{(0.509)}  & 4.227  \scriptsize{(0.516)}  & 23.292  \scriptsize{(1.600)} & 12.099  \scriptsize{(0.921)} & 15.733  \scriptsize{(1.671)} \\ 
%\hline
\end{tabular}}
\end{table}

In order to test our estimators in the incomplete data setting, in each simulation run we randomly choose $p$ points $0<t_{k_1}<\ldots<t_{k_p} <1$ from $\{t_{k} \}_{k=1}^{100}$ and we compute the thin-plate spline estimators of $\beta_0$ from the incomplete data $(X_i(t_{k_1}), \ldots, X_i(t_{k_p}), Y_i)$, $i =1, \ldots, n$.  We consider values of $p$ in $\{25, 50 ,100\}$ with $p = 100$ corresponding to fully observed data. To assess the performance of the competing estimators we rely on the squared prediction and estimation error, SPE and SEE, respectively, given by
\begin{align*}
\SPE = \frac{1}{n} \sum_{i=1}^n \left| \frac{1}{100} \sum_{j=1}^{100} X_i(t_k) \beta_0(t_k)  -\widehat{\alpha}_n - \sum_{j=2}^p X_i(t_{k_j})\widehat{\beta}_n(t_{k_j}) (t_{k_j}-t_{k_{j-1}}) \right|^2,
\end{align*}
and
\begin{align*}
\SEE =  \sum_{j=2}^p \left|\beta_0(t_{k_j})- \widehat\beta_n(t_{k_j}) \right|^2 \left(t_{k_j}-t_{k_{j-1}}\right). 
\end{align*}
Table~\ref{tab:res} presents the average SPEs and SEEs and their standard errors for the four competing estimators based on $1000$ replications with $n = 300$.

There are several interesting observations emerging from Table~\ref{tab:res}. The most notable is the rapid deterioration of the performance of the least-squares estimator upon deviation from the ideal model conditions. Indeed, while under \textbf{Model~1} the least-squares estimator marginally outperforms its competitors, under the heavy-tailed $t_2$-distributed errors in \textbf{Model~2} its lead quickly evaporates and the estimator ends up widely outperformed by all other estimators. It is quite remarkable that these robust estimators maintain a relatively stable performance under \textbf{Model~1} and \textbf{Model~2} with respect to prediction. Their performance with respect to estimation deteriorates under \textbf{Model~2} but clearly not to the same extent as the performance of the least-squares estimator. The Huber and logistic estimators perform quite comparably in both situations and prove to be more efficient than the $L_1$ estimator under light-tailed errors. To illustrate the differences in performance between the least-squares and robust estimators, Figure~\ref{fig:functions1} presents the least-squares and Huber estimates for $\beta_0$ under \textbf{Model~1} and \textbf{Model~2} on the top and bottom rows respectively. These plots show that while the Huber estimates remain stable under contamination, the least-squares estimates can become erratic leading to a noticeable deterioration in performance.                                                                             

It is interesting to note that all estimators appear vulnerable to contamination in the predictor space as given in \textbf{Model~3}, although not to the same extent. The fact that estimators based on convex $\rho$-functions are vulnerable to this type of contamination is well-known in the robustness literature, see \citet[Chapter 4]{Maronna:2019}. No estimator convincingly outperforms the others with respect to prediction in this model, but the robust estimators perform significantly better than the least-squares estimator with respect to estimation. Among the robust estimators, the $L_1$ estimator offers the most protection against outlying observations in the predictor space albeit not by a large margin. However, since the $L_1$-estimator can perform considerably worse than the Huber and logistic estimators under regular data, its use is, in our opinion, warranted only in cases when one suspects heavy contamination within the data. 

It is also worth noting that higher level of noise and larger $p$ affect all estimators in the same way. In particular, higher NSR makes both estimation and prediction harder for all estimators. The effect of a larger $p$, that is, more completely observed functional data, is overall positive and particularly noticeable for moderate NSR. For high NSR, the situation is less clear-cut and a larger $p$ in this setting may even lead to deterioration in the performance of the estimators. We attribute this difference to the fact that for a larger $p$ more coefficients in \eqref{eq:TP} need to be estimated from the data. If the level of noise is high, these coefficients cannot be accurately estimated and as a result the performance of the estimators need not necessarily improve. 

\begin{figure}[H]
    \centering
    \subfloat{\includegraphics[width = 0.496\textwidth]{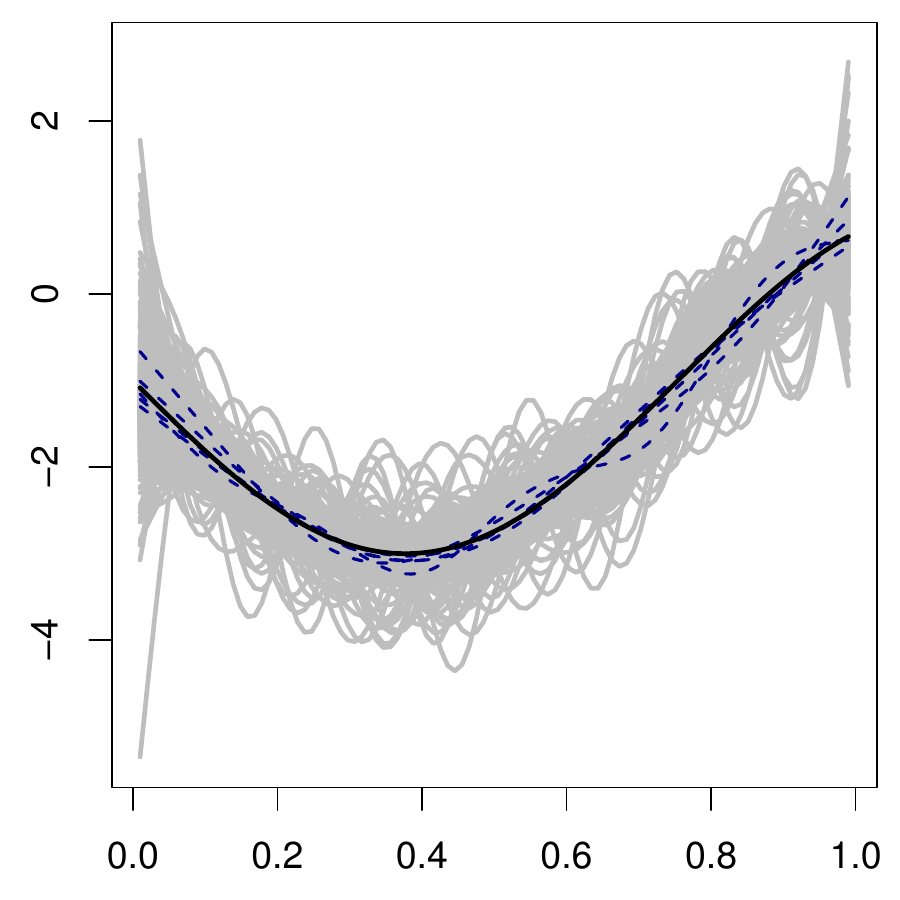}} \ \subfloat{\includegraphics[width = 0.496\textwidth]{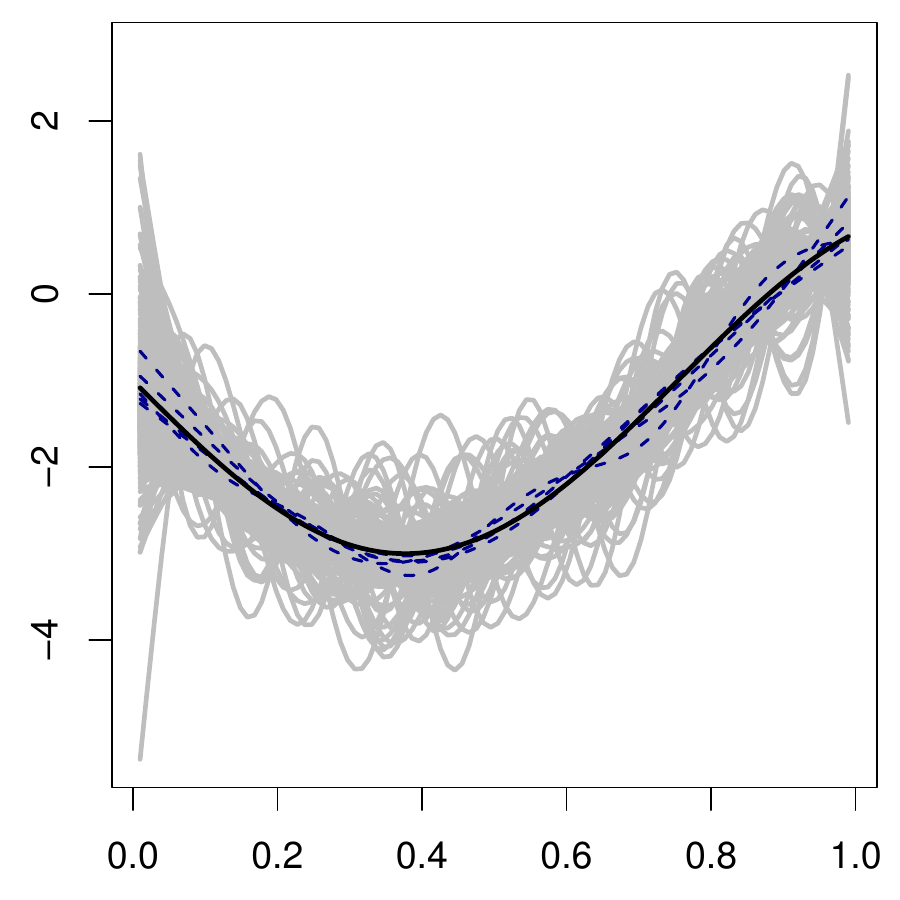}} \\
    \subfloat{\includegraphics[width = 0.496\textwidth]{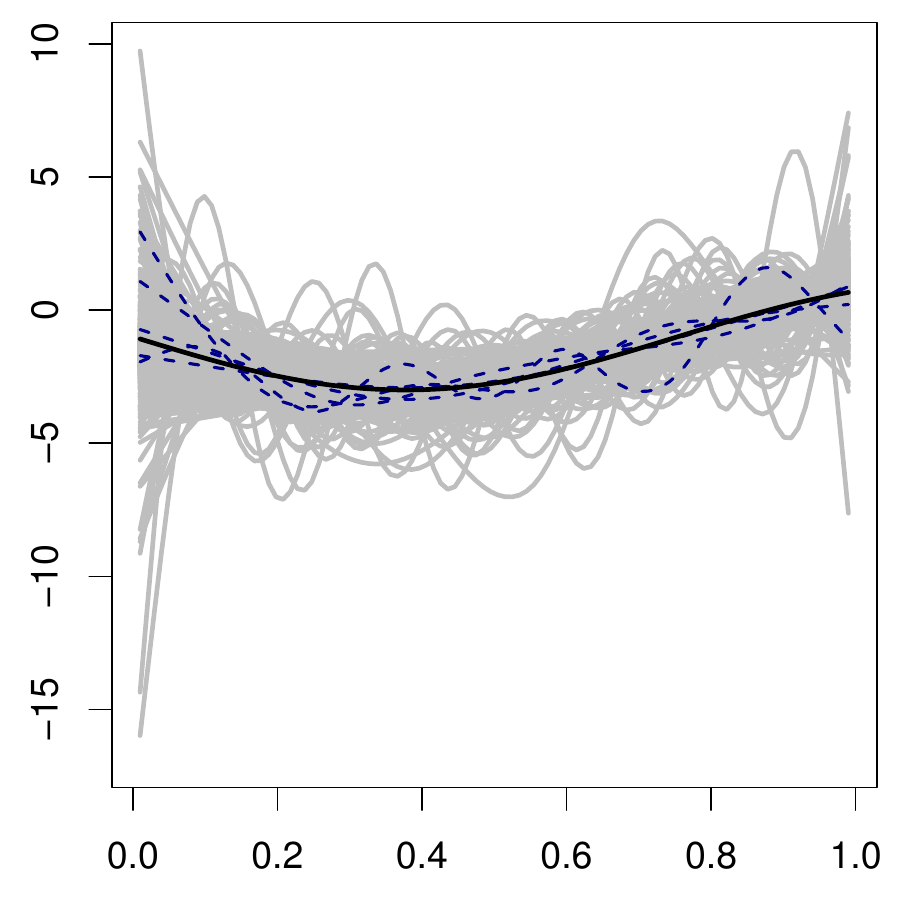}} \ \subfloat{\includegraphics[width = 0.496\textwidth]{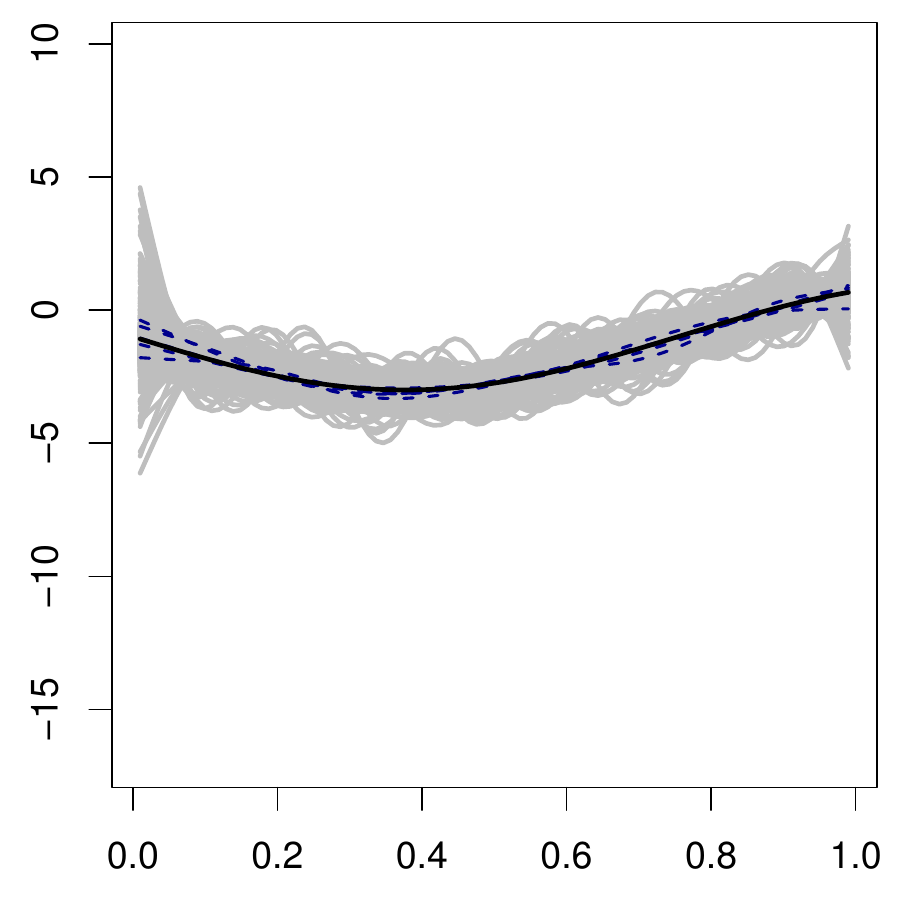}}
    \caption{1000 Least-squares (left) and Huber (right) estimates for \textbf{Model~1} (top panels) and \textbf{Model~2} (bottom panels)  under $\NSR=0.1$ and $p=100$. The lines (\solidd, \dashedd) depict the true coefficient function $\beta_0$ and the first $5$ estimated functions.}
    \label{fig:functions1}
\end{figure}

In summary, the results of our simulation study point towards the advantages of thin-plate spline estimators based on more slowly increasing $\rho$-functions relative to $\rho(x) = x^2$. With an appropriate choice of $\rho$, these estimators are very versatile in the sense that they are efficient in clean data and considerably safer in the presence of atypical observations than the least-squares thin-plate spline estimator.

\subsection{Computing Times}

Table~\ref{tab:times} provides the average computing times and standard errors for each estimator under all the settings of our numerical study. From these numbers it may be seen that, while the least-squares estimator is the fastest to compute, the Huber and logistic estimators are not too far behind. In fact, for \textbf{Model~1} and \textbf{Model~2}, for the Huber and logistic estimators in the vast majority of cases, the IRLS algorithm converges in at most four to five iterations. Thus, to obtain each one of these estimators it suffices to solve up to five least-square problems. By contrast, owing to its less smooth $\rho$-function, the $L_1$ estimator is much more computationally intensive.

% latex table generated in R 4.3.1 by xtable 1.8-4 package
% Fri Aug 18 09:07:47 2023
\begin{table}[H]
\centering
\caption{Means and standard errors (in brackets)
                          of the computation times in seconds.} 
\label{tab:times}
\resizebox{\textwidth}{!}{\begin{tabular}{cc|ccc|ccc}
   %\hline
 & & \multicolumn{3}{c}{$\NSR=0.1$} & \multicolumn{3}{c}{$\NSR=0.2$} \\
 & & $p=25$ & $p=50$ & $p=100$ & $p=25$ & $p=50$ & $p=100$ \\
   \hline
\parbox[t]{2mm}{\multirow{4}{*}{\rotatebox[origin=c]{90}{\textbf{{\small Model~1}}}}} & \textsf{S} & 0.156  \scriptsize{(0.015)}   & 0.337  \scriptsize{(0.031)}   & 0.982  \scriptsize{(0.101)}   & 0.156  \scriptsize{(0.014)}   & 0.340  \scriptsize{(0.036)}   & 0.986  \scriptsize{(0.106)}   \\ 
                                                                                             & \textsf{A} & 2.404  \scriptsize{(1.115)}   & 8.083  \scriptsize{(3.470)}   & 35.795  \scriptsize{(14.486)} & 2.632  \scriptsize{(1.164)}   & 8.473  \scriptsize{(3.543)}   & 36.316  \scriptsize{(13.402)} \\ 
                                                                                             & \textsf{H} & 0.489  \scriptsize{(0.712)}   & 0.772  \scriptsize{(0.633)}   & 2.675  \scriptsize{(1.520)}   & 0.752  \scriptsize{(1.014)}   & 1.113  \scriptsize{(1.673)}   & 2.588  \scriptsize{(1.124)}   \\ 
                                                                                             & \textsf{L} & 0.489  \scriptsize{(0.757)}   & 0.732  \scriptsize{(0.772)}   & 2.357  \scriptsize{(0.766)}   & 0.788  \scriptsize{(1.107)}   & 1.048  \scriptsize{(1.520)}   & 2.328  \scriptsize{(0.949)}   \\   \cdashline{1-8}
  \parbox[t]{2mm}{\multirow{4}{*}{\rotatebox[origin=c]{90}{\textbf{{\small Model~2}}}}} & \textsf{S} & 0.156  \scriptsize{(0.014)}   & 0.348  \scriptsize{(0.046)}   & 0.983  \scriptsize{(0.106)}   & 0.157  \scriptsize{(0.018)}   & 0.355  \scriptsize{(0.047)}   & 0.963  \scriptsize{(0.110)}   \\ 
                                                                                             & \textsf{A} & 2.451  \scriptsize{(1.192)}   & 8.333  \scriptsize{(3.550)}   & 36.243  \scriptsize{(13.968)} & 2.733  \scriptsize{(1.214)}   & 8.914  \scriptsize{(3.621)}   & 35.693  \scriptsize{(13.150)} \\ 
                                                                                             & \textsf{H} & 0.625  \scriptsize{(0.877)}   & 0.949  \scriptsize{(0.999)}   & 2.876  \scriptsize{(1.174)}   & 1.023  \scriptsize{(1.173)}   & 1.553  \scriptsize{(2.228)}   & 2.719  \scriptsize{(0.863)}   \\ 
                                                                                             & \textsf{L} & 0.694  \scriptsize{(1.017)}   & 0.928  \scriptsize{(1.077)}   & 2.636  \scriptsize{(0.701)}   & 1.188  \scriptsize{(1.367)}   & 1.573  \scriptsize{(2.319)}   & 2.534  \scriptsize{(0.664)}   \\   \cdashline{1-8}
  \parbox[t]{2mm}{\multirow{4}{*}{\rotatebox[origin=c]{90}{\textbf{{\small Model~3}}}}} & \textsf{S} & 0.158  \scriptsize{(0.015)}   & 0.352  \scriptsize{(0.040)}   & 0.979  \scriptsize{(0.111)}   & 0.154  \scriptsize{(0.018)}   & 0.354  \scriptsize{(0.049)}   & 0.937  \scriptsize{(0.104)}   \\ 
                                                                                             & \textsf{A} & 4.565  \scriptsize{(1.612)}   & 13.268  \scriptsize{(6.006)}  & 49.798  \scriptsize{(22.925)} & 4.614  \scriptsize{(1.267)}   & 14.792  \scriptsize{(5.286)}  & 48.045  \scriptsize{(19.948)} \\ 
                                                                                             & \textsf{H} & 3.154  \scriptsize{(1.680)}   & 8.140  \scriptsize{(6.355)}   & 3.898  \scriptsize{(3.896)}   & 3.215  \scriptsize{(1.279)}   & 10.592  \scriptsize{(5.431)}  & 5.475  \scriptsize{(6.696)}   \\ 
                                                                                             & \textsf{L} & 3.737  \scriptsize{(1.868)}   & 9.104  \scriptsize{(6.584)}   & 4.217  \scriptsize{(6.063)}   & 3.785  \scriptsize{(1.388)}   & 11.699  \scriptsize{(5.262)}  & 9.147  \scriptsize{(12.140)}  \\ 
   %\hline
\end{tabular}}
\end{table}

For \textbf{Model~3}, i.e., contamination in the predictor space, the Huber and logistic estimators require more iterations of the IRLS algorithm, but these  estimators remain still considerably easier to compute than the $L_1$ estimator. Overall, Table~\ref{tab:times} indicates that the Huber and logistic estimators are computationally feasible alternatives to the least-squares estimator particularly in light of their increased resistance to atypical observations.

\begin{figure}[H]
    \centering
    \subfloat{\includegraphics[width = 0.499\textwidth]{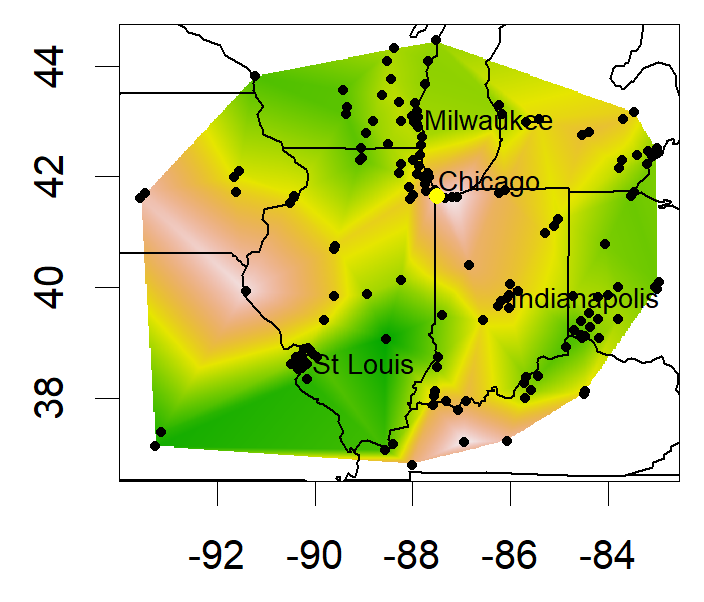}}  \subfloat{\includegraphics[width = 0.499\textwidth]{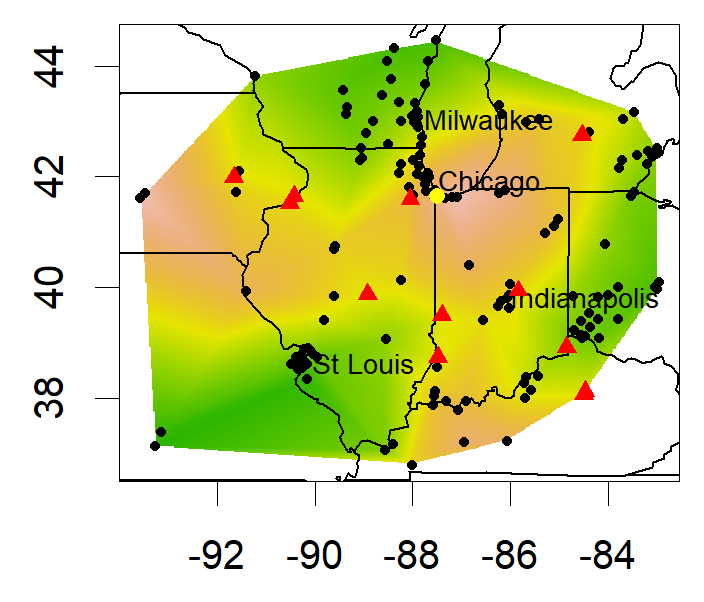}}
    \caption{Contours of the least-squares and Huber thin-plate spline coefficient function estimates on the convex hull of the data on the left and right panels, respectively. Lighter colors correspond to larger values of the estimates. The symbols (\circled, \circledd,  \triangled ) respectively indicate the positions of the monitoring stations, the position of our reference station and the outlying observations detected by the Huber estimates. }
    \label{fig:ozone1}
\end{figure}

\section{Application: Ozone Concentration in Midwestern United States}
\label{sec:6}

It is well-known that stratospheric ozone is beneficial to living organisms, as it protects them from ultraviolet sun radiation. High concentration of ground ozone, on the other hand, has been linked with a variety of respiratory problems, primarily for children and the elderly. Therefore, the ability to predict high concentrations of ground ozone may be helpful. As an illustration of the practical usefulness of the proposed family of estimators we analyze ground level ozone concentration in south Chicago as a function of ozone concentration in a broad area of the midwestern US. In particular, the data for this analysis consists of 7-hour average measurements of ground level ozone from 9am to 4pm in parts per billion (PPB) from $153$ stations during 89 days and are freely available as a part of the \texttt{fields} \textsf{R}-package \citep{Nyckha:2015} on CRAN.

To explain our approach in detail, let $\{Y_i\}_{i=1}^{89}$ denote the ground level ozone concentration in south Chicago on each one of the 89 days of measurement and let $\{X_i(t,s)\}_{i=1}^{89}$ denote the ground level ozone concentration on each one of these days as a function of longitude and latitude, denoted by $t$ and $s$, respectively. We consider the functional linear model
\begin{align*}
    Y_i = \alpha_0 +  \int_{37}^{45} \int_{-93}^{-83} X_i(t,s) \beta_0(t,s) \dd t \dd s + \epsilon_i, \quad (i=1, \ldots, 89),
\end{align*}
for some unknown $(\alpha_0,\beta_0) \in \mathbbm{R} \times \mathcal{H}^m(\mathbbm{R}^2)$. Here, the boundaries of integration result from the geographical positions of the measuring stations within the US. Of course, we are immediately faced with an incomplete data problem, as each $X_i$ is only observed at $152$ stations rather than in its entirety.  Nevertheless, our methodology is applicable and in order to estimate $(\alpha_0,\beta_0)$ we compute both the least-squares and Huber thin-plate spline estimators, which yield the contour plots in Figure~\ref{fig:ozone1}. Note that, since the ozone concentration reported by the monitoring station in south Chicago is used as the response variable, the geographic location of this station is excluded from the dataset.

\begin{figure}[H]
    \centering
    \subfloat{\includegraphics[width = 0.499\textwidth]{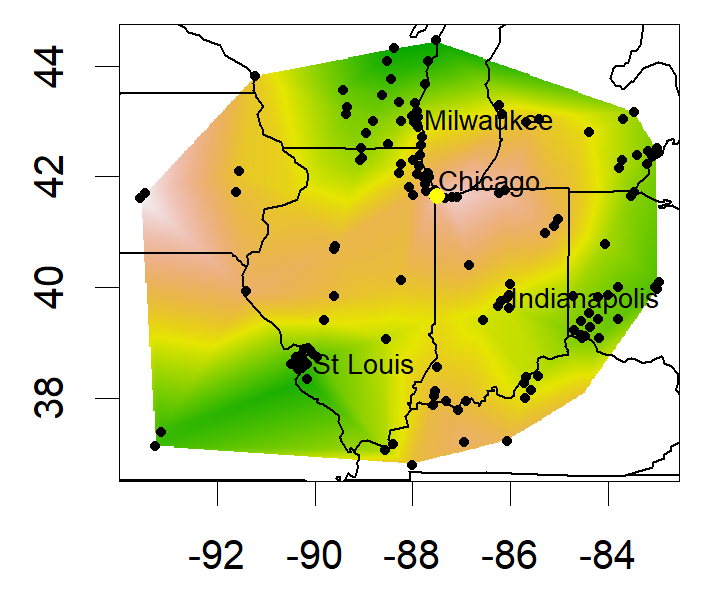}} \subfloat{\includegraphics[width = 0.499\textwidth]{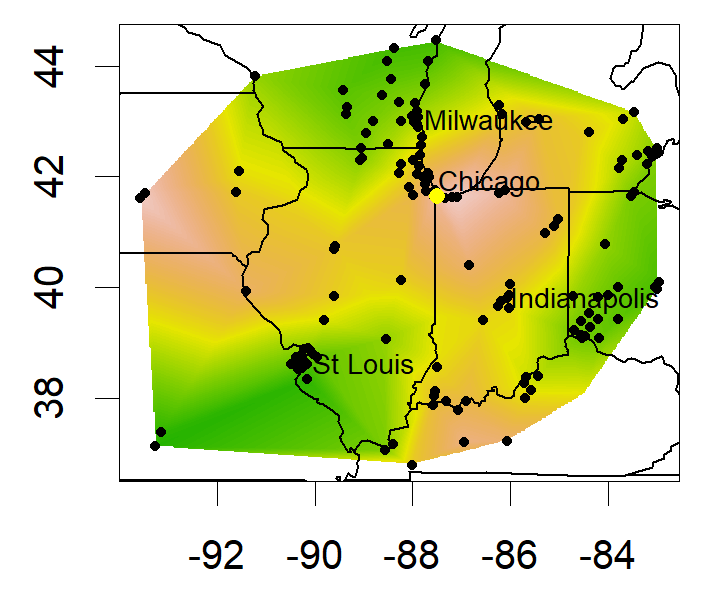}}
    \caption{Contours of the least-squares and Huber thin-plate spline coefficient function estimates on the convex hull of the data on the left and right panels, respectively, after the outliers have been removed. Lighter colors correspond to larger values of the estimates. The symbols (\circled, \circledd) respectively indicate the positions of the monitoring stations and the position of our reference station. }
    \label{fig:ozone2}
\end{figure}

It is natural to hypothesize that surface ozone concentration in one place would be positively correlated with surface ozone concentration in nearby places and that this correlation would wane as the distance increases. In other words, it is primarily ozone concentrations immediately in the vicinity of south Chicago that should be the most relevant. However, an examination of the contour plots in Figure~\ref{fig:ozone2} shows that while this is mostly true for the Huber estimates, the least-squares estimates indicate that the ozone concentrations in a large area in the north-west and south-east of St. Louis are very important predictors for the ozone concentration in south Chicago. This may seem puzzling, as the distance between these two places is more than 450 kilometers and so we may wonder whether the least-squares estimates are being distorted by atypical observations. 

In order to detect atypical and influential observations we may use the residuals of the more resistant Huber estimates. These estimates are much less drawn to atypical observations and therefore atypical observations show as large residuals. A classical detection rule in that respect involves examining the standardized residuals $r_i/\widehat\sigma_n$ and classifying as outliers observations whose absolute standardized residuals exceed $2.6$ \citep[Chapter 6]{Rousseeuw:1987}. In our case, this results in $14$ observations being classified as outliers; these observations are depicted with red triangles in the right panel of Figure~\ref{fig:ozone1}. 

It is interesting to observe that many of the outlying observations are located in or around the areas on which the least-squares and Huber estimates differ. As a sensitivity check, we remove the outlying observations from the data and recompute the estimators. Doing so now yields the contours in Figure~\ref{fig:ozone2}. These contours  reveal that while the Huber estimates have barely changed after the removal of the outlying observations, the least-squares estimates have undergone significant adjustment and are now very similar to the Huber estimates. Moreover, in the revised contour plot, the area around St. Louis seems to be much less influential in the determination of ground level ozone in Chicago and thus the estimates are now in line with prior expectations. This example demonstrates the sensitivity of least-squares estimators as well as the versatility of robust estimators, such as Huber estimators, which may be used to good effect both in the presence and absence of outliers and can even help detect outlying observations.

\section{Concluding Remarks}
\label{sec:7}

The theoretical and practical results of this paper provide justification for the use of a broad class of functional linear regression estimators in the ubiquitous setting of discretely sampled, possibly high dimensional, functional data. There are several research directions worth exploring from here, the most notable of which would be the generalization of our methodology to the setting of generalized linear models (GLM) allowing for different types of response variables, such as binary or count. The extension of our thin-plate spline estimator to the GLM context as well as the theoretical investigation of its properties in that setting is a promising topic for future research.

\iffalse
In the context of non-parametric regression, a robust family of GLM spline estimators was recently proposed by \citet{Kal:2023b}. The adaptation of this method for use in the setting of discretely observed functional data as well as 
\fi

Another interesting and practically useful direction of future research would be the investigation of the theoretical properties of robust model selection procedures, such as the $\tau$-scale of the leave-one-out residuals employed herein. To the best of our knowledge, there do not currently exist any theoretical results with respect to robust model selection in the functional data setting despite the fact that the vast majority of available estimation procedures make use of robust model selection criteria. Thus, research in this direction would fill a gap in the literature and provide practical guarantees with respect to the use of robust model selection procedures.

\section*{Acknowledgments}

Most of the present research was carried out while I. Kalogridis was visiting Charles University. I. Kalogridis gratefully acknowledges support from the Research Foundation-Flanders (project 1221122N). The research of S. Nagy was supported by Czech Science Foundation (project n. 23-05737S).

\appendix
\section{Appendix: Proofs of our Theoretical Results}

In our proofs, we occasionally use inequalities of the form $A \leq B$ for $A, B$ random variables defined on the same probability space $\Omega$. These inequalities are always to be interpreted in the almost sure sense, i.e. $A(\omega) \leq B(\omega)$ for $\mathbb{P}$-almost all $\omega \in \Omega$. Moreover, we use $\langle \cdot, \cdot \rangle_p$ denote the discrete inner product, i.e., $\langle f, g \rangle_p =  \sum_{j=1}^p f(\mathbf{t}_j) g(\mathbf{t}_j) \mu(A_j)$ and $\|\cdot\|_p$ to denote the associated norm, i.e., $\|f\|_p = \langle f, f\rangle_p^{1/2}$. For the proofs of Theorem~\ref{thm:1}, Corollary~\ref{cor:2}, Theorem~\ref{thm:2} and Theorem~\ref{thm:3} we assume that $\mu(A_j) = \mu(\mathcal I)/p$, but this assumption is not needed for the proofs of Proposition~\ref{Prop:1}, Corollary~\ref{Cor:1} and Proposition~\ref{Prop:2}. 

\subsection{Proof of Proposition~\ref{Prop:1}} % We begin by providing the proof of Proposition~\ref{Prop:1}.

%\begin{proof}[Proof of Proposition~\ref{Prop:1}]

We may restrict attention to $\beta \in \mathcal{H}^m(\mathbbm{R}^d)$, as $\alpha$ is a finite-dimensional parameter and therefore the existence of a minimizer $\widehat{\alpha}_n$ follows from the convexity of $\rho$. Our proof consists of two steps related to the existence and characterization of the minimizer $\widehat{\beta}_n \in \mathcal{H}^m(\mathbbm{R}^d)$.

\mypar{Step 1: Existence.} We aim to apply the existence theorem given by \citet[Theorem 2.9]{Gu:2013} and for this we need to characterize the null space of the penalty functional $J_m$. We claim that $J_m(g) = 0 \iff g \equiv 0$ so that the null space of $J_m$ consists only of the zero element. To prove this we shall show that $J_m(g) = 0 \iff g|_{\mathcal{I}} \equiv 0$ and $g$ is a polynomial on $\mathbbm{R}^d$ of total order $m$. Necessity is obvious, hence we focus on the sufficiency part, that is, we prove that (i) $J_m(g) = 0 \implies g|_{\mathcal{I}} = 0$, and (ii) \label{condii} $g$ is a polynomial on $\mathbbm{R}^d$ of total order $m$. Implication~(ii) is clear, as $J_m(g) = 0$ implies that $I_m(g) = 0$, which in turn implies that $g$ is a polynomial of total order $m$. Therefore, we need to show that, for every $g \in \mathcal{H}^m\left(\mathbbm{R}^d\right)$, $J_m(g) = 0$ implies that $g|_{\mathcal{I}} = 0$.  To see this, recall that by \ref{A1}, $\mathcal{I}$ is a bounded open set satisfying the uniform cone condition. One can thus consider also the Sobolev space $\mathcal{H}^{m}\left(\mathcal{I}\right)$ of functions $f \colon \mathcal{I} \to \mathbbm{R}$ defined analogously as $\mathcal{H}^m\left(\mathbbm{R}^d\right)$, equipped with the norm $\left\Vert \cdot \right\Vert_{\mathcal{H}^{m}\left(\mathcal{I}\right)}$ defined by
\begin{align}
    \label{eq:norm}
    \left\|g\right\|_{\mathcal{H}^m(\mathcal{I})}^2 = \int_{\mathcal{I}} |g(\mathbf{t})|^2 \dd \mathbf{t} + \sum_{m_1 + \ldots + m_d = m} \binom{m}{m_1, \ldots, m_d} \int_{\mathcal{I}} \left( \frac{\partial^m g(\mathbf{t})}{\partial t_{1}^{m_1} \ldots \partial t_{d}^{m_d}}  \right)^2 \dd \mathbf{t},
 \end{align}
where naturally $\mathbf{t} = (t_1, \dots, t_d)$. By Theorem 3.4 of \citet{Utr:1988}, there exist $C_0$ and $B_1$ depending only on $m, d, \mathcal{I}$ such that
\begin{align*}
\int_{\mathcal{I}}|g(\mathbf{t})|^2 d\mathbf{t} \leq \frac{C_0}{p} \sum_{j=1}^p |g(\mathbf{t}_j)|^2 + C_0B_1^{2m} \sum_{m_1 + \ldots + m_d = m} \binom{m}{m_1, \ldots, m_d} \int_{\mathcal{I}} \left( \frac{\partial^m g(\mathbf{t})}{\partial t_{1}^{m_1} \ldots \partial t_{d}^{m_d}}  \right)^2 \dd \mathbf{t}.
\end{align*}
By definition of $J_m(g)$ and since $\mathcal{I} \subset \mathbbm{R}^d$, we now find that 
\begin{align*}
\int_{\mathcal{I}}|g(\mathbf{t})|^2 \dd \mathbf{t} & \leq \frac{C_0}{p} \sum_{j=1}^p |g(\mathbf{t}_j)|^2 + C_0 B_1^{2m} \sum_{m_1 + \ldots + m_d = m} \binom{m}{m_1, \ldots, m_d} \int_{\mathbbm{R}^d} \left( \frac{\partial^m g(\mathbf{t})}{\partial t_{1}^{m_1} \ldots \partial t_{d}^{m_d}}  \right)^2 \dd \mathbf{t}
\\ & = 0,
\end{align*}
as $J_m(g) = 0$ implies that $g(\mathbf{t}_j) = 0$ for $j = 1, \ldots, p$, due to the fact that, by construction, $\mu(A_j)>0$. By hypothesis, $\mathcal{I}$ satisfies the uniform cone condition, so for $2m>d$ the Sobolev embedding theorem \citep[Theorem 6.3]{Adams:2003} implies the (compact) embedding
\begin{align*}
\mathcal{H}^{m}\left(\mathcal{I}\right) \to \mathcal{C}(\mathcal{I}),
\end{align*}
with the consequence that there exists a $c_0>0$ with the property that
\begin{align*}
\sup_{\mathbf{t} \in \mathcal{I}}\left|g(\mathbf{t})\right| \leq c_0 \left\|g\right\|_{\mathcal{H}^{m}\left(\mathcal{I}\right)},
\end{align*}
for every $g \in \mathcal{H}^{m}(\mathcal{I})$. Clearly, for every $g \in \mathcal{H}^{m}(\mathbbm{R}^d)$, $g|_{\mathcal{I}} \in \mathcal{H}^{m}(\mathcal{I})$, hence
\begin{align*}
\sup_{\mathbf{t} \in \mathcal{I}}\left|g|_{\mathcal{I}}(\mathbf{t})\right| & \leq c_0 \left\|g\right\|_{\mathcal{H}^m(\mathcal{I})} = 0,
\end{align*}
proving the assertion. Thus, the null space of $J_m$ consists only of the zero function so that the uniqueness condition of Theorem 2.9 in \citet{Gu:2013} is trivially satisfied. That theorem therefore applies and yields the existence of a minimizer $\widehat{\beta}_n \in \mathcal{H}^{m}(\mathbbm{R}^d)$ concluding the first step in our proof.

\mypar{Step 2: Characterization.} We claim that under the conditions of the proposition $\widehat{\beta}_n$ must necessarily be a natural thin-plate spline with knots at the $\{\mathbf{t}_j\}_{j=1}^p$. To show this, we apply Theorem~4 of \citet{Duchon:1977} according to which there exists exactly one natural thin-plate spline, $s$, taking prescribed values at the $\{\mathbf{t}_j\}_{j=1}^p$ while minimizing $I_m^2(g)$. That is, for every set of points $y_1, \ldots, y_p \in \mathbbm{R}$, $s: \mathbbm{R}^d \to \mathbbm{R}$ it solves the problem
\begin{align*}
\min_{g \in \mathcal{H}^m\left(\mathbbm{R}^d\right)} I_m^2(g), \quad \text{subject to} \quad g(\mathbf{t}_j) = y_j, \ (j = 1, \ldots, p).
\end{align*}
By \textbf{Step~1} of this proof, there exists a minimizer $\widehat{\beta}_n \in \mathcal{H}^{m}(\mathbbm{R}^d)$. Denote $\widehat{\boldsymbol{\beta}}_n = (\widehat{\beta}_n(\mathbf{t}_1), \ldots, \widehat{\beta}_n(\mathbf{t}_p))$. Then, by Theorem~4 of \citet{Duchon:1977}, there exists a natural thin-plate spline, $\widehat{s}_n \in \mathcal{H}^{m}(\mathbbm{R}^d)$, interpolating $\widehat{\beta}_n$ at the $\{\mathbf{t}_j\}_{j=1}^p$ while $I_m(\widehat{s}_n) <  I_m(\widehat{\beta}_n)$. Because of these properties, it follows that  $\widehat{\beta}_n$ must itself be a natural thin-plate spline. Indeed, let $L_n(\widehat{\alpha}_n, \widehat{\beta}_n, \widehat\sigma_n)$ denote the minimal value of the objective function, that is,
\begin{align*}
L_n(\widehat{\alpha}_n, \widehat{\beta}_n, \widehat\sigma_n) = \frac{1}{n} \sum_{i=1}^n \rho\left(\frac{Y_i-\widehat{\alpha}_n -  \sum_{j=1}^p X_i(\mathbf{t}_j) \widehat{\beta}_n(\mathbf{t}_j) \mu(A_j)}{\widehat\sigma_n} \right) + \lambda \left[\sum_{j=1}^p |\widehat{\beta}_n(\mathbf{t}_j)|^2 \mu(A_j) + I_m^2\left(\widehat{\beta}_n \right) \right].
\end{align*}
If $\widehat{\beta}_n$ were not a natural thin-plate spline, then $\widehat{s}_n$ would leave the first two terms of $L_n(\widehat{\alpha}_n, \widehat{\beta}_n, \widehat\sigma_n)$ unchanged, but, due to $I_m(\widehat{s}_n) <  I_m(\widehat{\beta}_n)$, we would have $L_n(\widehat{\alpha}_n, \widehat{s}_n, \widehat\sigma_n) < L_n(\widehat{\alpha}_n, \widehat{\beta}_n,\widehat\sigma_n)$, contradicting the fact that $\widehat{\beta}_n$ is a minimizer. The proof is complete.
%\end{proof}

\subsection{Proof of Corollary~\ref{Cor:1}}

Let $L_{n, \inf}$ denote the infimum of the objective function, that is,
\begin{align*}
L_{n, \inf} = \inf_{(\alpha,\beta) \in \mathbbm{R}\times  \mathcal{H}^{m}(\mathbbm{R}^d)} \left[ \frac{1}{n} \sum_{i=1}^n \rho\left(\frac{Y_i-\alpha -  \sum_{j=1}^p X_i(\mathbf{t}_j) \beta(\mathbf{t}_j) \mu(A_j)}{\widehat\sigma_n} \right) + \lambda J_m^2\left(\beta \right)\right].
\end{align*}
Note that $L_{n,\inf}$ is finite, as the objective function is bounded from below by $0$. Define $\mathcal{D} = \{(\alpha, \beta) \in \mathbbm{R} \times \mathcal{H}^m(\mathbbm{R}^d): L_n(\alpha, \beta, \widehat\sigma _n) = L_{n,\inf} \}$. Clearly, $\mathcal{D} = \bigcap_{k=1}^{\infty} \mathcal{D}_k$ with $\mathcal{D}_k =  \{(\alpha, \beta) \in \mathbbm{R} \times \mathcal{H}^m(\mathbbm{R}^d): L_n(\alpha, \beta, \widehat\sigma_n) \leq L_{n,\inf} + k^{-1} \}$. By the convexity of $L_n(\alpha, \beta, \widehat\sigma_n)$, each $\mathcal{D}_k$ is a convex set. Hence $\mathcal{D}$, as the intersection of convex sets, is also convex. By Proposition~\ref{Prop:1}, it is also non-empty, as $(\widehat{\alpha}_n, \widehat{\beta}_n) \in \mathcal{D}$. Suppose now that there exist distinct $(\alpha_1, \beta_1)$ and $(\alpha_2, \beta_2)$ in $\mathcal{D}$. By the convexity of $\mathcal{D}$ and the strict convexity of $\rho(x)$, we would have
\begin{align*}
L_{n,\inf} = L_n(\alpha_1/2+\alpha_2/2, \beta_1/2+\beta_2/2, \widehat\sigma_n) < L_n(\alpha_1, \beta_1, \widehat\sigma_n)/2 + L_n(\alpha_2, \beta_2, \widehat\sigma_n)/2 = L_{n, \inf},
\end{align*}
which is clearly impossible. It follows that $\mathcal{D}$ must consist of a single element, as was to be shown.

\subsection{Proof of Proposition~\ref{Prop:2}}

%We continue with the proof of Proposition~\ref{Prop:2}.

%\begin{proof}[Proof of Proposition~\ref{Prop:2}]

Define for simplicity $\|\beta\|_p^2 = p^{-1} \sum_{j=1}^p |\beta(t_j)|^2$ so that
\begin{align*}
\widetilde{J}_m^{ 2}(\beta) = \left\|\mathcal{P}\beta \right\|^2_p + \left\|\beta^{(m)}\right\|^2.
\end{align*}
By the properties of Hilbert projections \citep[see, e.g.,][Theorem 2.7]{Conway:1990}, we have $\|\mathcal{P}\beta \|^2_p \leq \|\beta\|_p^2$ and from this we can immediately see that $\widetilde{J}_m^{2}(\beta) \leq J_m^2(\beta)$. Thus, Proposition~\ref{Prop:2} holds with $c_1 = 1$. 

To prove the other inequality, use the Hilbert projection theorem \citep[Theorem 2.5]{Conway:1990} to write
\begin{align*}
\beta = \mathcal{P}\beta +  (\mathcal{I}-\mathcal{P})\beta.
\end{align*}
with $\mathcal{I}$ denoting the identity operator on $\mathcal{H}^m(0,1)$. Clearly,
\begin{align*}
J_m^2(\beta) & = \left\|\mathcal{P}\beta\right\|_p^2 + \left\|\beta- \mathcal{P}\beta \right\|_p^2 + 2 \langle \mathcal{P}\beta, \beta- \mathcal{P}\beta \rangle_p + \left\|\beta^{(m)}\right\|^2 
\\ & \leq 2 \left(\left\|\mathcal{P}\beta\right\|_p^2 + \left\|\beta- \mathcal{P}\beta \right\|_p^2\right) + \left\|\beta^{(m)}\right\|^2.
\end{align*}
Now, again by properties of the projections \citep[Theorem 2.5]{Conway:1990}, we have $ \|\beta- \mathcal{P}\beta\|_p^2 \leq \|\beta- v\|_p^2$ for every other polynomial $v$ of order $m$. Choose $v$ to be the Taylor polynomial about zero of order $m$, viz,
\begin{align*}
v(t) =  \sum_{i=0}^{m-1} \beta^{(i)}(0) \frac{t^{i}}{i!},
\end{align*}
so that, using the integral form of the remainder for Taylor polynomials and the Cauchy-Schwarz inequality,
\begin{align*}
\left\|\beta- \mathcal{P}\beta \right\|_p^2 \leq \left\|\beta- v \right\|_p^2 =  \frac{1}{p} \sum_{j=1}^p \left| \int_{0}^1 \frac{(t_j-u)_{+}^{m-1}}{(m-1)!} \beta^{(m)}(u) \dd u \right|^2 \leq C \left\|\beta^{(m)} \right\|^2,
\end{align*}
for some $C>0$ depending only on $m$. It follows that
\begin{align*}
J_m^2(\beta) \leq (2+2C)\widetilde{J}_m^{2}(\beta) = c_2^2 \widetilde{J}_m^{2}(\beta),
\end{align*}
for $c_2 = \sqrt{2(1+C)}$, as was to be shown.
% \end{proof}

%
%
%

\subsection{Auxiliary Results on Empirical Processes}

Our next task is to derive the modulus of continuity of the empirical process associated with our estimator, which we do in general terms. Therefore, the notations used in this section are mostly independent of that used in the rest of our paper.

Let $\mathcal{B}$ denote a subset of a (semi-)metric space $(\Lambda, {\mathfrak D})$ endowed with a semi-metric ${\mathfrak D}$ and let $N(\delta, \mathcal{B},{\mathfrak D})$ denote the number of balls required to cover $\mathcal{B}$, that is, 
\begin{align*}
N(\delta, \mathcal{B},{\mathfrak D}) = \min(N: \ \text{there exist} \  \{\beta_1, \ldots, \beta_N\} \subset \mathcal{B} \ \text{such that} \ \max_{\beta \in \mathcal{B}} \min_{j=1, \ldots, N} {\mathfrak D}(\beta, \beta_j) \leq \delta  ).
\end{align*}
Further, let $H(\delta, \mathcal{B},{\mathfrak D}) = \log N(\delta, \mathcal{B},{\mathfrak D})$ denote the $\delta$-entropy of $\mathcal{B}$. We assume that the semi-metric ${\mathfrak D}$ is of the form
\begin{align}   \label{eq:frak di}
{\mathfrak D}^2 = \frac{1}{n} \sum_{i=1}^n {\mathfrak D}_i^2,
\end{align}
where ${\mathfrak D}_1, \ldots, {\mathfrak D}_n$ are also semi-metrics on $\mathcal{B}$. Consider real-valued random variables $U_{i,\beta}, i =1, \ldots, n,\ \beta \in \mathcal{B}$. We assume that the processes $\{ U_{i,\beta}: \beta \in \mathcal{B}\}$, $i = 1, \ldots, n$, are independent and centered, and that
\begin{align}
\label{VDG:1}
\left|U_{i,\beta} - U_{i, \widetilde{\beta}} \right| \leq V_i\, {\mathfrak D}_i(\beta, \widetilde{\beta}), \quad( i =1, \ldots, n),\ \beta, \widetilde{\beta} \in \mathcal{B}
\end{align}
for uniformly sub-Gaussian $V_i$. That is,
\begin{align}
\label{VDG:2}
\max_{1 \leq i \leq n} K^2 \mathbb{E}\left\{e^{V_i^2/K^2}-1 \right\} \leq \sigma_0^2,
\end{align}
for positive $K$ and $\sigma_0$. Lemma 8.5 in \citet{vandeGeer:2000} presents an exponential inequality for the supremum of the difference of these processes.

\begin{lemma}[Lemma 8.5 in \citealp{vandeGeer:2000}]	
\label{lem:1}
Assume \eqref{VDG:1} and \eqref{VDG:2} and that $\sup_{\beta \in \mathcal{B}} {\mathfrak D}(\beta, \beta_0) \leq R$ for some $\beta_0 \in \Lambda$. Then for some constant $\widetilde{C}$ depending only on $K$ and $\sigma_0$ and for all $\delta>0$ and $\sigma>0$ satisfying
\begin{align}
\label{VDG:3}
\delta \geq \widetilde{C} \max\left( \int_{\delta/(8\sigma)}^{R} H^{1/2}(u, \mathcal{B}, {\mathfrak D}) \dd u, R \right),
\end{align}
we have
\begin{align}
\label{VDG:4}
\mathbb{P} \left( \sup_{\beta \in \mathcal{B}} \left|\frac{1}{\sqrt{n}} \sum_{i=1}^n \left(U_{i,\beta}-U_{i,\beta_0} \right) \right| \geq \delta, \frac{1}{n} \sum_{i=1}^n V_i^2 \leq \sigma^2 \right) \leq \widetilde{C} \exp\left[ - \frac{\delta^2}{\widetilde{C}^2 R^2} \right].
\end{align}
\end{lemma}

This lemma is not directly applicable in our proofs, as in our setting we aim to identify $\Lambda$ with the Sobolev space $\mathcal{H}^{m}(\mathcal{I})$ and $\mathcal{B}$ with a large subset of $\mathcal{H}^{m}(\mathcal{I})$ that does not have finite entropy. That is, the entropy integral in \eqref{VDG:3} diverges. But we will decompose $\mathcal{B}$ as
\begin{align*}
\mathcal{B} = \bigcup_{M \geq 1} \mathcal{B}_{M},
\end{align*}
with $\mathcal{B}_{M} = \{f \in \mathcal{B}, 1+J_m(f) \leq M\}$, where in general, $J_m$ can be any function $J_m \colon \mathcal{B} \to \mathbbm{R}_+$, such that $\mathcal{B}_{M}$ has finite entropy. In fact, under our assumptions it will be shown that
\begin{align*}
H(\delta, \mathcal{B}_M, {\mathfrak D}) \leq A_0 \left( \frac{M}{\delta} \right)^{\frac{d}{m}}, \quad \delta>0,
\end{align*}
for $A_0 > 0$ a constant, with the consequence that the corresponding entropy integral behaves like
\begin{align}
\label{VDG:5}
\int_{0}^{\delta} H^{1/2}(u, \mathcal{B}_M, {\mathfrak D}) \dd u \leq A_0^{\star} M^{\frac{d}{2m}} \delta^{1-\frac{d}{2m}}, \quad \delta>0,
\end{align}
for some global $A_0^{\star}>0$ depending only on $m$. In~\eqref{VDG:5}, we have also taken $\sigma \to \infty$ in the RHS of \eqref{VDG:3}. This is permissible,  because the LHS of~\eqref{VDG:5} converges for $2m>d$. Lemma~\ref{lem:2} below greatly extends Lemma~\ref{lem:1} by establishing the modulus of continuity of 
\begin{align*}
\sup_{\beta \in \mathcal{B}} \left|\frac{1}{\sqrt{n}} \sum_{i=1}^n \left( U_{i,\beta}-U_{i,\beta_0} \right) \right|,
\end{align*}
both in terms of ${\mathfrak D}$ and the function $J_m$. The result is not only useful for our purposes but it can also be of general interest.

\begin{lemma}
\label{lem:2}
Assume \eqref{VDG:1}, \eqref{VDG:2}, \eqref{VDG:5} and that $\sup_{\beta \in \mathcal{B}} {\mathfrak D}(\beta, \beta_0) \leq R$ for some $\beta_0 \in \Lambda$. Then, there exists $c>0$ such that for all $T \geq c$ we have
\begin{align}
\label{VDG:6}
\mathbb{P} \left( \sup_{\beta \in \mathcal{B}} \frac{\left|\frac{1}{\sqrt{n}} \sum_{i=1}^n \left(U_{i,\beta}-U_{i,\beta_0} \right) \right|}{{\mathfrak D}^{1-\frac{d}{2m}}(\beta, \beta_0) \{1+J_m(\beta)\}^{\frac{d}{2m}}} \geq T, \frac{1}{n} \sum_{i=1}^n V_i^2 \leq \sigma^2 \right) \leq c \exp\left[ - \frac{T^2}{c^2} \right].
\end{align}
\end{lemma}
\begin{proof} Under our assumptions, the conditions of Lemma~\ref{lem:1} are satisfied for each $\mathcal{B}_M$, hence the lemma is applicable. Our proof consists of iterative application of \eqref{VDG:4} on each $\mathcal{B}_{M}$ and the peeling technique \citep[see, e.g.,][p. 70]{vandeGeer:2000}. We break down the proof in two steps where in the first step we deal with $\mathcal{B}_M$ and in the second step pass to $\mathcal{B}$.

\mypar{Step 1.} We will prove that there exists a $c_1>0$ such that for all $ M \geq 1$,
\begin{align}
\label{VDG:7}
\mathbb{P} \left(  \sup_{\beta \in \mathcal{B}_M} \frac{ \left|\frac{1}{\sqrt{n}} \sum_{i=1}^n \left( U_{i,\beta}-U_{i,\beta_0} \right) \right| }{ {\mathfrak D}^{1-\frac{d}{2m}}(\beta, \beta_0)} \geq T \left( \frac{M}{2} \right)^{\frac{d}{2m}}, \frac{1}{n} \sum_{i=1}^n V_i^2 \leq \sigma^2\right) \leq c_1 \exp\left[- \frac{T^2 M^{\frac{d}{m}}}{c_1^2} \right].
\end{align}
To show \eqref{VDG:7} let $T = 2 \widetilde{C} A_0^{\star}  $ and use Boole's inequality to obtain
\begin{align*}
\mathbb{P} &\left(  \sup_{\beta \in \mathcal{B}_M} \frac{ \left|\frac{1}{\sqrt{n}} \sum_{i=1}^n \left( U_{i,\beta}-U_{i,\beta_0} \right) \right| }{ {\mathfrak D}^{1-\frac{d}{2m}}(\beta, \beta_0)} \geq T \left( \frac{M}{2} \right)^{\frac{d}{2m}}, \frac{1}{n} \sum_{i=1}^n V_i^2 \leq \sigma^2 \right)
\\ 
& \leq \sum_{s = 1}^{\infty}  \mathbb{P} \left(  \sup_{\beta \in \mathcal{B}_M, 2^{-s} R \leq {\mathfrak D}(\beta, \beta_0) \leq 2^{-s+1}R} \frac{ \left|\frac{1}{\sqrt{n}}  \sum_{i=1}^n \left( U_{i,\beta}-U_{i,\beta_0} \right) \right| }{ {\mathfrak D}^{1-\frac{d}{2m}}(\beta, \beta_0)} \geq T \left( \frac{M}{2} \right)^\frac{d}{2m}, \frac{1}{n} \sum_{i=1}^n V_i^2 \leq \sigma^2  \right)
\\ 
& \leq \sum_{s = 1}^{\infty}  \mathbb{P} \left(  \sup_{\beta \in \mathcal{B}_M, {\mathfrak D}(\beta, \beta_0) \leq 2^{-s+1}R}  \left|\frac{1}{\sqrt{n}} \sum_{i=1}^n \left( U_{i,\beta}-U_{i,\beta_0} \right) \right| \geq T \left(2^{-s} R\right)^{1-\frac{d}{2m}} \left( \frac{M}{2} \right)^\frac{d}{2m}, \frac{1}{n} \sum_{i=1}^n V_i^2 \leq \sigma^2\right)
\\ 
& = \sum_{s = 1}^{\infty}  \mathbb{P} \left(  \sup_{\beta \in \mathcal{B}_M, {\mathfrak D}(\beta, \beta_0) \leq 2^{-s+1}R} \left|\frac{1}{\sqrt{n}} \sum_{i=1}^n \left( U_{i,\beta}-U_{i,\beta_0} \right) \right|   \geq \widetilde{C} A_0^{\star}  \left(2^{-s+1} R\right)^{1-\frac{d}{2m}} M^\frac{d}{2m}, \frac{1}{n} \sum_{i=1}^n V_i^2 \leq \sigma^2  \right).
\end{align*}
Recalling that $M \geq 1$, we may apply Lemma~\ref{lem:1} on each of the summands to see that
\begin{align*}
\mathbb{P} & \left(  \sup_{\beta \in \mathcal{B}_M} \frac{ \left|\frac{1}{\sqrt{n}} \sum_{i=1}^n \left( U_{i,\beta}-U_{i,\beta_0} \right) \right| }{ {\mathfrak D}^{1-\frac{d}{2m}}(\beta, \beta_0)} \geq T  \left( \frac{M}{2} \right)^\frac{d}{2m}, \frac{1}{n} \sum_{i=1}^n V_i^2 \leq \sigma^2\right) 
\\ & \quad \leq \sum_{s=1}^{\infty}  \widetilde{C} \exp \left[ - \frac{ |A_0^{\star}|^2 (2^{-s+1}R)^{-\frac{d}{m}} M^{\frac{d}{m}}}{\widetilde{C}^2} \right] \leq c_1 \exp \left[ -\frac{T^2 M^{\frac{d}{m}}}{c_1^2} \right],
\end{align*}
for some $c_1>0$ depending only on $\widetilde{C}$, $m$ and $d$, as the series converges. The proof of \eqref{VDG:7} is complete.

\mypar{Step 2.} In the second step we use again the peeling technique along with \eqref{VDG:7} in order to complete the proof of the lemma. In particular, we have
\begin{align*}
\mathbb{P}  & \left( \sup_{\beta \in \mathcal{B}} \frac{\left|\frac{1}{\sqrt{n}}\sum_{i=1}^n \left(U_{i,\beta}-U_{i,\beta_0} \right) \right|}{{\mathfrak D}^{1-\frac{d}{2m}}(\beta, \beta_0) \{1+J_m(\beta)\}^{\frac{d}{2m}}} \geq T, \frac{1}{n} \sum_{i=1}^n V_i^2 \leq \sigma^2\right) \\
& \leq \sum_{s = 0}^{\infty} \mathbb{P} \left( \sup_{\beta \in \mathcal{B}, 2^{s} \leq 1+J_m(\beta) \leq 2^{s+1}} \frac{\left|\frac{1}{\sqrt{n}} \sum_{i=1}^n \left(U_{i,\beta}-U_{i,\beta_0} \right) \right|}{{\mathfrak D}^{1-\frac{d}{2m}}(\beta, \beta_0) \{1+J_m(\beta)\}^{\frac{d}{2m}}} \geq T, \frac{1}{n} \sum_{i=1}^n V_i^2 \leq \sigma^2 \right)
\\ & \leq \sum_{s=0}^{\infty} \mathbb{P} \left( \sup_{\beta \in \mathcal{B}, 1+J_m(\beta) \leq 2^{s+1}} \frac{\left|\frac{1}{\sqrt{n}} \sum_{i=1}^n \left(U_{i,\beta}-U_{i,\beta_0} \right) \right|}{{\mathfrak D}^{1-\frac{d}{2m}}(\beta, \beta_0)} \geq T 2^\frac{sd}{2m}, \frac{1}{n} \sum_{i=1}^n V_i^2 \leq \sigma^2 \right)
\\ & = \sum_{s=0}^{\infty} \mathbb{P} \left( \sup_{\beta \in \mathcal{B}_{2^{s+1}}} \frac{\left|\frac{1}{\sqrt{n}} \sum_{i=1}^n \left(U_{i,\beta}-U_{i,\beta_0} \right) \right|}{{\mathfrak D}^{1-\frac{d}{2m}}(\beta, \beta_0)} \geq T \left(\frac{2^{s+1}}{2} \right)^\frac{d}{2m}, \frac{1}{n} \sum_{i=1}^n V_i^2 \leq \sigma^2 \right)
\\ & \leq \sum_{s=0}^{\infty} c_1 \exp\left[ -\frac{T^2 2^{\frac{(s+1)d}{m}}}{c_1^2} \right] \leq c \exp\left[ - \frac{T^2}{c^2}\right],
\end{align*}
for some $c>0$, as the series converges. This completes the proof of the lemma.

\end{proof}
\subsection{Proof of Theorem~\ref{thm:1}}

Before presenting the proof of Theorem~\ref{thm:1}, it will be helpful to state and prove an auxiliary lemma involving ratios of sequences.

\begin{lemma}
\label{lem:3}
Let $\{x_n\}_{n=1}^{\infty}$ denote a sequence of non-negative real numbers and for $a>b$ denote the sequences $y_n = x_n^{a}/(1+x_n)^b$ and $z_n = x_n/(1+x_n)$. Then,
\begin{enumerate}[label=\Alph*., ref=\Alph*]
    \item \label{1.} For any $M > 0$ such that $\sup_n y_n \leq M$ there exists a finite $B = B(M)>0$ such that $\sup_n x_n \leq B$.
    \item \label{2.} If there exists an $\varepsilon \in (0,1)$ such that $z_n \leq 1-\varepsilon$ for all large $n$, then we can find finite $B = B(\varepsilon) > 0$ such that $\sup_n x_n \leq B$.
\end{enumerate}
% In particular, if $\left\{X_n\right\}_{n=1}^\infty$ is a sequence of non-negative random variables and $Y_n = X_n^{a}/(1+X_n)^b$, $Z_n = X_n/(1+X_n)$, then
% \begin{enumerate}[label=\Alph*.]
%     \item $Y_n = O_{\mathbb{P}}(1)$ implies $X_n = O_{\mathbb{P}}(1)$, and
%     \item If there exists an $\varepsilon \in (0,1)$ such that $z_n \leq 1-\varepsilon$ for all large $n$, then we can find finite $B = B(\varepsilon) > 0$ such that $\sup_n x_n \leq B$.
% \end{enumerate}
\end{lemma}

\begin{proof}[Proof of Lemma~\ref{lem:3}]
We start by proving part~\ref{1.}. Assume for contradiction that $\{x_n\}_{n=1}^{\infty}$ is unbounded. Then $\{x_n\}_{n=1}^{\infty}$ contains a  subsequence $\{x_{n_k}\}_{k=1}^{\infty}$ that diverges to $\infty$. But then $y_{n_k} \sim x_{n_k}^{a-b}$ where $a_n \sim b_n$ means that $\lim a_n/b_n = 1$. But then, since, by assumption, $a>b$, $y_{n_k}$ also diverges contradicting its boundedness. It follows that $\{x_n\}_{n=1}^{\infty}$ must be bounded, as was to be shown.

To prove part~\ref{2.}, it suffices to note that if $\{x_n\}_{n=1}^{\infty}$ contained a divergent subsequence, $\{x_{n_k}\}_{k=1}^{\infty}$, then $ \lim_{k \to \infty} z_{n_k} =1$ contradicting the fact that $z_n \leq 1-\varepsilon$ eventually. 
\end{proof}

%\begin{proof}[Proof of Theorem 1]

The proof of Theorem~\ref{thm:1} consists of the convexity step of \citet{vandeGeer:2002} and the derivation of the modulus of continuity of the empirical process with respect to the semi-metric $\left\| \cdot \right\|_{n,p}$ that is induced by our penalized M-estimator. Let $L_n(\beta)$ denote the objective function, that is,
\begin{align*}
L_n(\beta) = \frac{1}{n} \sum_{i = 1}^n \rho\left(Y_i - \sum_{j=1}^p X_i \left( \mathbf{t}_j \right) \beta\left(\mathbf{t}_j\right)\mu\left(A_j\right) \right) + \lambda J_m^2\left(\beta \right).
\end{align*}
By definition of the minimizer, we have $L_n(\widehat{\beta}_n) \leq L_n(\beta)$ for all $\beta \in \mathcal{H}^{m}(\mathbbm{R}^d)$. Moreover, $L_n$ is convex, as it is the sum of two convex functions. Define the convex combination 
    \begin{equation}    \label{eq:beta tilde}
    \widetilde{\beta}_n  = \gamma_n \widehat{\beta}_n + (1-\gamma_n)\beta_0 \qquad \mbox{with} \qquad \gamma_n = 1/(1+\|\widehat{\beta}_n-\beta_0\|_{n,p}).    
    \end{equation} 
Clearly $\gamma_n \in (0,1)$ and, by convexity, 
\begin{align}
\label{eq:A1}
L_n(\widetilde{\beta}_n) \leq \gamma_n L_n(\widehat{\beta}_n) + (1-\gamma_n) L_n(\beta_0) \leq L_n(\beta_0),
\end{align}
as, by assumption, $\beta_0 \in \mathcal{H}^m(\mathbbm{R}^d)$. Define for simplicity
\begin{align*}
\;M_n(\beta) := \frac{1}{n} \sum_{i=1}^n \rho\left(\epsilon_i + d_i + \sum_{j=1}^p X_i \left( \mathbf{t}_j\right) \left( \beta_0\left(\mathbf{t}_j\right) - \beta \left(\mathbf{t}_j \right) \right) \mu\left(A_j \right) \right),
\end{align*}
so that $L_n(\beta) = M_n(\beta) + \lambda J^2_m(\beta)$. Furthermore, define 
\begin{align*}
M(\beta) := \frac{1}{n} \sum_{i=1}^n \mathbb{E}_{\epsilon_i} \left\{ \rho\left(\epsilon_i + d_i +  \sum_{j=1}^p X_i \left( t_j\right)\left( \beta_0\left(\mathbf{t}_j\right) - \beta \left(\mathbf{t}_j \right) \right)\mu\left(A_j \right)\right) \right\},
\end{align*}
where $\mathbb{E}_{\epsilon_i} \{ \cdot \}$ denotes expectation with respect to the $\epsilon_i$ (recall that the $\epsilon_i$ and $X_i$ are assumed to be independent and the $\epsilon_i$ are independent). With this notation, rearranging \eqref{eq:A1} yields
\begin{align}
\label{eq:A2}
\left(M(\widetilde{\beta}_n) - M(\beta_0)\right) + \lambda J_m^2(\widetilde{\beta}_n) \leq \left( M_n(\beta_0) - M(\beta_0) \right) - \left( M_n(\widetilde{\beta}_n) - M(\widetilde{\beta}_n) \right) +\lambda J_m^2(\beta_0).
\end{align}
Our approach consists of deriving a lower bound on the LHS of \eqref{eq:A2} in terms of $\| \widetilde{\beta}_n-\beta_0 \|_{n,p}$ and an upper bound on the RHS, also in terms of $\|\widetilde{\beta}_n-\beta_0 \|_{n,p}$ on a set $A_{\varepsilon} \subset \Omega$ with $\mathbb{P}(A_{\varepsilon}) \geq 1-\varepsilon$ for any given $\varepsilon>0$. Combining these two bounds appropriately will yield
\begin{align}
\label{eq:A3}
\left\|\widetilde{\beta}_n-\beta_0 \right\|_{n,p} = O_{\mathbb{P}}\left( \log^2(n) \left\{n^{-\frac{m}{2m+d}} + \max_{1 \leq j \leq p} \diam^{\kappa}(A_j)  \right\}\right) \quad \text{and} \quad J_{m}\left(\widetilde{\beta}_n  \right) = O_{\mathbb{P}}(1),
\end{align}
under \ref{i} and \ref{ii} of the theorem. The definition of the convex combination $\widetilde{\beta}_n$ will then allow us to derive the same rate of convergence for $\|\widehat{\beta}_n-\beta_0 \|_{n,p}$ as well as the boundedness of $J_m(\widehat{\beta}_n)$.
 
\mypar{Step 1.} We begin by deriving a lower bound on the LHS of \eqref{eq:A2}. By \ref{A4} and Taylor's theorem, for every $\beta \in \mathcal{H}^{m}(\mathbbm{R}^d)$, we find
\begin{align*}
\mathbb{E}_{\epsilon_i}\left\{\rho\left(\epsilon_i+d_i+\langle X_i, \beta_0-\beta\rangle_p\right) \right\} & = g\left(d_i+\langle X_i, \beta_0-\beta\rangle_p \right) 
\\ & = g\left(\langle X_i, \beta_0-\beta\rangle_p \right) + d_i  g^{\prime} \left(\langle X_i, \beta_0-\beta\rangle_p\right) + O(d_i^2)
\\ & = g\left(\langle X_i, \beta_0-\beta\rangle_p \right) + d_i  g^{\prime} \left(0\right) +  O\left(\left| d_i   \langle X_i, \beta_0-\beta\rangle_p \right| \right) + O(d_i^2)
\\ & = g\left(\langle X_i, \beta_0-\beta\rangle_p \right) +  O\left(\left| d_i   \langle X_i, \beta_0-\beta\rangle_p \right| \right) + O(d_i^2)
\end{align*}
where we have used the fact that, by \ref{A4}, $g^{\prime}(0) = 0$ and $\sup_{t \in \mathbbm{R}} |g^{\prime \prime}(t)| \leq B$ for some finite $B>0$. Expanding $\mathbb{E}_{\epsilon_i}\{\rho(\epsilon_i+d_i)\}$  in a similar manner, we see that
\begin{align*}
\mathbb{E}_{\epsilon_i}\{\rho(\epsilon_i+d_i)\} = g(d_i)  = g(0) + d_i g^{\prime}(0) + O(d_i^2) = g(0) + O(d_i^2). 
\end{align*}
Combining these two expansions and averaging, we may deduce the existence of a large enough $C>0$ such that
\begin{equation}\label{eq:inf1}
\begin{aligned}   
 & M(\beta) - M(\beta_0) 
 \\ & = \frac{1}{n} \sum_{i=1}^n \mathbb{E}_{\epsilon_i} \left\{\rho\left(\epsilon_i+d_i+ \langle X_i, \beta_0-\beta\rangle_p \right) - \rho\left(\epsilon_i+d_i \right) \right\}
\\ &   \geq \frac{1}{n} \sum_{i=1}^n \mathbb{E}_{\epsilon_i}\left\{\rho\left(\epsilon_i+\langle X_i, \beta_0-\beta\rangle_p \right) - \rho\left(\epsilon_i\right) \right\} - \frac{C}{n} \sum_{i=1}^n d_i^2 - \frac{C}{n} \sum_{i=1}^n \left| d_i   \langle X_i, \beta_0-\beta\rangle_p \right| 
\\ &   \geq  \frac{1}{n} \sum_{i=1}^n  \mathbb{E}_{\epsilon_i}\left\{\rho\left(\epsilon_i+\langle X_i, \beta_0-\beta\rangle_p \right) -  \rho\left(\epsilon_i\right) \right\} - \frac{C}{n} \sum_{i=1}^n d_i^2 - C \left(\frac{1}{n} \sum_{i=1}^n d_i^2 \right)^{\frac{1}{2}} \left\|\beta-\beta_0\right\|_{n,p},
\end{aligned}
\end{equation}
where to obtain the second inequality we have used the Cauchy-Schwarz inequality. Now, again by the Cauchy-Schwarz inequality,
\begin{align*}
    \max_{1 \leq i \leq n} \left|\langle X_i, \beta - \beta_0 \rangle_p \right| \leq \max_{1 \leq i\leq n} \left\| X_i \right\|_{\infty} \left\|\beta-\beta_0\right\|_{p},
\end{align*}
for every $\beta \in \mathcal{H}^{m}(\mathcal{I})$. Now, by Boole's and Markov's inequalities, for all $u>0$, we have
\begin{align}
\label{eq:max}
    \mathbb{P}\left( \max_{1 \leq i\leq n} \left\| X_i \right\|_{\infty} > u\right) \leq  n \mathbb{P} \left( \left\| X_1 \right\|_{\infty} > u \right) \leq n\mathbb{E}\left\{\exp\left[t\left\|X\right\|_{\infty}^2\right] \right\} e^{-tu^2} \leq D ne^{-tu^2},
\end{align}
for some $D>0$. Setting $u = \sqrt{C \log(n)}$ for sufficiently large $C>0$ we now see that the RHS of \eqref{eq:max} is summable whence, by the Borel-Cantelli lemma, we may conclude that
\begin{align*}
    \max_{1 \leq i\leq n} \left\| X_i \right\|_{\infty} \leq \sqrt{C \log(n)},
\end{align*}
for all large $n$, almost surely. \iffalse
Furthermore, since, by \ref{A3}, $X$ is a mean-zero Gaussian process, for every $u>0$ we have
\begin{align}
\label{eq:max}
    \mathbb{P} \left( \max_{1 \leq i\leq n} \left\| X_i \right\|_{\infty} > u \right)  \leq  n \mathbb{P} \left( \left\| X_1 \right\|_{\infty} > u \right) \leq 2n \mathbb{P} \left( \sup_{\mathbf{t} \in \mathcal{I}} X_1\left(\mathbf{t} \right) > u \right).
\end{align}
Our aim is now to bound the last probability for all large $u$. To that end, we apply \citet[Theorem 4.1.1]{adler2007random} to get that there exists $u_0>0$ such that for all $u>u_0$ we have
\begin{align*}
    \mathbb{P} \left( \sup_{\mathbf{t} \in \mathcal{I}} X_1\left(\mathbf{t} \right) > u \right) \leq C_d u^{\frac{d}{\kappa}+\eta} \exp\left[ - \frac{u^2}{2 C_1} \right],
\end{align*}
for all $\eta>0$ and some constant $C_d>0$ depending only on $d$. Here we have also used \ref{A3} according to which $\sup_{\mathbf{t} \in \mathcal{I}} \mathbb{E}\{|X(\mathbf{t})|^2\} \leq C_1$. Setting $u = \sqrt{C \log(n)}$ for sufficiently large $C>0$ we now see that the RHS of \eqref{eq:max} is summable whence, by the Borel-Cantelli lemma, we may conclude that
\begin{align*}
    \max_{1 \leq i\leq n} \left\| X_i \right\|_{\infty} \leq \sqrt{C \log(n)},
\end{align*}
for all large $n$, almost surely. 
\fi
But, by definition of $J_m(\cdot)$, we now see that
\begin{align}   \label{eq:cs}
\max_{1 \leq i \leq n}\left|\langle X_i, \beta - \beta_0 \rangle_p \right| \leq  \sqrt{C\log(n)} J_m(\beta- \beta_0),
\end{align}
almost surely. By our assumptions, $\beta_0 \in \mathcal{H}^m(\mathbbm{R}^d)$, hence $J_m(\beta_0)$ is bounded. Without loss of generality, we thus take $J_m(\beta_0) \leq 1$. From~\eqref{eq:cs}, it follows that we can find a large enough constant $D_{\varepsilon} \geq 1$ such that
\begin{align*}
\max_{1 \leq i \leq n}\frac{\left|\langle X_i, \beta - \beta_0 \rangle_p \right|}{D_{\varepsilon} \sqrt{\log(n)}\left(1+J_m(\beta)\right)} \leq \varepsilon
\end{align*}
with high probability, where $\varepsilon$ is the constant in \ref{A6}. It follows that for all large $n$, $n \geq n_0$, say, 
\begin{equation}    \label{eq:convexity}
\begin{aligned}
\frac{1}{n} \sum_{i=1}^n  \mathbb{E}_{\epsilon_i}\left\{\rho\left(\epsilon_i+\langle X_i, \beta_0-\beta\rangle_p \right) -  \rho\left(\epsilon_i\right) \right\} &\geq \frac{1}{n} \sum_{i=1}^n  \mathbb{E}_{\epsilon_i}\left\{\rho\left(\epsilon_i+ \frac{\langle X_i, \beta_0-\beta\rangle_p}{D_{\varepsilon} \sqrt{\log(n)}\{1+J_m(\beta)\}} \right) -  \rho\left(\epsilon_i\right) \right\} 
\\ & \geq \varepsilon \frac{\left\|\beta-\beta_0\right\|^2_{n,p}}{D_{\varepsilon}^2 \log(n) \{1+J_m(\beta)\}^2}.
\end{aligned}
\end{equation}
The first inequality in \eqref{eq:convexity} follows from the fact that the function $g(t) = \mathbb{E}_{\epsilon_i} \left\{\rho\left(\epsilon_i+ t \right) \right\}$ is non-decreasing in $|t|$ due to \ref{A5} and the convexity of $\rho$ from~\ref{A4}, and because of our assumption that $D_{\varepsilon} \geq 1$, meaning that also $D_{\varepsilon} \sqrt{\log(n)} (1+J_m(\beta)) \geq 1$ for all large $n$ and $\beta \in \mathcal{H}^m(\mathbbm{R}^d)$. In the second inequality in~\eqref{eq:convexity}, we have used~\ref{A6}, and the definition of the semi-norm $\| \cdot \|_{n,p}$ from~\eqref{eq:np norm}.

The inequality~\eqref{eq:convexity} holds for every fixed $\beta \in \mathcal{H}^{m}\left( \mathbbm{R}^d\right)$, hence what we have shown in~\eqref{eq:inf1} and~\eqref{eq:convexity} is
\begin{align*}
\inf_{\beta \in \mathcal{H}^{m}\left(\mathbbm{R}^d\right)} \left[ \frac{M(\beta) - M(\beta_0) +  \frac{C}{n} \sum_{i=1}^n d_i^2 + C \left(\frac{1}{n} \sum_{i=1}^n d_i^2 \right)^{\frac{1}{2}} \left\|\beta-\beta_0\right\|_{n,p}}{\varepsilon \frac{\left\|\beta-\beta_0\right\|^2_{n,p}}{D_{\varepsilon}^2 
\log(n)\left(1+J_m(\beta)\right)^2}} \right] \geq 1.
\end{align*}
To establish a similar inequality for the random $\widetilde{\beta}_n \in \mathcal{H}^{m}(\mathbbm{R}^d)$ it suffices to note that, by the last inequality, we have
\begin{align}
\label{eq:A4}
M(\widetilde{\beta}_n) - M(\beta_0) & = M(\widetilde{\beta}_n) - M(\beta_0) + \frac{C}{n} \sum_{i=1}^n d_i^2 + C \left(\frac{1}{n} \sum_{i=1}^n d_i^2 \right)^{1/2} \left\|\widetilde{\beta}_n-\beta_0\right\|_{n,p} \nonumber
\\ & \quad - \frac{C}{n} \sum_{i=1}^n d_i^2 - C \left(\frac{1}{n} \sum_{i=1}^n d_i^2 \right)^{1/2} \left\|\widetilde{\beta}_n-\beta_0\right\|_{n,p} \nonumber
\\ & \geq \varepsilon \frac{\left\|\widetilde{\beta}_n-\beta_0\right\|^2_{n,p}}{D_{\varepsilon}^2\log(n)\left(1+J_m(\widetilde{\beta}_n)\right)^2} - \frac{C}{n} \sum_{i=1}^n d_i^2 - C \left(\frac{1}{n} \sum_{i=1}^n d_i^2 \right)^{1/2} \left\|\widetilde{\beta}_n-\beta_0\right\|_{n,p},
\end{align}
which provides the desired lower bound for \eqref{eq:A2}.

\mypar{Step 2.} The next step in our proof is the derivation of an upper bound for the RHS of \eqref{eq:A2} in terms of $\|\widetilde{\beta}_n - \beta_0 \|_{n,p}$. To that end, for each $\beta \in \mathcal{H}^{m}(\mathbbm{R}^d)$ define the independent centered processes
\begin{align*}
U_{i,\beta} = \rho\left(\epsilon_i + d_i + \langle X_i, \beta_0 - \beta \rangle_p \right) - \mathbb{E}_{\epsilon_i} \left\{\rho\left(\epsilon_i + d_i + \langle X_i, \beta_0 - \beta \rangle_p \right)  \right\}, \quad (i = 1, \ldots, n). 
\end{align*}
With this definition we find 
\begin{align*}
\left( M_n(\beta_0) - M(\beta_0) \right) - \left( M_n(\widetilde{\beta}_n) - M(\widetilde{\beta}_n) \right) = \frac{1}{n} \sum_{i=1}^n \left( U_{i,\beta_0} - U_{i,\widetilde{\beta}_n}  \right).
\end{align*}
Notice that, by definition~\eqref{eq:beta tilde} of $\widetilde{\beta}_n$, we have
\begin{align*}
\left\|\widetilde{\beta}_n - \beta_0\right\|_{n,p} = \gamma_n \left\|\widehat{\beta}_n-\beta_0\right\|_{n,p} = \frac{\left\|\widehat{\beta}_n-\beta_0\right\|_{n,p}}{1+\left\|\widehat{\beta}_n-\beta_0\right\|_{n,p}} \leq 1.
\end{align*}
Moreover, by \ref{A1}, $\{\mathbf{t}_{j}\}_{j=1}^p \subset \mathcal{I}$. These two facts imply that we may restrict attention to the subset $\mathcal{B} \subset \mathcal{H}^{m}(\mathcal{I})$ given by $\mathcal{B} =  \{\beta \in \mathcal{H}^{m}(\mathcal{I}): \|\beta - \beta_0\|_{n,p} \leq 1 \}$. We will apply Lemma~\ref{lem:2} to these processes after identifying ${\mathfrak D}(\beta, \widetilde{\beta})$ in~\eqref{eq:frak di} with $\|\beta - \widetilde{\beta}\|_{n,p}$, and ${\mathfrak D}_i(\beta, \widetilde{\beta})$ with $|\langle X_i, \beta - \widetilde{\beta} \rangle_p|$. By the Lipschitz continuity of $\rho$ assumed in \ref{A4}, we have
\begin{align*}
\left| U_{i,\beta} - U_{i,\widetilde{\beta}} \right| \leq 2 C_2 \left| \langle X_i, \beta - \widetilde{\beta} \rangle_p \right|, \quad (i = 1, \ldots, n),
\end{align*}
and the constant random variables $\{ 2\,C_2 \}_{i=1}^n$ are trivially uniformly sub-Gaussian. In order to apply Lemma \ref{lem:2}, we need to verify \eqref{VDG:5} by showing that the entropy integral for the class of functions $\mathcal{B}_M := \{\beta\in \mathcal{H}^{m}(\mathcal{I}):\|\beta - \beta_0\|_{n,p} \leq 1, 1+J_m(\beta) \leq M\}$ for $M \geq 1$ behaves like
\begin{align}
\label{eq:A5}
\int_{0}^{\delta} H^{1/2}(u, \mathcal{B}_M, d) \dd u \lesssim M^{d/(2m)} \delta^{1-d/(2m)}, \quad \delta>0.
\end{align}

To verify the bound in \eqref{eq:A5}, notice that, by our assumptions, for every $\beta \in \mathcal{H}^{m}(\mathcal{I})$, $\|\beta\|_{n,p}$ is up to a constant dominated by $\| \beta \|_p$, with high probability. This is because, by the Cauchy-Schwarz inequality,
\begin{align*}
\left\|\beta \right\|_{n,p}^2 = \frac{1}{n} \sum_{i=1}^n \left|  \sum_{j=1}^p X_i(\mathbf{t}_j)\beta(\mathbf{t}_j) \mu(A_j)  \right|^2 &\leq  \frac{1}{n} \sum_{i=1}^n \left[ \left\{  \sum_{j=1}^p \left|X_i(\mathbf{t}_j) \right|^2 \mu(A_j) \right\} \left\{ \sum_{j=1}^p \left|\beta(\mathbf{t}_j)\right|^2 \mu(A_j)\right\}  \right].
\end{align*}
But, by \ref{A3}, 
\begin{align*}
\mathbb{E}\left\{ \frac{1}{n} \sum_{i=1}^n \sum_{j=1}^p \left|X_i(\mathbf{t}_j) \right|^2 \mu(A_j) \right\} \leq \sup_{\mathbf{t} \in \mathcal{I}} \mathbb{E}\left\{ |X_1(\mathbf{t})|^2 \right\} \mu(\mathcal{I}),
\end{align*}
where we have used the fact that $A_1, \ldots, A_p$ form a partition of $\mathcal{I}$. The finiteness of the RHS implies the existence of a $M_2$ such that $\|\beta\|_{n,p} \leq M_2\|\beta\|_p$ with high probability. It follows that a cover of $\mathcal{B}_M$ in the $\| \cdot \|_p$-semimetric also provides a cover in the $\|\cdot\|_{n,p}$-semimetric. In other words,
\begin{align*}
H(u, \mathcal{B}_M, \left\|\cdot\right\|_{n,p}) \leq H(u/M_2, \mathcal{B}_M, \left\|\cdot\right\|_{p}), \quad u>0,
\end{align*}
with high probability. Now, $1+J_m(\beta) \leq M$ implies both that $1+\|\beta\|_p \leq M$ and that $1+I_m(\beta) \leq M$.  Hence, we obtain the following inclusions:
\begin{align*}
\mathcal{B}_M & \subset \left\{\beta \in \mathcal{H}^m(\mathcal{I}): \left\|\beta-\beta_0\right\|_{n,p} \leq 1, 1+ \left\|\beta\right\|_p \leq M, 1+I_m(\beta) \leq M  \right\}
\\ & \subset \left\{\beta \in \mathcal{H}^m(\mathcal{I}): 1+ \left\|\beta\right\|_p \leq M, 1+I_m(\beta) \leq M  \right\}.
\end{align*}
Therefore,
\begin{align*}
H(u, \mathcal{B}_M, \left\|\cdot\right\|_{p}) \leq H\left(u, \left\{\beta \in \mathcal{H}^m(\mathcal{I}): 1+ \left\|\beta\right\|_p \leq M, 1+I_m(\beta) \leq M  \right\} , \left\|\cdot\right\|_{p}\right), \quad u>0.
\end{align*}
We now bound the latter entropy. Since, for any $\beta \in \mathcal{H}^m(\mathcal{I})$, we have
\begin{align*}
\sum_{j=1}^p |\beta(\mathbf{t}_j)|^2 \mu(A_j) \leq \sup_{\mathbf{t} \in \mathcal{I}} |\beta(\mathbf{t})|^2 \mu(\mathcal{I}),
\end{align*}
it follows that
\begin{align*}
 H&\left(u, \left\{\beta \in \mathcal{H}^m(\mathcal{I}): 1+ \left\|\beta\right\|_p \leq M, 1+I_m(\beta) \leq M  \right\} , \left\|\cdot\right\|_{p} \right) 
  \\ & \quad  \leq  H_{\infty}\left(u/\sqrt{\mu(\mathcal{I})}, \left\{\beta \in \mathcal{H}^m(\mathcal{I}): 1+ \left\|\beta\right\|_p \leq M, 1+I_m(\beta) \leq M  \right\}\right), \quad u>0,
\end{align*}
with $H_{\infty}(u, \mathcal{F})$ denoting the entropy in the supremum norm of a class of functions $\mathcal{F}$. We will now show that there exists a constant $c_0>0$, independent of $M$, such that
\begin{align*}
\left\{\beta \in \mathcal{H}^m(\mathcal{I}): 1+ \left\|\beta\right\|_p \leq M, 1+I_m(\beta) \leq M  \right\} \subset \left\{\beta \in \mathcal{H}^m(\mathcal{I}): \left\|\beta\right\|_{\mathcal{H}^m(\mathcal{I})  }  \leq c_0 M \right\}.
\end{align*}
Indeed, take a $\beta \in \mathcal{H}^m(\mathcal{I})$ such that $1+\|\beta\|_p \leq M$ and  $1+I_m(\beta) \leq M$. By Theorem 3.4 of \citet{Utr:1988}, there exists a global $c_0>0$ such that
\begin{align*}
\int_{\mathcal{I}} |\beta(\mathbf{t})|^2 \dd \mathbf{t} &\leq  \frac{c_0}{p} \sum_{j=1}^p |\beta(\mathbf{t}_j)|^2 + c_0 I_m^2(\beta) \leq c_0(M-1)^2 + c_0(M-1)^2  =  2 c_0(M-1)^2.
\end{align*}
Hence, for every $\beta \in \mathcal{H}^m(\mathcal{I})$, by definition of $\| \cdot\|_{\mathcal{H}^m(\mathcal{I})}$,
\begin{align*}
\left\|\beta\right\|^2_{\mathcal{H}^m(\mathcal{I})  }  \leq \int_{\mathcal{I}} |\beta(\mathbf{t})|^2 \dd \mathbf{t}  + I_m^2(\beta) \leq 2c_0(M-1)^2 + (M-1)^2 \leq \max\{2c_0, 1 \} M^2,
\end{align*}
as $M \geq 1$. Combining all the above, we have shown that 
\begin{align}
\label{eq:A6}
H\left(u, \mathcal{B}_M, \left\|\cdot\right\|_{n,p}\right) \leq H_{\infty}\left(u/(\sqrt{\mu(\mathcal{I})}M_2), \left\{\beta \in \mathcal{H}^m(\mathcal{I}): \left\|\beta\right\|_{\mathcal{H}^m(\mathcal{I})  }  \leq c_0 M \right\} \right), \quad u>0,
\end{align}
with high probability. But the RHS of \eqref{eq:A6} is the supremum entropy of the closed ball in $\mathcal{H}^m(\mathcal{I})$ with radius $c_0M$. Since, by our assumptions $2m>d$, Proposition 6 in \citet{Cucker:2001} now yields
\begin{align*}
H_{\infty}\left(u, \left\{\beta \in \mathcal{H}^m(\mathcal{I}): \left\|\beta\right\|_{\mathcal{H}^m(\mathcal{I})  }  \leq c_0 M \right\} \right) \lesssim M^{\frac{d}{m}} u^{-\frac{d}{m}}, \quad u>0,
\end{align*}
which, upon taking square roots and integrating, yields \eqref{eq:A5}.

Since \eqref{eq:A5} holds,  applying Lemma~\ref{lem:2} next leads to
\begin{align*}
\sup_{\beta \in \mathcal{H}^{m}\left(\mathbbm{R}^d\right): \left\|\beta-\beta_0\right\|_{n,p} \leq 1}  \frac{ \left| \frac{1}{n} \sum_{i=1}^n \left( U_{i,\beta}- U_{i,\beta_0} \right) \right| }{\left\|\beta-\beta_0\right\|_{n,p}^{1-\frac{d}{2m}}\{1+J_m(\beta)\}^{\frac{d}{m}}}  = O_{\mathbb{P}}(n^{-\frac{1}{2}}).
\end{align*}
This inequality provides the desired upper bound for the RHS in \eqref{eq:A2}, namely,
\begin{align}
\label{eq:A7}
\left(M_n(\widetilde{\beta}_n) - M(\widetilde{\beta}_n)\right) - \left( M_n(\beta_0) - M(\beta_0) \right)  = O_\mathbb{P}(n^{-\frac{1}{2}}) \|\widetilde{\beta}_n-\beta_0\|_{n,p}^{1-\frac{d}{2m}} \{1+J_m(\widetilde{\beta}_n)\}^{\frac{d}{2m}}.
\end{align} 

\mypar{Step 3.} The third step of the proof involves combining the lower and upper bounds derived in the previous steps in order to obtain a rate of convergence for $\|\widetilde{\beta}_n-\beta_0\|_{n,p}$. Indeed, plugging \eqref{eq:A4} and \eqref{eq:A7} into \eqref{eq:A2} yields
\begin{align}
\label{eq:A8}
\varepsilon \frac{\left\|\widetilde{\beta}_n-\beta_0\right\|^2_{n,p}}{D_{\varepsilon}^2 \log(n) \left(1+J_m(\widetilde{\beta}_n)\right)^2} + \lambda J_m^2\left(\widetilde{\beta}_n\right) & \leq  O_\mathbb{P}(n^{-\frac{1}{2}}) \|\widetilde{\beta}_n-\beta_0\|_{n,p}^{1-\frac{d}{2m}} \{1+J_m(\widetilde{\beta}_n)\}^{\frac{d}{2m}} + \lambda J_m^2\left(\beta_0 \right) \nonumber
\\ & \quad  +  \frac{C}{n} \sum_{i=1}^n d_i^2 + C \left(\frac{1}{n} \sum_{i=1}^n d_i^2 \right)^{\frac{1}{2}} \left\|\widetilde{\beta}_n-\beta_0\right\|_{n,p},
\end{align}
with high probability. 

We determine next the order of $n^{-1} \sum_{i=1}^n d_ i^2$, for $\{d_i\}_{i=1}^n$ defined in~\eqref{eq:di}. It is easy to see that
\begin{align*}
\left|d_i\right| & \leq \left|\sum_{j=1}^p \int_{A_j} \left(\beta_0(\mathbf{t})-\beta_0(\mathbf{t}_j)\right) X_i(\mathbf{t}) \dd \mathbf{t} \right| +  \left| \sum_{j=1}^p \int_{A_j} \beta_0(\mathbf{t}_j)(X_i(\mathbf{t}) - X_i(\mathbf{t}_j)) \dd \mathbf{t} \right|
\\ & = I + II,
\end{align*}
say. To bound $I$, let us first observe that for $2m>d+1$, by the Sobolev embedding theorem \citep[Theorem 4.12]{Adams:2003}, $\mathcal{H}^{m}(\mathcal{I})$ embeds on the space of Lipschitz continuous functions on $\mathcal{I}$ denoted by $\mathcal{C}^{0,1}(\mathcal{I})$. In other words, for every $\beta \in \mathcal{H}^{m}(\mathcal{I})$ we have $\beta \in \mathcal{C}^{0,1}(\mathcal{I})$ and there exists a global $C>0$ such that
\begin{align*}
\sup_{\mathbf{t}, \mathbf{s} \in \mathcal{I}} \frac{\left|\beta\left(\mathbf{t}\right)-\beta\left(\mathbf{s}\right)\right|}{\left\|\mathbf{t}-\mathbf{s}\right\|_{\mathbbm{R}^d}} & \leq C \left\{ \int_{\mathcal{I}} |\beta\left(\mathbf{t}\right)|^2 \dd \mathbf{t} +  \sum_{m_1 + \ldots + m_d = m} \binom{m}{m_1, \ldots, m_d} \int_{\mathcal{I}} \left( \frac{\partial^m \beta(\mathbf{t})}{\partial t_{1}^{m_1} \ldots \partial t_{d}^{m_d}}  \right)^2 \dd \mathbf{t} \right\}^{\frac{1}{2}}
\\ & = C \left\|\beta\right\|_{\mathcal{H}^m(\mathcal{I})}.
\end{align*}
Recalling that, by assumption, $\beta_0 \in \mathcal{H}^{m}(\mathbbm{R}^d)$, hence also $\beta_0 \in \mathcal{H}^{m}(\mathcal{I})$, hence 
\begin{align*}
I \leq \sum_{j=1}^p \int_{A_j} \left|\left(\beta_0(\mathbf{t})-\beta_0(\mathbf{t}_j)\right) X_i(\mathbf{t}) \right| \dd \mathbf{t}  & \leq C \left\|\beta_0\right\|_{\mathcal{H}^m(\mathcal{I})} \sup_{\mathbf{t} \in \mathcal{I}}\left|X_{i}(\mathbf{t})\right|\sum_{j=1}^p \int_{A_j} \left\|\mathbf{t}-\mathbf{t}_j\right\|_{\mathbbm{R}^d} \dd \mathbf{t} \nonumber
\\ & \leq  C  \left\|\beta_0\right\|_{\mathcal{H}^m(\mathcal{I})} 
 \sup_{\mathbf{t} \in \mathcal{I}}\left|X_{i}(\mathbf{t})\right| \mu(\mathcal{I}) \max_{1 \leq j \leq p} \diam(A_j),
\end{align*}
for some $C>0$ where the last inequality follows from the fact that the sets $A_j$ form a partition of $\mathcal{I}$. Now note that, by \ref{A3}, $\mathbb{E}\{\|X_i\|_{\infty}\} < \infty$, as exponential moments imply the existence of polynomial moments so that
\begin{align}
\label{eq:A9}
    I  = O_{\mathbb{P}} \left(\max_{1 \leq j \leq p} \diam(A_j) \right).
\end{align}
To bound $II$, observe that
\begin{align*}
II \leq \sup_{\mathbf{t} \in \mathcal{I}}\left|\beta_0(\mathbf{t})\right| \sum_{j=1}^p \int_{A_j} \left|X_i\left(\mathbf{t}\right) - X_i\left(\mathbf{t}_j\right) \right|  \dd \mathbf{t}.
\end{align*}
Note that $\sup_{\mathbf{t} \in \mathcal{I}}\left|\beta_0(\mathbf{t})\right|$ is finite, as, by \ref{A1}, $\mathcal{I}$ is bounded and $\beta_0$ is continuous for all $m \geq 1$. By  the Cauchy-Schwarz inequality and \ref{A3}, we have
\begin{align*}
\mathbb{E}\left\{\left|X_1\left(\mathbf{t}\right) - X_1\left(\mathbf{t}_j\right) \right|  \right\} \leq \left[\mathbb{E}\left\{\left|X_1\left(\mathbf{t}\right) - X_1\left(\mathbf{t}_j\right) \right|^2 \right\} \right]^{1/2} \leq \sqrt{C_1} \left\|\mathbf{t}-\mathbf{t}_j \right\|^{\kappa}_{\mathbbm{R}^d},
\end{align*}
which, by Markov's inequality, leads to
\begin{align}
\label{eq:A10}
II = O_\mathbb{P} \left(\max_{1 \leq j \leq p} \diam^{\kappa}(A_j)\right).
\end{align}

As $\kappa \leq 1$, the bounds in \eqref{eq:A9} and \eqref{eq:A10} imply that 
    \begin{equation}    \label{eq: di}
    \frac{1}{n}\sum_{i=1}^n d_i^2 =  O_\mathbb{P} (\max_{1 \leq j \leq p} \diam^{2\kappa}(A_j))    
    \end{equation}
so that
\begin{align}
\label{eq:A11}
\underbrace{\varepsilon \frac{\left\|\widetilde{\beta}_n-\beta_0\right\|^2_{n,p}}{\log(n) D_{\varepsilon}^2\left\{1+J_m\left(\widetilde{\beta}_n\right)\right\}^2} + \lambda J_m^2\left(\widetilde{\beta}_n\right)}_{a} & \leq  \underbrace{O_\mathbb{P}(n^{-\frac{1}{2}}) \|\widetilde{\beta}_n-\beta_0\|_{n,p}^{1-\frac{d}{2m}} \{1+J_m(\widetilde{\beta}_n) \}^{\frac{d}{2m}}}_{b} + \underbrace{\lambda J_m^2\left(\beta_0 \right)}_{c} \nonumber
\\ & \quad  + \underbrace{O_\mathbb{P} \left( \max_{1 \leq j \leq p} \diam^{2 \kappa}(A_j)\right)}_{d} \nonumber
\\ & \quad  + \underbrace{O_\mathbb{P} \left( \max_{1 \leq j \leq p} \diam^{\kappa}(A_j)\right) \left\|\widetilde{\beta}_n-\beta_0\right\|_{n,p}}_{e}.
\end{align}
We next show that the bound in \eqref{eq:A11} implies \eqref{eq:A3}. The former is an inequality of the form $a \leq b+c+d + e $ for real numbers $a,b,c,d,e $ and therefore we must have either $a \leq 4b$ or $a \leq 4c$ or $a \leq 4d$ or $a \leq 4e$ (for if that were not true we would have $a >b+c+d+e$). We consider each one of these four possible cases and show that in all cases \eqref{eq:A3} holds. 

\mypar{Case 1.} Starting from the first possible case, we have $a \leq 4\,b$, that is,
\begin{align*}
\varepsilon \frac{\left\|\widetilde{\beta}_n-\beta_0\right\|^2_{n,p}}{\log(n) D_{\varepsilon}^2\left\{1+J_m\left(\widetilde{\beta}_n\right)\right\}^2} + \lambda J_m^2\left(\widetilde{\beta}_n\right) & \leq  O_\mathbb{P}(n^{-\frac{1}{2}}) \left\|\widetilde{\beta}_n-\beta_0\right\|_{n,p}^{1-\frac{d}{2m}} \left\{1+J_m(\widetilde{\beta}_n)\right\}^{\frac{d}{2m}}.
\end{align*}
Since both terms on the LHS are non-negative, this inequality implies both 
\begin{align}
\label{eq:A12}
\varepsilon \frac{\left\|\widetilde{\beta}_n-\beta_0\right\|^2_{n,p}}{\log(n) D_{\varepsilon}^2\left\{1+J_m\left(\widetilde{\beta}_n\right)\right\}^2} \leq O_\mathbb{P}(n^{-\frac{1}{2}}) \|\widetilde{\beta}_n-\beta_0\|_{n,p}^{1-\frac{d}{2m}} \left\{1+J_m(\widetilde{\beta}_n)\right\}^{\frac{d}{2m}} ,
\end{align}
and
\begin{align}
\label{eq:A13}
 \lambda J_m^2\left(\widetilde{\beta}_n\right) \leq O_\mathbb{P}(n^{-\frac{1}{2}}) \|\widetilde{\beta}_n-\beta_0\|_{n,p}^{1-\frac{d}{2m}} \left\{1+J_m(\widetilde{\beta}_n)\right\}^{\frac{d}{2m}}.
\end{align}
We will use such a separation into two inequalities also in the analysis of the remaining cases that follow. From \eqref{eq:A12} we get
\begin{align}
\label{eq:A14}
\left\|\widetilde{\beta}_n-\beta_0\right\|_{n,p} \leq O_\mathbb{P}\left(\log^{\frac{2m}{2m+d}}(n)n^{-\frac{m}{2m+d}} \right) \left\{1+J_m(\widetilde{\beta}_n)\right\}^{\frac{4m+d}{2m+d}}
\end{align}
which upon plugging into \eqref{eq:A13} yields
\begin{align*}
 \lambda J_m^2\left(\widetilde{\beta}_n\right) \leq O_\mathbb{P}\left(n^{-\frac{2m}{2m+d}} \log^{\frac{2m-d}{2m+d}}(n)\right) \left\{1+J_m(\widetilde{\beta}_n) \right\}^{\frac{4m}{2m+d}},
\end{align*}
or equivalently,
\begin{align}   \label{eq:power}
\frac{J_m^2\left(\widetilde{\beta}_n\right)}{\left\{1+J_m(\widetilde{\beta}_n)\right\}^{\frac{4m}{2m+d}}} \leq O_\mathbb{P}\left(n^{-\frac{2m}{2m+d}} \log^{\frac{2m-d}{2m+d}}(n)\right) \lambda^{-1},
\end{align}
whose RHS is $O_{\mathbb{P}}(1)$ provided that $n^{-2m/(2m+d)} \log^{(2m-d)/(2m+d)}(n)\lambda^{-1} = O_{\mathbb{P}}(1)$. This is ensured by \ref{ii}, as, for all large $n$,
\begin{align*}
    n^{-\frac{2m}{2m+d}} \log^{\frac{2m-d}{2m+d}}(n)\lambda^{-1} \leq n^{-\frac{2m}{2m+d}} \log^{\frac{3}{2}}(n)\lambda^{-1} = O_{\mathbb{P}}(1).
\end{align*}
Now, since the RHS of \eqref{eq:power} is $O_{\mathbb{P}}(1)$ and $2>4m/(2m+d)$, an application of part~\ref{1.} of Lemma~\ref{lem:3} reveals that $J_m(\widetilde{\beta}_n) = O_{\mathbb{P}}(1)$ also. From \eqref{eq:A14} we  now see  that $\|\widetilde{\beta}_n-\beta_0\|_{n,p} = O_{\mathbb{P}}(\log^{2}(n) n^{-m/(2m+d)})$ and consequently \eqref{eq:A3} holds.

\mypar{Case~2.} Focusing now on the second possible case, i.e., on
\begin{align*}
\varepsilon \frac{\left\|\widetilde{\beta}_n-\beta_0\right\|^2_{n,p}}{\log(n) D_{\varepsilon}^2\left\{1+J_m\left(\widetilde{\beta}_n\right)\right\}^2} + \lambda J_m^2\left(\widetilde{\beta}_n\right) \leq 4\lambda J_m^2\left(\beta_0\right),
\end{align*}
we immediately see that $\lambda J_m^2\left(\widetilde{\beta}_n\right) \leq 4\lambda J_m^2\left(\beta_0\right) = O_{\mathbb{P}}(\lambda)$ entails $J_m(\widetilde{\beta}_n) = O_{\mathbb{P}}(1)$, which along with the boundedness of $J_m^2\left(\beta_0\right)$ leads to 
\begin{align*}
    \|\widetilde{\beta}_n-\beta_0\|_{n,p} = O_{\mathbb{P}}(\lambda^{1/2} \log^{1/2}(n)) = \log^2(n) O_{\mathbb{P}}\left( n^{-m/(2m+d)} + \max_{1 \leq j \leq p} \diam^{\kappa}(A_j) \right),
\end{align*}
by our assumption~\ref{i} on $\lambda$. Thus, \eqref{eq:A3} is again obtained. 

\mypar{Case~3.} The third possible case is
\begin{align*}
\varepsilon \frac{\left\|\widetilde{\beta}_n-\beta_0\right\|^2_{n,p}}{D_{\varepsilon}^2 \log(n) \left\{1+J_m\left(\widetilde{\beta}_n\right)\right\}^2} + \lambda J_m^2\left(\widetilde{\beta}_n\right) \leq O_\mathbb{P} \left( \max_{1 \leq j \leq p} \diam^{2 \kappa}(A_j)\right).
\end{align*}
It is clear from this inequality that $J_m(\widetilde{\beta}_n) = O_{\mathbb{P}}(1)$ provided that $\max_{1 \leq j \leq p} \diam^{2 \kappa}(A_j) \lambda^{-1} = O_{\mathbb{P}}(1)$, which is ensured by \ref{ii}, as, for all large $n$,
\begin{align*}
\max_{1 \leq j \leq p} \diam^{2 \kappa}(A_j) \lambda^{-1} \leq \max_{1 \leq j \leq p} \diam^{2 \kappa}(A_j) \log^\frac{3}{2}(n) \lambda^{-1} = O_{\mathbb{P}}(1).
\end{align*}
Therefore, $\|\widetilde{\beta}_n-\beta_0\|_{n,p} = O_\mathbb{P}( \log^{1/2}(n) \max_{1 \leq j \leq p} \diam^{\kappa}(A_j)) =  O_\mathbb{P}( \log^2(n) \max_{1 \leq j \leq p} \diam^{\kappa}(A_j))$ and \eqref{eq:A3} is again established.

\mypar{Case~4.} The fourth and final possible case is
\begin{align}
\label{eq:A15}
\varepsilon \frac{\left\|\widetilde{\beta}_n-\beta_0\right\|^2_{n,p}}{D_{\varepsilon}^2 \log(n) \left\{1+J_m\left(\widetilde{\beta}_n\right)\right\}^2} + \lambda J_m^2\left(\widetilde{\beta}_n\right) & \leq  O_\mathbb{P} \left( \max_{1 \leq j \leq p} \diam^{\kappa}(A_j)\right) \left\|\widetilde{\beta}_n-\beta_0\right\|_{n,p}.
\end{align}
Inequality \eqref{eq:A15} immediately gives
\begin{align}   \label{eq:case4}
\left\|\widetilde{\beta}_n-\beta_0\right\|_{n,p} \leq  O_\mathbb{P} \left( \max_{1 \leq j \leq p} \diam^{\kappa}(A_j) \log(n) \right) \left\{1+J_m\left(\widetilde{\beta}_n\right)\right\}^2,
\end{align}
which, upon plugging this back into the RHS of~\eqref{eq:A15}, leads to
\begin{align*}
\frac{J_m^2\left(\widetilde{\beta}_n\right)}{\left\{1+J_m\left(\widetilde{\beta}_n\right)\right\}^2} \leq O_\mathbb{P} \left( \max_{1 \leq j \leq p} \diam^{2 \kappa}(A_j) \log(n)\right) \lambda^{-1}. 
\end{align*}
By \ref{ii}, the RHS of this inequality will be strictly smaller than $1$ for all large $n$ with high probability, as
\begin{align*}
    \left( \max_{1 \leq j \leq p} \diam^{2 \kappa}(A_j) \log(n)\right) \lambda^{-1} =  \left( \max_{1 \leq j \leq p} \diam^{2 \kappa}(A_j) \log^{3/2}(n)\right) \log^{-\frac{1}{2}}(n) \lambda^{-1} 
 = O_{\mathbb{P}} \left( \log^{-\frac{1}{2}}(n) \right).
\end{align*}
It follows from part~\ref{2.} of Lemma~\ref{lem:3} that $J_m(\widetilde{\beta}_n) = O_{\mathbb{P}}(1)$ and \eqref{eq:case4} leads to
\begin{align*}
\left\|\widetilde{\beta}_n-\beta_0\right\|_{n,p} = O_\mathbb{P} \left( \log^2(n) \max_{1 \leq j \leq p} \diam^{\kappa}(A_j)\right),
\end{align*}
so that \eqref{eq:A3} is verified. We have thus completed the proof of \eqref{eq:A3}. 

\mypar{Step 4.} The last step of our proof involves obtaining the same rate of convergence for $\|\widehat{\beta}_n-\beta_0\|_{n,p}$ and establishing that $J_m(\widehat{\beta}_n) = O_{\mathbb{P}}(1)$. For this, we use the triangle inequality to get
\begin{align}   \label{eq:triangle}
\left\|\widehat{\beta}_n-\beta_0\right\|_{n,p} \leq \left\|\widehat{\beta}_n-\widetilde{\beta}_n\right\|_{n,p} +  \left\|\widetilde{\beta}_n-\beta_0\right\|_{n,p}.
\end{align}
Using the definition of $\widetilde{\beta}_n$ in~\eqref{eq:beta tilde} we can express 
    \[  \left\|\widehat{\beta}_n-\widetilde{\beta}_n\right\|_{n,p} = \left( 1 - \gamma_n \right) \left\|\widehat{\beta}_n-\beta_0 \right\|_{n,p}.  \]
Plugging this into~\eqref{eq:triangle} and using~\eqref{eq:A3} we get
\begin{align*}
\frac{\left\|\widehat{\beta}_n-\beta_0\right\|_{n,p}}{1+\left\|\widehat{\beta}_n-\beta_0\right\|_{n,p}} = O_{\mathbb{P}}\left( \log^2(n) \left\{n^{-
\frac{m}{2m+d}} +  \max_{1 \leq j \leq p} \diam^{\kappa}(A_j)  \right\}\right).
\end{align*}
Solving for $\|\widehat{\beta}_n-\beta_0\|_{n,p}$ yields 
\begin{align*}
\left\|\widehat{\beta}_n-\beta_0\right\|_{n,p}(1-o_{\mathbb{P}}(1))  =  O_{\mathbb{P}}\left(\log^2(n) \left\{n^{-\frac{m}{2m+d}} +  \max_{1 \leq j \leq p} \diam^{\kappa}(A_j) \right\}  \right)
\end{align*}
so that 
\begin{align}   \label{eq:final beta bound}
\left\|\widehat{\beta}_n-\beta_0\right\|_{n,p}  = O_{\mathbb{P}}\left(\log^2(n) \left\{n^{-\frac{m}{2m+d}} +  \max_{1 \leq j \leq p} \diam^{\kappa}(A_j) \right\}  \right),
\end{align}
which is the first result of the theorem. To conclude the proof we deduce the boundedness of $J_m(\widehat{\beta}_n)$ from the boundedness of $J_m(\widetilde{\beta}_n)$. For this recall that $J_m(\cdot)$ is a semi-norm on $\mathcal{H}^{m}(\mathbbm{R}^d)$ so that it is homogeneous and satisfies the triangle inequality. Moreover, by construction, $\gamma_n \leq 1$. Therefore, with the help of~\eqref{eq:beta tilde} and \eqref{eq:A3} we obtain
\begin{align*}
\gamma_n\,J_m(\widehat{\beta}_n-\beta_0) = J_m(\gamma_n(\widehat{\beta}_n-\beta_0)) \leq J_m(\widetilde{\beta}_n) + J_m(\beta_0) = O_{\mathbb{P}}(1).
\end{align*}
The boundedness of $\|\widehat{\beta}_n-\beta_0\|_{n,p}$ from~\eqref{eq:final beta bound} now implies that $\gamma_n^{-1} = O_{\mathbb{P}}(1)$ so that $J_m(\widehat{\beta}_n-\beta_0) = O_{\mathbb{P}}(1)$ also. It follows by another application of the triangle inequality that
\begin{align*}
J_m\left(\widehat \beta_n\right) \leq J_m\left(\beta_0\right) + J_m\left(\widehat \beta_n - \beta_0 \right) = O_{\mathbb{P}}(1),
\end{align*}
which is the desired result.

%\end{proof}

%\subsection{Proof of Corollary~\ref{cor:2}}

%The proof is a straightforward adaptation of the proof of Theorem~\ref{thm:1} where in \ref{A6'} allows us to bound $\max_{ i\leq n} |\langle X_i, \beta - \beta_0 \rangle_p|$ during \textbf{\textit{Step 1}} as follows
%\begin{align*}
%\max_{1 \leq i \leq n} \left|\langle X_i, \beta - \beta_0 \rangle_p\right| \leq C_2 \left\|\beta-\beta_0 \right\|_p \leq C_2 J_m(\beta-\beta_0).
%\end{align*}
%The result can then be argued in a similar manner.

\subsection{Proof of Corollary~\ref{cor:2}}

Observe that, since the $\{A_j\}_{j=1}^p$ partition $\mathcal{I}$, we can write
\begin{align*}
\left\|\widehat{\beta}_n - \beta_0 \right\|_n^2 = \frac{1}{n} \sum_{i=1}^n\left| \int_{\mathcal{I}} X_i(\mathbf{t})\left(\widehat{\beta}_n(\mathbf{t}) - \beta_0(\mathbf{t})\right) \dd \mathbf{t} \right|^2 = \frac{1}{n} \sum_{i=1}^n\left| \sum_{j=1}^p \int_{A_j} X_i(\mathbf{t})\left(\widehat{\beta}_n(\mathbf{t}) - \beta_0(\mathbf{t})\right) \dd \mathbf{t} \right|^2,
\end{align*}
so that, using the inequality $|x+y|^2 \leq 2(|x|^2+|y|^2)$ twice, the definition of the $d_i$ in \eqref{eq:di} as well as the definition of the semi-norm $\| \cdot \|_{n,p}$ in \eqref{eq:np norm}, we get
\begin{align}
\label{eq:c1}
\left\|\widehat{\beta}_n - \beta_0 \right\|_n^2 & \leq \frac{2}{n} \nonumber \sum_{i=1}^n \left| \sum_{j=1}^p \int_{A_j} \left\{ X_i(\mathbf{t}) \left(\widehat{\beta}_n(\mathbf{t}) - \beta_0(\mathbf{t})\right) -  X_i(\mathbf{t}_j) \left(\widehat{\beta}_n(\mathbf{t}_j) - \beta_0(\mathbf{t}_j)\right) \right\} \dd \mathbf{t} \right|^2
\\ & \quad + \frac{2}{n} \sum_{i=1}^n \left| \sum_{j=1}^p  \nonumber X_i(\mathbf{t}_j)\left(\widehat \beta_n(\mathbf{t}_j) - \beta_0(\mathbf{t}_j)\right) \mu(A_j) \right|^2
\\ &  \leq \frac{4}{n} \sum_{i=1}^n \left|\sum_{j=1}^p \int_{A_j} \{X_i(\mathbf{t})\widehat\beta_n(\mathbf{t}) - \nonumber X_i(\mathbf{t}_j)\widehat\beta_n(\mathbf{t}_j) \} \dd \mathbf{t} \right|^2 
\\ & \quad + \frac{4}{n} \sum_{i=1}^n \left|\sum_{j=1}^p \int_{A_j} \{X_i(\mathbf{t})\beta_0(\mathbf{t}) -X_i(\mathbf{t}_j)\beta_0(\mathbf{t}_j) \} \dd \mathbf{t} \right|^2  \nonumber
\\ & \quad + \frac{2}{n} \sum_{i=1}^n \left| \sum_{j=1}^p X_i(\mathbf{t}_j)\left(\widehat \beta_n(\mathbf{t}_j) - \beta_0(\mathbf{t}_j)\right) \mu(A_j) \right|^2 \nonumber
\\ & = \frac{4}{n} \sum_{i=1}^n \left|\sum_{j=1}^p \int_{A_j} \{X_i(\mathbf{t})\widehat\beta_n(\mathbf{t}) -X_i(\mathbf{t}_j)\widehat\beta_n(\mathbf{t}_j) \} \dd \mathbf{t} \right|^2 + \frac{4}{n} \sum_{i=1}^n d_i^2 + 2\left\|\widehat\beta_n - \beta_0\right\|_{n,p}^2 .
\end{align}
A rate of convergence for $\|\widehat \beta_n - \beta_0\|_{n,p}$ was obtained in Theorem~\ref{thm:1}. From equation~\eqref{eq: di} in the proof of that theorem we also know that $n^{-1}\sum_{i=1}^n d_i^2 = O_{\mathbb{P}}(\max_{j \leq p} \diam^{2\kappa}(A_j))$. Hence, to prove the corollary we only need to determine the order of the first term in the RHS of \eqref{eq:c1}.

To bound the first term in \eqref{eq:c1}, write
\begin{align}
\label{eq:c2}
\left|\sum_{j=1}^p \int_{A_j} \{X_i(\mathbf{t})\widehat\beta_n(\mathbf{t}) -X_i(\mathbf{t}_j)\widehat\beta_n(\mathbf{t}_j) \} \dd \mathbf{t} \right| & \leq \left|\sum_{j=1}^p \int_{A_j} \left(\widehat \beta_n(\mathbf{t})-\widehat \beta_n(\mathbf{t}_j)\right) X_i(\mathbf{t}) \dd \mathbf{t} \right| \nonumber
\\ & \quad +  \left| \sum_{j=1}^p \int_{A_j}  \widehat \beta_n(\mathbf{t}_j)(X_i(\mathbf{t}) - X_i(\mathbf{t}_j)) \dd \mathbf{t} \right| \nonumber
\\ & = I + II,
\end{align}
say. To bound $I$ in \eqref{eq:c2}, let us note that since, by assumption $2m>d+1$, $\mathcal{H}^m(\mathcal{I})$ embeds on the space of Lipschitz continuous functions on $\mathcal{I}$, which we previously denoted by $\mathcal{C}^{0,1}(\mathcal{I})$. Since, as shown in Proposition~\ref{Prop:1}, $\widehat\beta_n \in \mathcal{H}^m(\mathcal{I})$, we therefore have
\begin{align*}
\sup_{\mathbf{t}, \mathbf{s} \in \mathcal{I}} \frac{\left|\widehat\beta_n(\mathbf{t}) - \widehat \beta_n(\mathbf{s})\right|}{\left\|\mathbf{t}-\mathbf{s} \right\|_{\mathbbm{R}^d }} & \leq   C \left\{\int_{\mathcal{I}} \left|\widehat \beta_n(\mathbf{t})\right|^2 \dd \mathbf{t} +  \sum_{m_1 + \ldots + m_d = m} \binom{m}{m_1, \ldots, m_d} \int_{\mathcal{I}} \left( \frac{\partial^m \widehat\beta_n(\mathbf{t})}{\partial t_{1}^{m_1} \ldots \partial t_{d}^{m_d}}  \right)^2 \dd \mathbf{t}\right\}^{1/2}
\\ & \leq C \left\{\int_{\mathcal{I}} \left|\widehat \beta_n(\mathbf{t})\right|^2 \dd \mathbf{t} +  J_m^2(\widehat\beta_n)\right\}^{1/2},
\end{align*}
for some global $C>0$. By Theorem~\ref{thm:1}, $J_m(\widehat\beta_n) = O_{\mathbb{P}}(1)$. At the same time, by Theorem 3.4 of \citet{Utr:1988}, there exist $C_0$ and $B_1$, depending only on $m, d, \mathcal{I}$ such that
\begin{align*}
    \int_{\mathcal{I}} \left|\widehat \beta_n(\mathbf{t})\right|^2 \dd \mathbf{t} \leq C_0 \sum_{j=1}^p \left|\widehat\beta_n(\mathbf{t}_j) \right|^2 \mu(A_j) + C_0 I_m^2(\widehat\beta_n) \leq C_0 J_m^2\left(\widehat \beta_n\right).
\end{align*}
By the boundedness of $J_m(\widehat{\beta}_n)$, which also implies the boundedness of $\sum_{j=1}^p |\widehat\beta_n(\mathbf{t}_j) |^2 \mu(A_j)$, it now follows that
\begin{align*}
    \sup_{\mathbf{t}, \mathbf{s} \in \mathcal{I}} \frac{\left|\widehat\beta_n(\mathbf{t}) - \widehat \beta_n(\mathbf{s})\right|}{\left\|\mathbf{t}-\mathbf{s} \right\|_{\mathbbm{R}^d }} = O_{\mathbb{P}}(1),
\end{align*}
from which we obtain
\begin{align*}
    I \leq O_{\mathbb{P}}(1) \sup_{\mathbf{t} \in \mathcal{I}} \left|X_i(\mathbf{t}) \right| \sum_{j=1}^p \int_{A_j}\left\|\mathbf{t}-\mathbf{t}_j\right\|_{\mathbbm{R}^d}  \dd \mathbf{t}.
\end{align*}
Under \ref{A3}, $\sup_{\mathbf{t} \in \mathcal{I}} |X_i(\mathbf{t})| = O_{\mathbb{P}}(1)$. Hence, we may conclude that 
\begin{align}
\label{eq:c3}
    I = O_{\mathbb{P}}\left( \max_{1 \leq j \leq p} \diam(A_j) \right).
\end{align}

To bound $II$ in \eqref{eq:c2}, let us note that for $2m>d$, $\mathcal{H}^{m}(\mathcal{I})$ also embeds on the space $\mathcal{C}(\mathcal{I})$ of continuous functions on $\mathcal{I}$, hence
\begin{align*}
\sup_{\mathbf{t} \in \mathcal{I}} \left| \widehat\beta_n(\mathbf{t}) \right| \leq C \left\{\int_{\mathcal{I}} \left|\widehat \beta_n(\mathbf{t})\right|^2 \dd \mathbf{t} +  J_m^2(\widehat\beta_n)\right\}^{1/2},
\end{align*}
for some $C>0$. Arguing as with the bound for $I$, we can deduce that $\sup_{\mathbf{t} \in \mathcal{I}} |\widehat\beta_n(\mathbf{t}) | = O_{\mathbb{P}}(1)$. Therefore,
\begin{align*}
    II \leq \sup_{\mathbf{t}  \in \mathcal{I}} \left|\widehat\beta_n(\mathbf{t}) \right| \sum_{j=1}^p \int_{A_j} \left|X_i(\mathbf{t}) - X_i(\mathbf{t}_j) \right| \dd \mathbf{t},
\end{align*}
and \ref{A3} yields 
\begin{align}
\label{eq:c4}
    II = O_{\mathbb{P}}\left( \max_{1 \leq j \leq p} \diam^{\kappa}(A_j) \right).
\end{align}
Combining \eqref{eq:c1}, \eqref{eq:c2}, \eqref{eq:c3} and \eqref{eq:c4} with Theorem~\ref{thm:1} now leads to the result.

\subsection{Proof of Theorem~\ref{thm:2}}

The proof of Theorem~\ref{thm:2} relies on a suitable decomposition of the empirical process along with its tightness as a function of $\sigma$ over $|\sigma-\sigma_0|\leq \varepsilon$ for some $\varepsilon>0$. The crucial elements for our proof are established in Lemma~\ref{lem:4} below. For its statement we adopt our notation from the proof of Theorem~\ref{thm:1}.

\begin{lemma}
\label{lem:4}
For some fixed $(\sigma_0, \beta_0) \in \mathbbm{R} \times \mathcal{H}^{m}(\mathbbm{R}^d)$, consider the independent processes $V_{i,\beta,\sigma}$ and $Z_{i,\beta,\sigma}$ given by
\begin{align*}
    V_{i, \beta,\sigma} & = \int_{d_i}^{d_i+\langle X_i,\beta_0-\beta\rangle_p} \left\{\psi \left(\frac{\epsilon_i+t}{\sigma} \right) - \psi \left(\frac{\epsilon_i}{\sigma} \right) \right\} \dd t 
    \\ & \quad - \mathbb{E}_{\epsilon_i} \left\{ \int_{d_i}^{d_i+\langle X_i,\beta_0-\beta\rangle_p} \left\{\psi \left(\frac{\epsilon_i+t}{\sigma} \right) - \psi \left(\frac{\epsilon_i}{\sigma} \right) \right\} \dd t \right\}, \quad (i = 1\ldots, n),
    \\  Z_{i,\beta,\sigma} & = \left(\psi\left(\frac{\epsilon_i}{\sigma} \right) - \mathbb{E}_{\epsilon_i}\left\{\psi\left(\frac{\epsilon_i}{\sigma} \right) \right\} \right)  \langle X_i, \beta_0-\beta\rangle_p,  \quad (i = 1\ldots, n).
\end{align*}
Here, both $V_{i,\beta, \sigma}$ and $Z_{i,\beta,\sigma}$ are indexed by $\beta \in \mathcal{H}^m(\mathbbm{R}^d)$ with $2m>d$ and $\sigma \in [\sigma_0-\varepsilon, \sigma_0+\varepsilon]$ for some small $\varepsilon>0$ such that $\sigma_0-\varepsilon>0$.  Then, the following asymptotic results hold.
\begin{enumerate} [label=\Alph*., ref=\Alph*]
\item \label{AA.} $\begin{aligned}[t]
\sup_{\substack{\beta \in \mathcal{H}^m(\mathbbm{R}^d): \left\|\beta-\beta_0\right\|_{n,p} \leq 1\\ |\sigma-\sigma_0|\leq \varepsilon }} \frac{\left|\frac{1}{\sqrt{n}} \sum_{i=1}^n V_{i,\beta,\sigma} \right|}{\left\|\beta-\beta_0\right\|_{n,p}^{1-\frac{d}{2m}}\left\{1+J_m(\beta)\right\}^{\frac{d}{2m}}} = O_{\mathbb{P}} \left(\sqrt{\log(n)} \right).
\end{aligned}$
\item \label{BB.} $\begin{aligned}[t]
\sup_{\substack{\beta \in \mathcal{H}^m(\mathbbm{R}^d): \left\|\beta-\beta_0\right\|_{n,p} \leq 1\\ |\sigma-\sigma_0|\leq \varepsilon }} \frac{\left|\frac{1}{\sqrt{n}} \sum_{i=1}^n Z_{i,\beta,\sigma} \right|}{\left\|\beta-\beta_0\right\|_{n,p}^{1-\frac{d}{2m}}\left\{1+J_m(\beta)\right\}^{\frac{d}{2m}}} = O_{\mathbb{P}} \left(\sqrt{\log(n)} \right).
\end{aligned}$
\end{enumerate}
\end{lemma}

\begin{proof}[Proof of Lemma~\ref{lem:4}]

We begin by proving part~\ref{AA.}. Consider first the case of a fixed $\sigma \in [\sigma_0-\varepsilon, \sigma_0+\varepsilon]$ with $\varepsilon$ small enough so that $\sigma_0-\varepsilon>0$. Observe that in this case Lemma~\ref{lem:2} applies, as each $V_{i,\beta,\sigma}$ has mean zero, $V_{i, \beta_0, \sigma} = 0$ and, for every $\beta, \beta^{\prime} \in \mathcal{H}^m(\mathbbm{R}^d)$, by \ref{B4} and the fundamental theorem of calculus, we find that
\begin{align*}
& \left|\int_{d_i}^{d_i+\langle X_i,\beta_0-\beta\rangle_p} \left\{\psi \left(\frac{\epsilon_i+t}{\sigma} \right) - \psi \left(\frac{\epsilon_i}{\sigma} \right) \right\} \dd t - \int_{d_i}^{d_i+\langle X_i,\beta_0-\beta^{\prime}\rangle_p} \left\{\psi \left(\frac{\epsilon_i+t}{\sigma} \right) - \psi \left(\frac{\epsilon_i}{\sigma} \right) \right\} \dd t \right|
\\ & = \left| \int_{d_i+\langle X_i,\beta_0-\beta^{\prime}\rangle_p}^{d_i+\langle X_i,\beta_0-\beta\rangle_p} \left\{\psi \left(\frac{\epsilon_i+t}{\sigma} \right) - \psi \left(\frac{\epsilon_i}{\sigma} \right) \right\} \dd t\right| \leq 2 \left\|\psi \right\|_{\infty} \left|\langle X_i, \beta - \beta^{\prime} \rangle_p \right|,
\end{align*} 
and constants are trivially uniformly sub-Gaussian. Furthermore, as we demonstrate in the proof of Theorem~\ref{thm:1} (see \textbf{Step 2} there),
\begin{align*}
    H\left(\delta,\{\beta\in \mathcal{H}^{m}(\mathcal{I}):\|\beta - \beta_0\|_{n,p} \leq 1, 1+J_m(\beta) \leq M\}, \|\cdot\|_{n,p} \right) \lesssim M^{\frac{d}{2m}} \delta^{1-\frac{d}{2m}}, \quad \delta>0, M \geq 1.
\end{align*}
Hence, applying Lemma~\ref{lem:2} yields the existence of a $c>0$ such that for all $T \geq c$
\begin{align}
\label{eq:exp}
\mathbb{P} \left( \sup_{\beta \in \mathcal{H}^m(\mathbbm{R}^d): \left\|\beta-\beta_0\right\|_{n,p} \leq 1} \frac{\left|\frac{1}{\sqrt{n}} \sum_{i=1}^n V_{i,\beta,\sigma} \right|}{\left\|\beta-\beta_0\right\|_{n,p}^{1-\frac{d}{2m}}\left\{1+J_m(\beta)\right\}^{\frac{d}{2m}}} \geq T \right) \leq c \exp \left[ -\frac{T^2}{c^2} \right].
\end{align}
This result is valid for each $\sigma \in [\sigma_0-\varepsilon,\sigma_0+\varepsilon]$, but it is not uniform. To make the result uniform for $\sigma \in [\sigma_0-\varepsilon,\sigma_0+\varepsilon]$, notice first that condition \ref{B4} ensures that for $\varepsilon_1 = 1/(\sigma_0+\varepsilon)$ there exists a $D_{\sigma_0}$ such that for all $s \geq \varepsilon_1, t \geq\varepsilon_1$ we have
\begin{align}
\label{eq:tail}
    \sup_{x \in \mathbbm{R}}\left|\psi(tx)-\psi(sx) \right| \leq D_{\sigma_0}\left|t-s\right|.
\end{align}
Split $ [\sigma_0-\varepsilon,\sigma_0+\varepsilon]$ into $N$ subintervals each of them having radius $\leq T/(2\sqrt{n}D_{\sigma_0})$ and select  $\{\sigma_k\}_{k=1}^N$, one $\sigma_k$ in each one of these $N$ subintervals. By Lemma 2.5 in \citet{vandeGeer:2000}, we can have
\begin{align}
\label{eq:coverball}
N \leq \left( \frac{8\sqrt{n} D_{\sigma_0}\varepsilon}{T}+1 \right).
\end{align}
Observe next that
\begin{align*}
& \sup_{\substack{\beta \in \mathcal{H}^m(\mathbbm{R}^d): \left\|\beta-\beta_0\right\|_{n,p} \leq 1 \\ |\sigma-\sigma_0|\leq \varepsilon}}  \frac{\left|\frac{1}{\sqrt{n}} \sum_{i=1}^n V_{i,\beta,\sigma} \right|}{\left\|\beta-\beta_0\right\|_{n,p}^{1-\frac{d}{2m}}\left\{1+J_m(\beta)\right\}^{\frac{d}{2m}}} 
\\ &  \leq \max_{1 \leq k \leq N}\sup_{\beta \in \mathcal{H}^m(\mathbbm{R}^d): \left\|\beta-\beta_0\right\|_{n,p} \leq 1} \frac{\left|\frac{1}{\sqrt{n}} \sum_{i=1}^n V_{i,\beta,\sigma_k} \right|}{\left\|\beta-\beta_0\right\|_{n,p}^{1-\frac{d}{2m}}\left\{1+J_m(\beta)\right\}^{\frac{d}{2m}}}
\\ & \quad  + \sup_{|\sigma-\sigma_0|\leq \varepsilon} \min_{1 \leq k \leq N} \sup_{\beta \in \mathcal{H}^m(\mathbbm{R}^d): \left\|\beta-\beta_0\right\|_{n,p} \leq 1}  \frac{\left|\frac{1}{\sqrt{n}} \sum_{i=1}^n \{V_{i,\beta,\sigma}-V_{i,\beta,\sigma_k} \}\right|}{\left\|\beta-\beta_0\right\|_{n,p}^{1-\frac{d}{2m}}\left\{1+J_m(\beta)\right\}^{\frac{d}{2m}}},
\end{align*}
where to obtain the second term on the RHS of this inequality we have used the inequality $|\sup_{x}f(x) - \sup_{x}g(x)| \leq \sup_{x} |f(x) - g(x)|$ for real valued functions $f$ and $g$. Using \eqref{eq:tail}, for every $\sigma$ satisfying $|\sigma-\sigma_k| \leq T/(2\sqrt{n}D_{\sigma_0})$ for some $k \in \{1, \ldots, N\}$, we now see that
\begin{align*}
\left|\frac{1}{\sqrt{n}} \sum_{i=1}^n \{V_{i,\beta,\sigma}-V_{i,\beta,\sigma_k}\} \right| & \leq \left| \frac{1}{\sqrt{n}} \sum_{i=1}^n \int_{d_i}^{d_i+\langle X_i,\beta_0-\beta\rangle_p} \left\{\psi \left(\frac{\epsilon_i+t}{\sigma} \right) - \psi \left(\frac{\epsilon_i+t}{\sigma_k} \right) \right\} \dd t  \right| 
 \\ & \quad + \left| \frac{1}{\sqrt{n}} \sum_{i=1}^n \int_{d_i}^{d_i+\langle X_i,\beta_0-\beta\rangle_p} \left\{\psi \left(\frac{\epsilon_i}{\sigma} \right) - \psi \left(\frac{\epsilon_i}{\sigma_k} \right) \right\} \dd t\right|
\\ & \leq 2D_{\sigma_0} \sqrt{n}\left|  \sigma-\sigma_k \right| \frac{1}{n} \sum_{i=1}^n \left|\langle X_i,\beta-\beta_0 \rangle_p \right|
\\ & \leq T\left\|\beta-\beta_0\right\|_{n,p},
\end{align*}
where to obtain the last inequality we have also used the Cauchy-Schwarz inequality and the definition of $\| \cdot \|_{n,p}$ in \eqref{eq:npd norm}. Since $J_m(\beta) \geq 0$ for all $\beta \in \mathcal{H}^m(\mathbbm{R}^d)$ and $d/(2m)<1$, the quotient may be bounded by
\begin{align*}
     \sup_{|\sigma-\sigma_0|\leq \varepsilon} & \min_{1 \leq k \leq N} \sup_{\beta \in \mathcal{H}^m(\mathbbm{R}^d): \left\|\beta-\beta_0\right\|_{n,p} \leq 1}  \frac{\left|\frac{1}{\sqrt{n}} \sum_{i=1}^n \{V_{i,\beta,\sigma}-V_{i,\beta,\sigma_k} \}\right|}{\left\|\beta-\beta_0\right\|_{n,p}^{1-\frac{d}{2m}}\left\{1+J_m(\beta)\right\}^{\frac{d}{2m}}}
    \\ & \quad \leq T \sup_{\beta \in \mathcal{H}^m(\mathbbm{R}^d): \left\|\beta-\beta_0\right\|_{n,p} \leq 1} \frac{\|\beta - \beta_0\|_{n,p}}{\left\|\beta-\beta_0\right\|_{n,p}^{1-\frac{d}{2m}}\left\{1+J_m(\beta)\right\}^{\frac{d}{2m}}}
    \\ & \quad \leq T \sup_{\beta \in \mathcal{H}^m(\mathbbm{R}^d): \left\|\beta-\beta_0\right\|_{n,p} \leq 1} \frac{\|\beta - \beta_0\|_{n,p}^{\frac{d}{2m}}}{\left\{1+J_m(\beta)\right\}^{\frac{d}{2m}}}
    \\ & \quad \leq T.
\end{align*}
From this, Boole's inequality, \eqref{eq:exp} and \eqref{eq:coverball} we find
\begin{align*}
&\mathbb{P} \left( \sup_{\substack{\beta \in \mathcal{H}^m(\mathbbm{R}^d): \left\|\beta-\beta_0\right\|_{n,p} \leq 1\\ |\sigma-\sigma_0|\leq \varepsilon }} \frac{\left|\frac{1}{\sqrt{n}} \sum_{i=1}^n V_{i,\beta,\sigma} \right|}{\left\|\beta-\beta_0\right\|_{n,p}^{1-\frac{d}{2m}}\left\{1+J_m(\beta)\right\}^{\frac{d}{2m}}} \geq 2T  \right)
\\ & \quad \leq \mathbb{P}\left(\max_{1 \leq k \leq N}\sup_{\beta \in \mathcal{H}^m(\mathbbm{R}^d): \left\|\beta-\beta_0\right\|_{n,p} \leq 1} \frac{\left|\frac{1}{\sqrt{n}} \sum_{i=1}^n V_{i,\beta,\sigma_k} \right|}{\left\|\beta-\beta_0\right\|_{n,p}^{1-\frac{d}{2m}}\left\{1+J_m(\beta)\right\}^{\frac{d}{2m}}} \geq T \right)
\\ & \quad = \mathbb{P}\left(\bigcup_{k=1}^N \left\{ \sup_{\beta \in \mathcal{H}^m(\mathbbm{R}^d): \left\|\beta-\beta_0\right\|_{n,p} \leq 1} \frac{\left|\frac{1}{\sqrt{n}} \sum_{i=1}^n V_{i,\beta,\sigma_k} \right|}{\left\|\beta-\beta_0\right\|_{n,p}^{1-\frac{d}{2m}}\left\{1+J_m(\beta)\right\}^{\frac{d}{2m}}} \geq T \right\}\right)
\\ & \quad \leq \sum_{k=1}^N \mathbb{P}\left( \sup_{\beta \in \mathcal{H}^m(\mathbbm{R}^d): \left\|\beta-\beta_0\right\|_{n,p} \leq 1} \frac{\left|\frac{1}{\sqrt{n}} \sum_{i=1}^n V_{i,\beta,\sigma_k} \right|}{\left\|\beta-\beta_0\right\|_{n,p}^{1-\frac{d}{2m}}\left\{1+J_m(\beta)\right\}^{\frac{d}{2m}}} \geq T \right)
\\ & \quad \leq \left( \frac{8\sqrt{n} D_{\sigma_0}\varepsilon}{T}+1 \right) c \exp \left[ -\frac{T^2}{c^2} \right].
\end{align*}
The result follows by setting $T = \sqrt{C \log(n)}$ for a sufficiently large $C>0$. To prove part~\ref{BB.}, use exactly the same arguments given that each $Z_{i,\beta,\sigma}$ is a centered process with $Z_{i,\beta_0,\sigma} = 0$.
\end{proof}

We now move on to the proof of Theorem~\ref{thm:2}. We only sketch the proof, as most of the arguments from the proof of Theorem~\ref{thm:1} carry over to this case. Let $L_n(\beta, \widehat{\sigma}_n)$ denote the objective function, that is,
\begin{align*}
L_n(\beta, \widehat{\sigma}_n) = \frac{1}{n} \sum_{i = 1}^n \rho\left(\frac{Y_i - \sum_{j=1}^p X_i \left( \mathbf{t}_j \right) \beta\left(\mathbf{t}_j\right)\mu\left(A_j\right)}{\widehat{\sigma}_n} \right) + \lambda J_m^2\left(\beta \right).
\end{align*}
Defining the convex combination $\widetilde{\beta}_n  = \gamma_n \widehat{\beta}_n + (1-\gamma_n)\beta_0$ with $\gamma_n = 1/(1+\|\widehat{\beta}_n-\beta_0\|_{n,p})$, we have 
\begin{align}
\label{eq:B1}
L_n(\widetilde{\beta}_n, \widehat{\sigma}_n) \leq \gamma_n L_n(\widehat{\beta}_n, \widehat{\sigma}_n) + (1-\gamma_n) L_n(\beta_0, \widehat{\sigma}_n) \leq L_n(\beta_0, \widehat{\sigma}_n).
\end{align}
Set
\begin{align*}
M_n(\beta, \widehat{\sigma}_n) := \frac{1}{n} \sum_{i=1}^n \rho\left(\frac{\epsilon_i + d_i + \sum_{j=1}^p X_i \left( \mathbf{t}_j\right) \left( \beta_0\left(\mathbf{t}_j\right) - \beta \left(\mathbf{t}_j \right) \right) \mu\left(A_j \right) }{\widehat{\sigma}_n} \right),
\end{align*}
so that $L_n(\beta, \widehat{\sigma}_n) = M_n(\beta, \widehat{\sigma}_n) + \lambda J^2_m(\beta)$. Furthermore, set
\begin{align*}
M(\beta, \widehat{\sigma}_n) := \frac{1}{n} \sum_{i=1}^n \mathbb{E}_{\epsilon_i} \left\{ \rho\left( \frac{\epsilon_i + d_i +  \sum_{j=1}^p X_i \left( \mathbf{t}_j\right)\left( \beta_0\left(\mathbf{t}_j\right) - \beta \left(\mathbf{t}_j \right) \right)\mu\left(A_j \right)}{\widehat{\sigma}_n}\right) \right\}.
\end{align*}
Rearranging \eqref{eq:B1} yields
\begin{align}
\label{eq:B2}
M(\widetilde{\beta}_n, \widehat{\sigma}_n) - M(\beta_0, \widehat{\sigma}_n) + \lambda J_m^2(\widetilde{\beta}_n) & \leq \left( M_n(\beta_0, \widehat{\sigma}_n) - M(\beta_0, \widehat{\sigma}_n) \right) \nonumber \\ &  - \left( M_n(\widetilde{\beta}_n, \widehat{\sigma}_n) - M(\widetilde{\beta}_n, \widehat{\sigma}_n) \right) +\lambda J_m^2(\beta_0).
\end{align}
As in the proof of Theorem~\ref{thm:1}, our approach consists of deriving a lower bound on the LHS of \eqref{eq:B2} in terms of $\| \widetilde{\beta}_n-\beta_0 \|_{n,p}$ and an upper bound on the RHS, also in terms of $\|\widetilde{\beta}_n-\beta_0 \|_{n,p}$. Combining these two bounds appropriately will then yield the result:
\begin{align}
\label{eq:B3}
\left\|\widetilde{\beta}_n-\beta_0 \right\|_{n,p} = \log^2(n) O_{\mathbb{P}}\left(n^{-\frac{m}{2m+d}} + \max_{1 \leq j \leq p} \diam^{\kappa}(A_j)  \right) \quad \text{and} \quad J_{m}\left(\widetilde{\beta}_n \right) = O_{\mathbb{P}}(1),
\end{align}
for $\lambda$ satisfying \ref{ip} and \ref{iip}. The definition of the convex combination $\widetilde{\beta}_n$ will then allow us to derive the same rate of convergence for $\|\widehat{\beta}_n-\beta_0 \|_{n,p}$ as well as the boundedness of $J_m(\widehat{\beta}_n)$.

\mypar{Step 1}. We begin by deriving a lower bound on the LHS of \eqref{eq:B2}. First, notice that, since $\widehat{\sigma}_n \xrightarrow{\mathbb{P}} \sigma_0$ by \ref{B5}, we may assume that $|\widehat{\sigma}_n - \sigma_0| \leq \varepsilon$ for all $\varepsilon>0$ and large $n$. Choose $\varepsilon$ small enough so that \ref{B6} is satisfied and $\sigma_0-\varepsilon>0$. Since $\widehat\sigma_n \in [\sigma_0-\varepsilon, \sigma_0+\varepsilon]$, we have
\begin{align*}
M(\widetilde{\beta}_n, \widehat{\sigma}_n) - M(\beta_0, \widehat{\sigma}_n) \geq \inf_{|\sigma-\sigma_0|\leq \varepsilon} \left[M(\widetilde{\beta}_n, \sigma) - M(\beta_0, \sigma)\right].
\end{align*}
By \ref{B6} and after interchanging expectation and differentiation, which is permitted by the boundedness of $\psi$, we see that
\begin{align*}
    \mathbb{E}_{\epsilon_i}\left\{\rho\left(\frac{\epsilon_i+d_i+\langle X_i, \beta_0-\beta\rangle_p}{\sigma}\right) \right\} & =  \mathbb{E}_{\epsilon_i}\left\{\rho\left(\frac{\epsilon_i+\langle X_i, \beta_0-\beta\rangle_p}{\sigma}\right) \right\} + \frac{d_i}{\sigma} g_{\sigma}\left(\frac{\langle X_i, \beta_0-\beta\rangle_p}{\sigma} \right) + O(d_i^2)
    \\ & =  \mathbb{E}_{\epsilon_i}\left\{\rho\left(\frac{\epsilon_i+\langle X_i, \beta_0-\beta\rangle_p}{\sigma}\right) \right\} + O\left( \left| d_i\langle X_i, \beta_0-\beta\rangle_p \right| \right)  + O(d_i^2),
\end{align*}
where to obtain the second equality we have used the fact that $g_{\sigma}(0) = \mathbb{E}\{\psi(\epsilon_1/\sigma) \} = 0$ for every $\sigma \in [\sigma_0-\epsilon,\sigma_0+\epsilon]$. Similarly,
\begin{align*}
\mathbb{E}_{\epsilon_i}\left\{\rho\left(\frac{\epsilon_i+d_i}{\sigma}\right)\right\} = \mathbb{E}_{\epsilon_i}\left\{\rho\left(\frac{\epsilon_i}{\sigma}\right)\right\} + O(d_i^2). 
\end{align*}
Combining these two expansions and averaging, we find
\begin{align*}
 M(\beta, \sigma) - M(\beta_0, \sigma)   & \geq  \frac{1}{n} \sum_{i=1}^n  \mathbb{E}_{\epsilon_i}\left\{\rho\left(\frac{\epsilon_i+\langle X_i, \beta_0-\beta\rangle_p}{\sigma} \right) -  \rho\left(\frac{\epsilon_i}{\sigma}\right) \right\} - \frac{C}{n} \sum_{i=1}^n d_i^2 - 
 \\ & \quad  - C \left(\frac{1}{n} \sum_{i=1}^n d_i^2 \right)^{1/2} \left\|\beta-\beta_0\right\|_{n,p},
\end{align*}
by the Cauchy-Schwarz inequality, for some large $C>0$. Reasoning as in \eqref{eq:convexity} in the proof of Theorem~\ref{thm:1}, we have 
\begin{equation}    \label{eq:convexityB}
\begin{aligned}
\frac{1}{n} \sum_{i=1}^n  \mathbb{E}_{\epsilon_i}\left\{\rho\left(\frac{\epsilon_i+\langle X_i, \beta_0-\beta\rangle_p}{\sigma} \right) -  \rho\left(\frac{\epsilon_i}{\sigma}\right) \right\} & \geq \varepsilon \frac{\left\|\beta-\beta_0\right\|^2_{n,p}}{D_{\varepsilon}^2 \log(n) (1+J_m(\beta))^2},
\end{aligned}
\end{equation}
for every $\beta \in \mathcal{H}^{m}(\mathbbm{R}^d)$, with high probability. The RHS of \eqref{eq:convexityB} is uniform in $\sigma \in [\sigma_0-\varepsilon, \sigma_0+\varepsilon]$. Furthermore, it follows as in the proof of Theorem~\ref{thm:1} that this inequality also holds for $\widetilde{\beta}_n$. Therefore,
\begin{align}
\label{eq:lower}
M(\widetilde{\beta}_n, \widehat{\sigma}_n) - M(\beta_0, \widehat{\sigma}_n)  & \geq \inf_{|\sigma-\sigma_0|\leq \varepsilon} \left[M(\widetilde{\beta}_n, \sigma) - M(\beta_0, \sigma) \right]  \nonumber
\\ &  \geq\frac{ \varepsilon \left\|\widetilde{\beta}_n-\beta_0\right\|^2_{n,p}}{D_{\varepsilon}^2\log(n)\left(1+J_m(\widetilde{\beta}_n)\right)^2} - \frac{C}{n} \sum_{i=1}^n d_i^2 \nonumber
\\ & \quad - C \left(\frac{1}{n} \sum_{i=1}^n d_i^2 \right)^{1/2} \left\|\widetilde{\beta}_n-\beta_0\right\|_{n,p},
\end{align}
which provides the desired lower bound for \eqref{eq:B2}.

\mypar{Step 2}. To derive an upper bound for \eqref{eq:B2}, use the fundamental theorem of calculus to write
\begin{align*}
    \rho\left(\frac{\epsilon_i + d_i + \langle X_i, \beta_0-\widetilde\beta_n\rangle_p}{\widehat{\sigma}_n}  \right) - \rho\left(\frac{\epsilon_i+d_i}{\widehat\sigma_n} \right) & = \frac{1}{\widehat{\sigma}_n}\int_{d_i}^{d_i + \langle X_i, \beta_0-\widetilde{\beta}_n\rangle_p} \left\{ \psi \left(\frac{\epsilon_i+t}{\widehat\sigma_n} \right) - \psi\left( \frac{\epsilon_i}{\widehat{\sigma}_n} \right) \right\} \dd t
    \\ & \quad + \frac{1}{\widehat{\sigma}_n} \psi\left( \frac{\epsilon_i}{\widehat{\sigma}_n}\right) \langle X_i, \beta_0-\widetilde{\beta}_n\rangle_p,
\end{align*}
so that, in the notation of Lemma~\ref{lem:4},
\begin{align*}
    \widehat{\sigma}_n\left|\left(M_n(\beta_0, \widehat{\sigma}_n) - M(\beta_0, \widehat{\sigma}_n) \right) - \left( M_n(\widetilde{\beta}_n, \widehat{\sigma}_n) - M(\widetilde{\beta}_n, \widehat{\sigma}_n) \right)\right| \leq \left|\frac{1}{n} \sum_{i=1}^n V_{i,\widetilde \beta_n, \widehat\sigma_n} \right| + \left| \frac{1}{n} \sum_{i=1}^n Z_{i,\widetilde \beta_n, \widehat\sigma_n}\right|.
    \end{align*}
An application of Lemma~\ref{lem:4} now yields
\begin{align}
\label{eq:upper}
\nonumber
    \left(M_n(\beta_0, \widehat{\sigma}_n) - M(\beta_0, \widehat{\sigma}_n) \right) &- \left( M_n(\widetilde{\beta}_n, \widehat{\sigma}_n)  - M(\widetilde{\beta}_n, \widehat{\sigma}_n) \right) 
    \\ & = O_{\mathbb{P}}\left(\sqrt{\frac{\log(n)}{n} }\left\|\widetilde\beta_n-\beta_0\right\|_{n,p}^{1-\frac{d}{2m}}\left\{1+J_m(\widetilde{\beta}_n)\right\}^{\frac{d}{2m}} \right),
\end{align}
as $\widehat\sigma_n \in [\sigma_0-\varepsilon, \sigma_0+\varepsilon]$ for all large $n$. Thus, \eqref{eq:upper} provides the desired upper bound for \eqref{eq:B2}.

\mypar{Step 3}. Combining \eqref{eq:B2} with \eqref{eq:lower} and \eqref{eq:upper}, we find
\begin{align*}
 \frac{\varepsilon\left\|\widetilde{\beta}_n-\beta_0\right\|^2_{n,p}}{D_{\varepsilon}^2 \log(n) \left\{1+J_m\left(\widetilde{\beta}_n\right)\right\}^2} + \lambda J_m^2\left(\widetilde{\beta}_n\right) & \leq  O_\mathbb{P}\left( \sqrt{\frac{\log(n)}{n}}\right) \|\widetilde{\beta}_n-\beta_0\|_{n,p}^{1-\frac{d}{2m}} \{1+J_m(\widetilde{\beta}_n) \}^{\frac{d}{2m}} + \lambda J_m^2\left(\beta_0 \right) \nonumber
\\ & \quad  + O_\mathbb{P} \left( \max_{1 \leq j \leq p} \diam^{2 \kappa}(A_j)\right) \nonumber
\\ & \quad  + O_\mathbb{P} \left( \max_{1 \leq j \leq p} \diam^{\kappa}(A_j)\right) \left\|\widetilde{\beta}_n-\beta_0\right\|_{n,p},
\end{align*}
the only difference with the corresponding step in the proof of Theorem~\ref{thm:1} being the presence of an additional $\sqrt{\log(n)}$-term on the RHS. Reasoning along the lines of that proof and making use of \ref{ip} and \ref{iip}, \eqref{eq:B3} can be easily verified.

\mypar{Step 4}. The last step in our proof involves establishing the same rate of convergence for $\|\widehat{\beta}_n-\beta_0\|_{n,p}$ as $\|\widetilde{\beta}_n-\beta_0\|_{n,p}$  and using $J_m(\widetilde\beta_n) = O_{\mathbb{P}}(1)$ in order to deduce $J_m(\widehat\beta_n) = O_{\mathbb{P}}(1)$. Both of these facts follow from identical arguments as in the proof of Theorem~\ref{thm:1} and a detailed proof is thus omitted. This concludes the proof of Theorem~\ref{thm:2}.

\subsection{Proof of Theorem~\ref{thm:3}}

To prove that $\widehat\sigma_n \xrightarrow{\mathbb{P}} \sigma_0$  it suffices to prove that $\mathbb{P}(\widehat\sigma_n \geq \sigma_0+\varepsilon) \to 0$ and  
$\mathbb{P}(\widehat\sigma_n\leq\sigma_0-\varepsilon) \to 0 $ for every $\varepsilon>0$. Consider first the event $\{\widehat\sigma_n \geq \sigma_0 + \varepsilon\}$ and recall that $\widehat\sigma_n$ is the solution of
\begin{align}
\label{eq:mscale1}
    \frac{1}{n}\sum_{i=1}^n \chi \left( \frac{\epsilon_i+d_i+(\alpha_0-\widehat\alpha_n)+ \sum_{j=1}^p X_i(\mathbf{t}_j)\left(\beta_0(\mathbf{t}_j)-\widehat\beta_n(\mathbf{t}_j)\right) \mu(A_j)  }{\widehat\sigma_n} \right) = \frac{1}{2},
\end{align}
where $\{d_i\}_{i=1}^n$ are the discretization errors introduced in \eqref{eq:di} and $(\widehat{\alpha}_n,\widehat\beta_n)$ are the thin-plate estimates. By the monotonicity of $\chi$, on the set  $\{\widehat\sigma_n \geq \sigma_0+\varepsilon\}$ we have
\begin{align}
\label{eq:mscale2}
    \frac{1}{n}\sum_{i=1}^n \chi \left( \frac{\epsilon_i+d_i+(\alpha_0-\widehat\alpha_n)+ \sum_{j=1}^p X_i(\mathbf{t}_j)\left(\beta_0(\mathbf{t}_j)-\widehat\beta_n(\mathbf{t}_j)\right) \mu(A_j)  }{\sigma_0+\varepsilon} \right) > \frac{1}{2}.
\end{align}
Now, by a first order Taylor expansion, the LHS of \eqref{eq:mscale2} may be rewritten as
\begin{align*}
     \frac{1}{n}\sum_{i=1}^n \chi \left( \frac{\epsilon_i }{\sigma_0+\varepsilon} \right) + \underbrace{\frac{1}{n(\sigma_0+\varepsilon)} \sum_{i=1}^n \chi^{\prime} \left(\epsilon_i^{\star}\right)\left(d_i + \left(\alpha_0-\widehat\alpha_n\right) + \sum_{j=1}^p X_i(\mathbf{t}_j)\left(\beta_0(\mathbf{t}_j)-\widehat\beta_n(\mathbf{t}_j)\right) \mu(A_j) \right)}_{  S_n}
\end{align*}
for some mean values $\left\{\epsilon_i^{\star}\right\}_{i=1}^n$. But, $\chi^{\prime}$ is bounded and, as we show in the proof of Theorem~\ref{thm:1} (see \eqref{eq: di} there), $n^{-1} \sum_{i=1}^n |d_i|^2 = o_{\mathbb{P}}(1)$, $|\widehat\alpha_n - \alpha_0| = o_{\mathbb{P}}(1)$ and $\|\widehat\beta_n-\beta_0 \|_{n,p} = o_{\mathbb{P}}(1)$. Combining these facts along with the triangle and Cauchy-Schwarz inequalities we now see that
\begin{align*}
   \left|S_n\right|  &  = \left| \frac{1}{n(\sigma_0+\varepsilon)} \sum_{i=1}^n \chi^{\prime} \left(\epsilon_i^{\star}\right)\left(d_i + \left(\alpha_0-\widehat\alpha_n \right)+ \sum_{j=1}^p X_i(\mathbf{t}_j)\left(\beta_0(\mathbf{t}_j)-\widehat\beta_n(\mathbf{t}_j)\right) \mu(A_j) \right)\right| 
    \\ & \quad \leq \frac{\left\| \chi^{\prime}\right\|_{\infty}}{\sigma_0+\varepsilon} \left(\frac{1}{n} \sum_{i=1}^n \left|d_i\right| + \left|\widehat\alpha_n-\alpha_0 \right| + \frac{1}{n} \sum_{i=1}^n \left|\langle X_i, \beta_0- \widehat \beta_n \rangle_p\right| \right)
    \\ & \quad \leq  \frac{\left\| \chi^{\prime}\right\|_{\infty}}{\sigma_0+\varepsilon} \left( \left\{\frac{1}{n}\sum_{i=1}^n \left|d_i\right|^2 \right\}^{1/2} + \left|\widehat\alpha_n-\alpha_0 \right| + \left\|\widehat\beta_n - \beta_0 \right\|_{n,p} \right)
    \\ & \quad  = o_\mathbb{P}(1).
\end{align*}
Therefore, on the set $\{\widehat\sigma_n \geq \sigma_0 +\varepsilon\}$ we have
\begin{align}
    \label{eq:mscale3}
   \frac{1}{n} \sum_{i=1}^n \chi \left(\frac{\epsilon_i} {\sigma_0+\varepsilon}\right) + S_n > \frac{1}{2}.
\end{align}
But the summands $\{\chi(\epsilon_i/(\sigma_0+\varepsilon))\}_{i=1}^n$ are i.i.d. and bounded so that, by the WLLN, $n^{-1} \sum_{i=1}^n \chi(\epsilon_i/(\sigma_0+\varepsilon)) \xrightarrow{\mathbb{P}} \mathbb{E}\{\chi(\epsilon_1/(\sigma_0+\varepsilon)) \}$. By definition of $\sigma_0$ % in~\eqref{eq:sigma0} 
as well as the strict monotonicity of $\chi$, $\mathbb{E}\{\chi(\epsilon_1/(\sigma_0+\epsilon)) \} < 1/2$. Since $S_n = o_{\mathbb{P}}(1)$, it  now follows that 
\begin{align*}
    \lim_{n \to \infty}\mathbb{P} \left( \widehat{\sigma}_n \geq \sigma_0 + \varepsilon \right) \leq 
    \lim_{n \to \infty}\mathbb{P} \left( \frac{1}{n} \sum_{i=1}^n \chi \left(\frac{\epsilon_i} {\sigma_0+\varepsilon}\right) + S_n > \frac{1}{2} \right) = 0.
\end{align*}
A similar argument shows that $\mathbb{P}(\widehat\sigma_n\leq\sigma_0-\varepsilon) \to 0 $. Combining these two limits establishes the result of the theorem.

\bibliographystyle{apalike}
\bibliography{papers}
\end{document}